\documentclass[ejs,authoryear,linksfromyear]{imsart}
\RequirePackage{amsmath,amssymb,amsthm, amsfonts, bm,bbm, relsize}
\RequirePackage{graphicx,fancyhdr,epstopdf}
\RequirePackage{enumitem,mathrsfs,multirow}
\RequirePackage{natbib}
\RequirePackage[ruled,vlined]{algorithm2e}
\RequirePackage[OT1]{fontenc}
\RequirePackage[colorlinks,citecolor=blue,urlcolor=blue]{hyperref}
\usepackage{multirow}
\usepackage{adjustbox}

\SetKwInOut{Input}{Input}
\SetKwInOut{Output}{Output}

\doi{10.1214/154957804100000000}
\pubyear{0000}
\volume{0}
\firstpage{1}
\lastpage{1}
\startlocaldefs
\DeclareMathOperator*{\argmin}{arg\,min} %argmin

\newcommand{\sumin}{\sum_{i=1}^n} % sum from i = 1 to n
\newcommand{\CONV}[1]{\stackrel{\text{#1}}{\longrightarrow}} % convergence mode
\newcommand{\bnw}{\widehat{\bm{\beta}}_n^w} % weighted lasso estimator
\newcommand{\bnwa}{\widehat{\bm{\beta}}_{n(1)}^w} 
\newcommand{\bSC}{\widehat{\bm{\beta}}_n^{\text{SC}}} % strongly consistent estimator
\newcommand{\bLS}{\widehat{\bm{\beta}}_n^{\text{OLS}}} % LSE 
\newcommand{\bMLE}{\widehat{\bm{\beta}}_n^{\text{MLE}}} % MLE
\newcommand{\bLAS}{\widehat{\bm{\beta}}_n^{\text{LAS}}} % LASSO
\newcommand{\eLS}{\bm{e}_n^{\text{OLS}}} % LSE residual vector
\newcommand{\eLAS}{\bm{e}_n^{\text{LAS}}} % Lasso residual vector

\newcommand{\cnwa}{C_{n(11)}^w}
\newcommand{\cnwb}{C_{n(12)}^w}
\newcommand{\cnwc}{C_{n(21)}^w}
\newcommand{\cnwd}{C_{n(22)}^w}
\newcommand{\znwa}{\bm{Z}_{n(1)}^w}
\newcommand{\znwb}{\bm{Z}_{n(2)}^w}
\newcommand{\znwc}{\bm{Z}_{n(3)}^w}
\newcommand{\huna}{\widehat{\bm{u}}_{n(1)}}
\newcommand{\hunb}{\widehat{\bm{u}}_{n(2)}}
\newcommand{\cnwas}{C_{n(11*)}^w}
\newcommand{\cnwbs}{C_{n(12*)}^w}
\newcommand{\cnwcs}{C_{n(21*)}^w}
\newcommand{\cnwds}{C_{n(22*)}^w}
\newcommand{\znwas}{\bm{Z}_{n(1*)}^w}
\newcommand{\znwbs}{\bm{Z}_{n(2*)}^w}
\newcommand{\hunas}{\widehat{\bm{u}}_{n(1*)}}
\newcommand{\hunbs}{\widehat{\bm{u}}_{n(2*)}}

\numberwithin{equation}{section}
\theoremstyle{plain}
\newtheorem{thm}{Theorem}[section]

\newtheorem{lem}{Lemma}[section] 
\newtheorem{proposition}{Proposition}[section]
\newtheorem{remark}{Remark}[section]
\endlocaldefs

\begin{document}
	
	\begin{frontmatter}
		
		\title{Random-weighting in LASSO regression}
		\runtitle{Random-weighting}
		
		\begin{aug}
			\author{\fnms{Tun Lee} \snm{Ng}\ead[label=e1]{tunlee@stat.wisc.edu}}
			\and
			\author{\fnms{Michael A.} \snm{Newton}\ead[label=e2]{newton@stat.wisc.edu}}
			
			\address{Department of Statistics \\
				1300 University Ave, Madison WI 53706 \\
				\printead{e1,e2}}
			
			\runauthor{T.L.Ng and M.A.Newton}
			
		\end{aug}
		
		\begin{abstract}
			We  establish statistical properties of random-weighting methods in LASSO regression under different regularization parameters $\lambda_n$ and suitable regularity conditions. The random-weighting methods in view concern  repeated optimization of a randomized objective function,
			motivated by the need for computationally efficient uncertainty quantification in contemporary estimation settings.  In the context of LASSO regression, we repeatedly assign  analyst-drawn random weights to terms in the objective function, and optimize to obtain a sample of random-weighting estimators. We show that existing approaches  have conditional model selection consistency and conditional asymptotic normality at different growth rates of $\lambda_n$ as $n \to \infty$. We  propose an extension to the available random-weighting methods and establish that the resulting samples attain conditional sparse normality and conditional consistency in a growing-dimension setting.  We illustrate the proposed methodology using synthetic
			and benchmark data sets, and we discuss the relationship of the results to approximate nonparametric Bayesian analysis and to perturbation bootstrap methods.
		\end{abstract}
		
		
		
		\tableofcontents
		
	\end{frontmatter}

\section{Introduction}

Consider the well-studied linear regression model with fixed design 
\begin{align} 
\label{eq.LinearModel}
\bm{Y} = \beta_{\mu} \bm{1}_n + X \bm{\beta} + \bm{\epsilon},
\end{align}
where $\bm{Y} = (y_1, \ldots, y_n)' \in \mathbb{R}^n$ is the response vector, $\bm{1}_n$ is a $n \times 1$ vector of ones, $X \in \mathbb{R}^{n \times p_n}$ is the design matrix, $\bm{\beta}$ is the vector of regression coefficients, and $\bm{\epsilon} = (\epsilon_1, \ldots, \epsilon_n)'$ is the vector of independent and identically distributed (i.i.d.) random errors with mean 0 and variance $\sigma^2_{\epsilon}$. Without loss of generality, we assume that the columns of $X$ are centered, and take $\widehat{\beta}_{\mu} = \bar{Y}$, in which case we can replace $\bm{Y}$ in (\ref{eq.LinearModel}) with $\bm{Y} - \bar{Y} \bm{1}_n$, and concentrate on inference for $\bm{\beta}$. Again, without loss of generality, we also assume $\bar{Y} =0$. Let $\bm{\beta_0} \in \mathbb{R}^{p_n}$ be the true model coefficients with $q$ non-zero components, where $q \leq \min(p_n, n)$. Note that $\bm{Y}, X$ and $\bm{\epsilon}$ are all indexed by sample size $n$, but we omit the subscript whenever this does not cause confusion. 

Recall, the LASSO estimator is given by 
\begin{align} \label{eq.LassoObj}
\bLAS
:= \argmin_{\bm{\beta}} 
\sumin ( y_i - \bm{x}_i' \bm{\beta} )^2 
+ \lambda_n \sum_{j=1}^{p_n} |\beta_j|,
\end{align}   
for a scalar penalty $\lambda_n$ \citep{Lasso}, where $\bm{x}_i'$ is the $i^{th}$ row of $X$. From a Bayesian perspective, this objective function corresponds to the negative log posterior density from a Gaussian likelihood and a double Exponential (Laplace) prior, which may be represented with a scale mixture of normals \citep{ScaleMixtureNormal}, and so the solution to (\ref{eq.LassoObj}) is also 
the maximum a posteriori (MAP) estimator  in a certain Bayesian model. Full
posterior analysis in this model is possible using the Gibbs sampler \citep{BayesianLasso}, though, in regression and related models, persistent questions
of Monte Carlo convergence may complicate the interpretation of Gibbs sampler output, especially in high dimensions \citep[e.g.,][]{WellingTeh,rajaratnam2015mcmc,robert2018accelerating,qin2019convergence}.

The penalized regression model is a canonical example
in the broad class of penalized inference procedures, and \citet{WBB} considered the random-weighting approach on a class of penalized likelihood objective functions to obtain approximate posterior samples. They saw good performance 
in high-dimensional regression, trend-filtering and deep learning applications. In particular, their random-weighting version of~(\ref{eq.LassoObj}) is
\begin{align} \label{eq.Bnw.setup}
\bnw := \argmin_{\bm{\beta}}
\left\{
\sumin W_i ( y_i - \bm{x}_i' \bm{\beta} )^2 
+ \lambda_n \sum_{j=1}^{p_n} W_{0,j} |\beta_j|
\right\}, 
\end{align}      
where the analyst first chooses a distribution $F_W$ with $P(W>0) = 1$ and $\mathbb{E}(W^4) < \infty$, and constructs $W_i \stackrel{iid}{\sim} F_W$ for all $i = 1, 2, \cdots, n$. 
The precise treatment of penalty-associated weights $\bm{W}_0=(W_{0,1}, \cdots, W_{0,p_n})$ induces several random-weighting variations, the simplest of which has 
\begin{align} \label{eq.NoWeightPenalty}
W_{0,j} = 1 \,\, \forall \,\, j,
\end{align}
or the penalty terms all share a common random weight
\begin{align} \label{eq.Bnw.WeightPenalty}
W_{0,j} = W_0 \,\, \forall \,\, j,
\text{ where } (W_0, W_i )\stackrel{iid}{\sim} F_W
\,\, \forall \,\, i,
\end{align}
and the most elaborate of which has all entries
\begin{align} \label{eq.Bnw.WeightPenalty2}
(W_{0,j}, W_i) \stackrel{iid}{\sim} F_W \,\, \forall \,\, i,j.
\end{align}

These random-weighting algorithms (as laid out in Algorithm \ref{alg:ALG_general}) produce independent samples and are trivially parallelizable over $b= 1, \ldots, B$. \citet{WBB} compared them to MCMC-based computations via  the Bayesian LASSO \citep{BayesianLasso}, and demonstrated  good numerical properties in terms of estimation error, prediction error, credible set construction, and  agreement with the Bayesian LASSO posterior. 

In the present work we investigate  asymptotic properties of~(\ref{eq.Bnw.setup}), with attention on properties of
the conditional distribution given data.  By allowing different rates of growth of the regularization parameter $\lambda_n$, and under suitable regularity conditions, we prove that the random-weighting method has the following properties:
\begin{itemize}
	\item conditional model selection consistency (for both growing $p_n$ and fixed $p$)
	\item conditional consistency (for fixed $p_n = p$)
	\item conditional asymptotic normality (for fixed $p_n = p$)
\end{itemize}
for all three weighting schemes (\ref{eq.NoWeightPenalty}), (\ref{eq.Bnw.WeightPenalty}) and (\ref{eq.Bnw.WeightPenalty2}). We find there is no common $\lambda_n$ that would allow random-weighting samples to have conditional sparse normality (i.e., simultaneously to enjoy conditional model selection consistency and to achieve conditional asymptotic normality on the true support of $\bm{\beta}$) even under fixed $p_n = p$ setting. Consequently, we propose an extension to the random-weighting framework (\ref{eq.Bnw.setup}) by adopting a two-step procedure in the optimization step as laid out in Algorithm \ref{alg:ALG_general_2step}. We  prove that a common
regularization rate $\lambda_n$ allows random-weighting samples to achieve conditional sparse normality and conditional consistency properties under growing $p_n$ setting.   

To begin, we present a brief literature review to elucidate how  random-weighting algorithms arise from two different statistical motivations, and how our work complements the existing literature. 

%\subsection{Weighted bootstrap from Bayesian perspective}
\subsection{Random weighting from a Bayesian perspective}

The present paper 
%(and those of \citet{WBB}) 
began with a Bayesian perspective in mind. In fact, the random-weighting approach belongs to a class of weighted bootstrap algorithms which arose from the search for scalable, accurate posterior inference tools. An early example in this class is the weighted likelihood bootstrap (WLB), which was designed to yield approximate posterior samples in parametric models \citep{Newton&Raftery}. Compared to Markov Chain Monte Carlo (MCMC), for example, WLB provides computationally efficient approximate posterior samples in cases where likelihood optimization is relatively easy. Framing WLB in contemporary context, \citet{WBB} introduced the Weighted Bayesian Bootstrap (WBB) by extending the posterior approximation scheme to penalized likelihood objective functions which found useful applications in several aforementioned settings. 

Others have also recognized the utility of weighted bootstrap computations 
%for i.i.d. sampling 
beyond the realm of parametric posterior approximation. A critical perspective was provided by \citet{Bissiri2016JRSSB} with the concept of generalized Bayesian inference. Rather than constructing a fully specified probabilistic model for data, they used loss functions to connect information in the data to functionals of interest. \citet{Lyddon2019Biometrika} discovered a key connection between the generalized Bayesian posterior and WLB sampling, and constructed a
modified random-weighting method called the loss-likelihood bootstrap to leverage this connection. 

Further links to nonparametric Bayesian inference were recently reported in \citet{Lyddon2018NIPS} and \citet{Fong2019ICML}, who introduced Bayesian nonparametric learning (NPL). Their perspective concerns the parameter, denoted $\theta$ following conventional presentations, as residing in some parameter space $\Theta$, usually a nice subset of $p$-dimensional Euclidean space. Instead of adopting the typical model-based approach, which would treat $\theta$ as an index to probability distributions in the specified model, their focus was more nonparametric. Whether or not the model specification is valid, they identified the distribution within the parametric model that is closest to the generative distribution $F$ as a solution to an optimization problem
\begin{eqnarray}
\label{eq:basic}
\theta := \theta(F) := 
\argmin_{t \in \Theta} 
\int l(t, y) dF(y).
\end{eqnarray}
Here $y$ denotes a data point, which is distributed $F$, and $l(,)$ is a loss function specified by the analyst.  Denoting $p(y|\theta)$ as the density function in a working probability model, a natural loss function is $l(\theta,y) = -\log p(y|\theta)$. From the nonparametric perspective,  $\theta$ becomes a model-guided feature of $F$.

If we  place a Dirichlet prior on $F$ and have a random sample $(y_1, y_2, \cdots, y_n)$ of data points, then the posterior for $F$ is also Dirichlet process \citep[e.g.,][]{ferguson1973}. Operationally, posterior sampling of $\theta=\theta(F)$ is achieved by  sampling $F$ from its Dirichlet posterior and recomputing $\theta=\theta(F)$ each time -- i.e. by repeating the optimization in (\ref{eq:basic}). Using the stick-breaking construction  \citep[e.g.,][]{Sethuraman1994,Ishwaran2002}, \citet{Fong2019ICML} show that this sampling
is achieved approximately, with error vanishing
 as for $n \to \infty$, by repeatedly optimizing
\begin{eqnarray} \label{eq:BayesianNPL}
\argmin_{t \in \Theta}
\left\{
\sumin
W_i l(t, y_i) 
\right\}  
\end{eqnarray}
for random weights $(W_1, \cdots, W_n) \stackrel{iid}{\sim}  Exp(1)$.
Their Bayesian NPL approach could be extended to include regularization
\begin{eqnarray} \label{eq:LossNPL}
\argmin_{t \in \Theta}
\left\{
\sumin
W_i l(t, y_i)
+ \gamma g(t) 
\right\}  
\end{eqnarray}  
for some regularization parameter $\gamma > 0$ and penalty function $g(\cdot)$, and thus the proposed
LASSO random-weighting~(\ref{eq.Bnw.setup}) has a Bayesian-NPL interpretation by taking $l(t,y) = \Vert y-t \Vert^2$ and $g(t) = \Vert t \Vert_1$.   

Whether we aim for approximate parametric Bayes, generalized Bayes, or model-guided nonparametric Bayes, it is important to understand the distributional properties of these random-weighting procedures.
Precise answers are difficult, even with simple loss functions \citep[e.g.,][]{hjort2005exact}, and so asymptotic methods are helpful to study the conditional distribution of $\theta(F)$ given data. Adopting a Dirichlet prior on $F$, \citet{Fong2019ICML} pointed out that WBB sampling is consistent  under suitable regularity conditions, due to posterior consistency property of the Dirichlet process (e.g., \cite{ghosal1999}, \cite{ghosal2000}). \citet{Newton&Raftery}'s first-order analysis  yields the same Gaussian limits as the standard Bernstein-von-Mises results \citep[e.g.,][]{vanderVaartbook} under a correctly-specified Bayesian parametric model. Under model misspecification setting, \citet{Lyddon2019Biometrika} showed that the Gaussian limits of random weighting do not coincide with their Bayesian counterparts in \citet{kleijn2012}. Instead, they  mimic the Gaussian limits in \citet{huber1967} -- the asymptotic covariance matrix  is the well-known sandwich covariance matrix from robust-statistics literature.

With the work reported here, we aim to extend asymptotic analysis for random-weighting methods to high-dimensional linear regression models. Our work adapts frequentist-theory asymptotic arguments, notably the works of \citet{Knight&Fu} and \citet{BinYu}, to the present context.     
 
\subsection{Connection to perturbation bootstrap}    

Whilst the random-weighting approach 
%is largely motivated 
may be motivated
from a Bayesian perspective, its resemblance to  existing bootstrap algorithms, especially the perturbation bootstrap, warrants a comparison between random-weighting and the non-Bayesian bootstrap literature. 
%Compared to random-weighting, bootstrap techniques arose from a different statistical motivation. 
The (naive) perturbation bootstrap was introduced by \citet{Jin2001} as a method to estimate sampling distributions of estimators related to $U$-process-structured objective functions. \citet{ChatterjeeBose2005}  established first-order distributional consistency of a generalized perturbation bootstrap technique 
%which includes Efron's classical bootstrap and (naive) perturbation bootstrap, 
in  M-estimation where they allowed both $n \to \infty$ and $p_n \to \infty$. That paper also pointed out that for broader classes of models, the generalized bootstrap method is not second-order accurate without appropriate bias-correction and studentization. In particular, the work in (naive) perturbation bootstrap resembles the Bayesian NPL objective function (\ref{eq:BayesianNPL}). Subsequently, \citet{Minnier2011} proved the first-order distributional consistency of the perturbation bootstrap for \citet{Zou2006}'s Adaptive LASSO (ALasso) and \citet{SCAD}'s smoothly clipped absolute deviation (SCAD) under fixed-$p$ setting in order to construct accurate confidence regions for ALasso and SCAD estimators. Again, their work has the flavor of Bayesian Loss-NPL (\ref{eq:LossNPL}) where the loss function is either ALasso or SCAD. More recently, \citet{DasGreg2019} extended the work of \citet{Minnier2011} by introducing a suitably Studentized version of modified perturbation bootstrap ALasso estimator that achieves second-order correctness in distributional consistency even when $p_n \to \infty$. 

Various bootstrap techniques have been considered to construct confidence regions for standard LASSO estimators in (\ref{eq.LassoObj}) under different model settings, including fixed or random design, as well as homoscedastic or heteroscedastic errors $\bm{\epsilon}$. \citet{Knight&Fu} first considered the residual bootstrap under fixed design and homoscedastic error. \citet{Chatterjee2010} presented a rigorous proof for the heuristic discussion of \citet{Knight&Fu}'s Section 4 to show that the LASSO residual bootstrap samples fail to be distributionally consistent unless $\bm{\beta}_0$ is not sparse, for which \citet{Knight&Fu} invoked the Skorokhod's argument. Subsequently, \citet{Chatterjee2011Jasa} rectified the shortcoming by proposing a modified residual bootstrap method by thresholding the Lasso estimator. Meanwhile, \citet{Camponovo2015} proposed a modified paired-bootstrap technique and established its distributional consistency to approximate the distribution of Lasso estimators in linear models with random design and heteroscedastic errors. Recently, \citet{Das2019} considered the perturbation bootstrap method for Lasso estimators under both fixed and random designs with heteroscedastic errors. Since centering on the thresholded Lasso estimator \citep[c.f.][]{Chatterjee2011Jasa} resulted in distributional inconsistency of the naive perturbation bootstrap, \citet{Das2019} proceeded with a suitably Studentized version of modified perturbation bootstrap (c.f. \citet{DasGreg2019}) to rectify the shortcoming.                  

Interestingly, the setup of naive perturbation bootstrap in \citet{Das2019} mimics the proposed random-weighting approach~(\ref{eq.Bnw.setup}) in LASSO regression with weighting scheme (\ref{eq.NoWeightPenalty}), but there remain some differences in our approach. 
%We are interested in investigating the statistical properties of the random-weighting method in LASSO regression as introduced by \citet{WBB}, whereas \citet{Das2019} focused on constructing a valid bootstrap approach for LASSO estimators. \textcolor{blue}{[]Tun, I'm not sure the specific issue with this sentence.]}  
\citet{Das2019} also considered heteroscedastic error term $\bm{\epsilon}$, which we do not consider in this paper. Meanwhile, the weighting schemes considered in this paper are slightly more flexible, since we also consider the cases where independent random weights are also assigned on the LASSO penalty term in weighting schemes (\ref{eq.Bnw.WeightPenalty}) and (\ref{eq.Bnw.WeightPenalty2}). The random weights in \citet{Das2019}'s perturbation bootstrap are restricted to independent draws from distribution with $\sigma^2_W = \mu_W^2$, whereas we consider any positive random weights with finite fourth moment. Furthermore, our extended random-weighting framework in Section \ref{sec:MainResults2} attains conditional sparse normality property under growing $p_n$ setting, whereas \citet{Das2019}'s (modified) perturbation bootstrap method achieves distributional consistency under fixed dimensional $(p_n = p)$ setting.       
      
%To the best of our knowledge, all the cited bootstrapping techniques for vanilla Lasso estimator in (\ref{eq.LassoObj}) were restricted to fixed-$p$ setting. Whilst our original motivation stems from a Bayesian perspective, our work also contributes to the bootstrap literature by extending the theory on conditional model selection consistency of the vanilla Lasso estimator in (\ref{eq.LassoObj}) under growing $p_n$ setting with fixed design and homoscedastic errors $\bm{\epsilon}$. In addition, we also present a simple first-order asymptotic result (conditional consistency and conditional asymptotic normality) under fixed-$p$ setting, to illustrate the fact that even with the very simple setup of random-weighting (\ref{eq.Bnw.setup}), this technique could achieve different asymptotic properties depending on our calibration of regularization parameter $\lambda_n$. This contrasts with the more advanced Lasso-related bootstrap techniques that enjoy more accurate distributional consistency property in growing $p_n$ settings, at the expense of more elaborate steps -- ALasso instead of vanilla Lasso, and suitably-Studentized modified perturbation bootstrap for ALasso instead of the naive perturbation bootstrap \citep{DasGreg2019}. 

We now outline the remaining sections of the paper. In Section \ref{sec:problem_setup}, we set the regularity assumptions, probability space and necessary notations used throughout, and then we report our main results in Section \ref{sec:Main}. Subsequently, in Section \ref{sec:discuss}, we argue that the random-weighting approach has meaningful approximate Bayesian inference and sampling theory interpretations. We then present extensive simulation studies in Section \ref{sec:numerical} to illustrate how the three random-weighting schemes (\ref{eq.NoWeightPenalty}), (\ref{eq.Bnw.WeightPenalty}) and (\ref{eq.Bnw.WeightPenalty2}) compare with other existing methods. Application to a housing-prices data set is also given. Finally, Appendix \ref{sec_Appendix} provides extensive details about the proofs for all lemmas, theorems and propositions.    

\section{Problem Setup} \label{sec:problem_setup}

We assume throughout that the unknown number  of truly relevant predictors, $q$, is fixed, that  
\begin{align} \label{assume_4thmoment}
\mathbb{E} (\epsilon_i^4) < \infty \,\, \forall \,\, i,
\end{align}
\noindent and all $p_n$ predictors are bounded, i.e. $\exists$ $M_1 > 0$ such that 
\begin{align} \label{assume_boundedX}
|x_{ij}| \leq M_1 \quad \forall \quad i = 1, \ldots, n\,\, ; \,\, j = 1, \ldots, p_n,  
\end{align}
\noindent where $x_{ij}$ refers to the $(i,j)^{th}$ element of $X$. 

Without loss of generality, we  partition $\bm{\beta}_0$ into 
\[
\bm{\beta}_0 = 
\begin{bmatrix}
\bm{\beta}_{0(1)} \\
\bm{\beta}_{0(2)}
\end{bmatrix}
,
\]
where $\bm{\beta}_{0(1)}$ refers to the $q \times 1$ vector of non-zero true regression parameters, and $\bm{\beta}_{0(2)}$ is a $(p_n-q) \times 1$ zero vector. Similarly, we  partition the columns of the design matrix $X$ into  
\[
X = 
\begin{bmatrix}
X_{(1)} & X_{(2)}
\end{bmatrix}
\]
which corresponds to $\bm{\beta}_{0(1)}$ and $\bm{\beta}_{0(2)}$ respectively. 

We consider both fixed-dimensional ($p_n = p$) and growing-dimensional ($p_n$ increases with $n$) settings. In the growing dimensional ($p_n$ increases with $n$) setting, we assume that for some $M_2 >0$, 
\begin{align} \label{assume_X'X_11}
\bm{\alpha}' 
\left[ \dfrac{X'_{(1)} X_{(1)}}{n} \right]
\bm{\alpha}
\geq M_2
\quad
\forall
\quad
\Vert \bm{\alpha} \Vert_2 = 1.
\end{align}
Note that assumptions (\ref{assume_boundedX}) and (\ref{assume_X'X_11}), coupled with the fact that $q$ is fixed, ensure that $\frac{1}{n} X'_{(1)} X_{(1)}$ is invertible $\forall$ $n$, a fact that we rely on in this paper.

Meanwhile, for fixed-dimensional ($p_n = p$) setting, we assume that $\text{rank}(X) = p$ and there exists a non-singular matrix $C$ such that
\begin{align} \label{assume_X'X}
\dfrac{1}{n} X'X = \dfrac{1}{n} \sumin \bm{x}_i \bm{x}_i' \to C 
\quad \text{as } n \to \infty,
\end{align} 
\noindent where $\bm{x}_i$ is the $i^{th}$ row of the design matrix $X$. \\ 

\noindent
\textbf{Comments on assumptions}: The fixed-$q$ assumption is commonly found in Bayesian linear-model literature, such as \cite{JohnsonRossell}, and \cite{NarisettyHe2014}. Since we intend to compare the random-weighting approach with posterior inference, we make the fixed-$q$ assumption to align with existing Bayesian theory. The 
finite-moment assumption (\ref{assume_4thmoment}) of $\bm{\epsilon}$ is commonly found in literature \citep[e.g.,][]{Camponovo2015,Das2019}  is  weaker than the normality assumption commonly specified under a Bayesian approach \citep[e.g.,][]{BayesianLasso,JohnsonRossell,NarisettyHe2014}. Assumption (\ref{assume_boundedX}) can also be found in some seminal papers, such as \citet{BinYu} and \citet{Chatterjee&Lahiri}, and in fact, can be (trivially) achieved by standardizing the covariates. Assumption (\ref{assume_X'X_11}) is equivalent to providing a lower bound to the minimum eigenvalue of $\frac{1}{n} X_{(1)}'X_{(1)}$. This eigenvalue assumption is very common in both frequentist and Bayesian literature, such as \cite{BinYu} and \cite{NarisettyHe2014}. Finally, assumption (\ref{assume_X'X}) is  common in the LASSO literature under fixed $p$ setting, which can be traced back to \cite{Knight&Fu} and \cite{BinYu}. This assumption basically explains the relationship between the predictors under a fixed design model, and can be interpreted as the direct counterpart to the variance-covariance matrix of $X$ under a random design model. For the case of growing $p_n$, assumption (\ref{assume_X'X}) is no longer appropriate since the dimension of $\frac{1}{n} X'X$ grows.  \\ 

\noindent
\textbf{Probability Space:} There are two sources of variation in the random-weighting setup (\ref{eq.Bnw.setup}), namely the error terms $\bm{\epsilon}$ and the user-defined weights $\bm{W}$. In this paper, we consider a common probability space with the common probability measure $P = P_D \times P_W$, where $P_D$ is the probability measure of the observed data $Y_1, Y_2, \cdots$, and $P_W$ is the probability measure of the triangular array of random weights \citep{mason1992rank}. The use of product measure reflects the independence of user-defined $\bm{W}$ and data-associated $\bm{\epsilon}$. We focus on the conditional probabilities given data, that is, given the sigma-field $\mathcal{F}_n$ generated by $\bm{\epsilon}$:
$$
\mathcal{F}_n := \sigma(Y_1, \ldots, Y_n) = \sigma(\epsilon_1, \ldots, \epsilon_n). 
$$
The study of convergence of these conditional probabilities $P( \, \cdot \, |\mathcal{F}_n)$ under a weighted bootstrap framework is not new; see, for example, \citet{mason1992rank} 
and \citet{Lyddon2019Biometrika}. We now outline some definitions and notations in this respect.   
\\

\noindent
\textbf{Conditional Convergence Notations:} Let random variables (or vectors) $U, V_1, V_2, \ldots $ be defined on $(\Omega, \mathcal{A})$. We say $V_n$ converges in conditional probability $a.s.$ $P_D$ to $U$ if for every $\delta > 0$,
$$
P( \Vert V_n - U \Vert > \delta | \mathcal{F}_n ) \to 0 \quad a.s. \,\, P_D
$$
as $n \to \infty$. The notation $a.s.$ $P_D$ is read as \textit{almost surely under} $P_D$, and means \textit{for almost every infinite sequence of data} $Y_1, Y_2, \cdots$. For brevity, this convergence is denoted
$$
V_n \CONV{c.p.} U \quad a.s. \,\, P_D.
$$ 

Similarly, we say $V_n$ converges in conditional distribution $a.s.$ $P_D$ to $U$ if for any Borel set $A \subset \mathbb{R}$, 
$$
P( V_n \in A | \mathcal{F}_n ) \to P(U \in A)  \quad a.s. \,\, P_D 
$$     
as $n \to \infty$. For brevity, this convergence is denoted
$$
V_n \CONV{c.d.} U \quad a.s. \,\, P_D.
$$ 

In addition, for random variables (or vectors) $V_1, V_2, \ldots $ and random variables $U_1, U_2, \ldots $, we say 
$$
V_n = O_p(U_n) \quad a.s. \,\, P_D
$$
if and only if , for any $\delta > 0$, there is a constant $C_\delta > 0$ such that $a.s. \, P_D$,
$$
\sup_n P \left(
	\Vert V_n \Vert \geq C_\delta |U_n| \,\, \Big| \mathcal{F}_n
\right)  < \delta;
$$
whereas
$$
V_n = o_p(U_n) \quad a.s. \,\, P_D
$$
if and only if
$$
\dfrac{V_n}{U_n} \CONV{c.p.} 0 \quad a.s. \,\, P_D.
$$

\noindent
\textbf{Other Notation:} Following the usual convention, denote $\Phi \{.\}$ as the cumulative distribution function of the standard normal distribution. For two random variables $U$ and $V$, the expression $U \perp V$ is read as ``$U$ is independent of $V$". Denote $\Vert \cdot \Vert_2$ and $\Vert \cdot \Vert_F$ as the $l_2$ norm and Frobenius norm respectively. Let $\bm{1}_k$ and $I_k$ be $k \times 1$ vector of ones and $k \times k$ identity matrix respectively for some integer $k \geq 2$. Besides that, for any two vectors $\bm{u}$ and $\bm{v}$ of the same dimension, we denote $\bm{u} \circ \bm{v}$ as the Hadamard (entry-wise) product of the two vectors. In addition, define 
\[
\begin{bmatrix}
C_{n(11)} & C_{n(12)} \\
C_{n(21)} & C_{n(22)} 
\end{bmatrix}
:= \dfrac{1}{n} X'X = \dfrac{1}{n}
\begin{bmatrix}
X'_{(1)} X_{(1)} & X'_{(1)} X_{(2)} \\
X'_{(2)} X_{(1)} & X'_{(2)} X_{(2)}
\end{bmatrix}
.
\]
Notice that an immediate consequence of Assumption (\ref{assume_X'X}) is that 
$$
C_{n(ij)} \to C_{ij} \,\, \forall \,\, i,j = 1,2,
$$
where $C_{11}$ is invertible. Furthermore, denote $\mu_W$ and $\sigma^2_W$ as the mean and variance of the random weight distribution $F_W$. Let $D_n = diag(W_1, \ldots, W_n)$, and define 
\[
\begin{bmatrix}
\cnwa & \cnwb \\
\cnwc & \cnwd 
\end{bmatrix}
:= \dfrac{1}{n} X' D_n X = \dfrac{1}{n}
\begin{bmatrix}
X'_{(1)} D_n X_{(1)} & X'_{(1)} D_n X_{(2)} \\
X'_{(2)} D_n X_{(1)} & X'_{(2)} D_n X_{(2)}
\end{bmatrix}
.
\]
Notice that $D_n$ does not contain any penalty weights $W_{0,j}$. For weighting scheme (\ref{eq.Bnw.WeightPenalty2}), the penalty weights $\bm{W}_0 = (W_{0,1}, \cdots, W_{0,p_n})$ could also be partitioned into
\[
\bm{W}_0 = 
\begin{bmatrix}
\bm{W}_{0(1)} \\
\bm{W}_{0(2)}
\end{bmatrix}
,
\]
which corresponds to the partition of $\bm{\beta}_0$. For ease of notation, define
\begin{align*}
\znwa &= \dfrac{1}{\sqrt{n}} X_{(1)}' D_n \bm{\epsilon}, \\
\znwb &= \dfrac{1}{\sqrt{n}} X_{(2)}' D_n \bm{\epsilon}, \\
\znwc &= C_{n(21)} C_{n(11)}^{-1} \znwa - \znwb, \\
\widetilde{C}^w_n &= \cnwc \left( \cnwa \right)^{-1} - C_{n(21)} C_{n(11)}^{-1}.
\end{align*}
Finally, the function $\text{sgn}(\cdot)$ maps positive entry to 1, negative entry to -1 and zero to zero. An estimator $\widehat{\bm{\beta}}$ is said to be equal in sign to the true parameter $\bm{\beta}_0$, if 
$$
\text{sgn}(\widehat{\bm{\beta}}) = \text{sgn}(\bm{\beta}_0),
$$
and is denoted as
$$
\widehat{\bm{\beta}} \stackrel{s}{=} \bm{\beta}_0.
$$

\section{Main Results} \label{sec:Main}

\subsection{One-step Procedure} \label{sec:MainResults}

%In this subsection, 
We investigate the asymptotic properties of random-weighting draws~(\ref{eq.Bnw.setup}) obtained from Algorithm \ref{alg:ALG_general}, 
%This is the original framework of random-weighting in LASSO regression 
which coincides
with the weighted Bayesian bootstrap method 
considered by \citet{WBB}. 
For convenience, we shall call this  the ``one-step procedure" to distinguish it from the extended framework that we shall discuss in Section \ref{sec:MainResults2}. 

\begin{algorithm}
	\SetAlgoLined
	\caption{Random-Weighting in LASSO regression}
	\label{alg:ALG_general}
	\Input{ 
		\begin{itemize}[itemsep=0pt]
			\item data: $D = (\bm{y}, X)$ 
			\item regularization parameter: $\lambda_n$
			\item number of draws: $B$
			\item choice of random weight distribution: $F_W$ 
			\item choice of weighting schemes: (\ref{eq.NoWeightPenalty}), (\ref{eq.Bnw.WeightPenalty}) or (\ref{eq.Bnw.WeightPenalty2})  
		\end{itemize}
	}
	\Output{ $B$ parameter samples $\{\widehat{\bm{\beta}}_n^{w,b} \}_{b=1}^B$ }
	\For {$b = 1$ to $B$} {
		Draw i.i.d. random weights from $F_W$ and substitute them into ~(\ref{eq.Bnw.setup}) \;
		Store $\widehat{\bm{\beta}}_n^{w,b}$ obtained by optimizing ~(\ref{eq.Bnw.setup}) \; 
	}   
\end{algorithm}

First, we establish the property of conditional model selection given data. In particular, we are interested in the conditional probability of the random-weighting samples matching the signs of $\bm{\beta}_0$. Notably, sign consistency is stronger than variable selection consistency, which requires only matching of zeros. Nevertheless, we agree with \citet{BinYu}'s argument of considering sign consistency -- it allows us to avoid situations where models have matching zeroes but reversed signs, which hardly qualify as correct models. We begin with a result that establishes the lower bound for this conditional probability.    
\begin{proposition} \label{lem_ModelSelect}
	Suppose $p_n \leq n$ and ${\rm rank}(X) = p_n$. Assume (\ref{assume_4thmoment}), (\ref{assume_boundedX}) and (\ref{assume_X'X_11}). Furthermore, assume the \textbf{strong irrepresentable condition} \citep{BinYu}: there exists a positive constant vector $\bm{\eta}$ such that 
	\begin{align} \label{assume_StrongIrrepresent}
	\left|
	C_{n(21)} 
	\left( C_{n(11)} \right)^{-1} 
	\text{sgn} \left( \bm{\beta}_{0(1)} \right)
	\right| \leq
	\bm{1}_{p_n-q} - \bm{\eta},
	\end{align}
	where $0 < \eta_j \leq 1$ $\forall$ $j = 1, \ldots, p_n-q$, and the inequality holds element-wise. Then, for all $n \geq p_n$,
	$$
	P\left(
	\bnw (\lambda_n) \stackrel{s}{=} \bm{\beta}_0
	\big| \mathcal{F}_n
	\right)	\geq 
	P \left( 
	A_n^w \cap B_n^w 
	\big| \mathcal{F}_n
	\right),
	%\,\, a.s. \,\, P_D, 
	$$ 
	where
	\begin{itemize}[wide]
		\item [(a)] for weighting scheme (\ref{eq.NoWeightPenalty}),
		\begin{align*}
		A_n^w &\equiv  
		\left\{
		\left\vert 
		\left( \cnwa \right)^{-1} 
		\left(
		\znwa - 
		\dfrac{\lambda_n}{2 \sqrt{n}} 
		\text{sgn} \left[ \bm{\beta}_{0(1)} \right]
		\right)
		\right\vert
		\leq  \sqrt{n}
		\left\vert \bm{\beta}_{0(1)} \right\vert
		\,\, \text{{\rm element-wise}}
		\right\} \\
		B_n^w &\equiv 
		\left\{
		\left\vert 
		\widetilde{C}^w_n 
		\left(
		\znwa - 
		\dfrac{\lambda_n}{2 \sqrt{n}} 
		\text{sgn} \left[ \bm{\beta}_{0(1)} \right]
		\right)
		+ \znwc 
		\right\vert
		\leq 
		\dfrac{\lambda_n}{2 \sqrt{n}} 
		\bm{\eta} 
		\,\, \text{{\rm element-wise}}
		\right\} ;
		\end{align*}
		\item [(b)] for weighting scheme (\ref{eq.Bnw.WeightPenalty}),
		\begin{align*}
		A_n^w &\equiv  
		\left\{
		\left\vert 
		\left( \cnwa \right)^{-1} 
		\left(
		\znwa - 
		\dfrac{\lambda_n W_0}{2 \sqrt{n}} 
		\text{sgn} \left[ \bm{\beta}_{0(1)} \right]
		\right)
		\right\vert
		\leq  \sqrt{n}
		\left\vert \bm{\beta}_{0(1)} \right\vert
		\,\, \text{{\rm element-wise}}
		\right\} \\
		B_n^w &\equiv 
		\left\{
		\left\vert 
		\widetilde{C}^w_n 
		\left(
		\znwa - 
		\dfrac{\lambda_n W_0}{2 \sqrt{n}} 
		\text{sgn} \left[ \bm{\beta}_{0(1)} \right]
		\right)
		+ \znwc 
		\right\vert
		\leq 
		\dfrac{\lambda_n W_0}{2 \sqrt{n}} 
		\bm{\eta} 
		\,\, \text{{\rm element-wise}}
		\right\}; 
		\end{align*}
		\item [(c)] for weighting scheme (\ref{eq.Bnw.WeightPenalty2}),
			\begin{alignat*}{2}
		A_n^w &\equiv  
		&&\bigg\{
		\left\vert 
		\left( \cnwa \right)^{-1} 
		\left(
		\znwa - 
		\dfrac{\lambda_n}{2 \sqrt{n}}
		\bm{W}_{0(1)} \circ
		\text{sgn} \left[ \bm{\beta}_{0(1)} \right]
		\right)
		\right\vert \\
		& &&\leq  \sqrt{n}
		\left\vert \bm{\beta}_{0(1)} \right\vert
		\,\, \text{{\rm element-wise}}
		\bigg\} \\
		B_n^w &\equiv 
		&&\bigg\{
		\left\vert 
		\widetilde{C}^w_n 
		\left(
		\znwa - 
		\dfrac{\lambda_n}{2 \sqrt{n}} 
		\bm{W}_{0(1)} \circ
		\text{sgn} \left[ \bm{\beta}_{0(1)} \right]
		\right)
		+ \znwc 
		\right\vert \\
		& &&\leq 
		\dfrac{\lambda_n}{2 \sqrt{n}} 
		\left(
		\bm{W}_{0(2)} - 
		\left\vert
		C_{n(21)} 
		\left( C_{n(11)} \right)^{-1}
		\bm{W}_{0(1)} \circ
		\text{sgn} \left[ \bm{\beta}_{0(1)} \right]
		\right\vert
		\right)
		\, \text{{\rm element-wise}}
		\bigg\}. 
		\end{alignat*}
	\end{itemize}
\end{proposition}

The ${\rm rank}(X) = p_n \leq n$ assumption in Proposition \ref{lem_ModelSelect} ensures that the random-weighting setup (\ref{eq.Bnw.setup}) has a unique solution \citep{Osborne2000}. For a random-design setting, the ${\rm rank}(X) = p_n \leq n$ assumption can be replaced with the assumption that $X$ is drawn from a joint continuous distribution \citep{LassoUnique}.    

The strong irrepresentable condition (\ref{assume_StrongIrrepresent}) can be seen as a constraint on the relationship between active covariates and inactive covariates, that is, the total amount of an irrelevant covariate ``represented" by a relevant covariate must be strictly less than one. Similar to \citet{BinYu}'s argument, $A_n^w$ refers to recovery of the signs of coefficients for $\bm{\beta}_{0(1)}$, and $B_n^w$ further implies obtaining $\widehat{\bm{\beta}}^w_{n(2)} = \bm{0}$ given $A_n^w$. The regularization parameter $\lambda_n$ continues to play the role of trade-off between $A_n^w$ and $B_n^w$: higher $\lambda_n$ leads to larger $B_n^w$ but smaller $A_n^w$, which forces the random-weighting method to drop more covariates, and vice versa. Meanwhile, larger $\bm{\eta}$ in (\ref{assume_StrongIrrepresent}), which could be interpreted as lower ``correlation" between active covariates and inactive covariates, increases $B_n^w$ but does not affect $A_n^w$, thus allowing the random-weighting method to better select the true model. \citet{BinYu} also gave a few sufficient conditions that ensure the following designs of $X$ satisfy condition (\ref{assume_StrongIrrepresent}):
\begin{itemize}
	\item constant positive correlation,
	\item bounded correlation,
	\item power-decay correlation, 
	\item orthogonal design, and
	\item block-wise design. 
\end{itemize} 
Again, we would like to highlight the fact that conditional on $\mathcal{F}_n$, the randomness of $A_n^w$ and $B_n^w$ derives from the random weights instead of $\bm{\epsilon}$. Besides that, notice how the presence of different penalty weights in weighting scheme (\ref{eq.Bnw.WeightPenalty2}) affects the strong irrepresentable condition (\ref{assume_StrongIrrepresent}) in $B^w_n$. We will see how these different weighting schemes affect the constraints on $p_n$ and $\lambda_n$ in order to achieve conditional model selection consistency.      

\begin{thm} \label{thm_Model_Select}
	\textbf{(Conditional Model Selection Consistency)} Assume assumptions in Proposition \ref{lem_ModelSelect}. 
	\begin{itemize}
		\item [(a)] Under weighting schemes (\ref{eq.NoWeightPenalty}) and (\ref{eq.Bnw.WeightPenalty}), if there exists $\frac{1}{2} < c_1 <  c_2 < 1.5 - c_1$ and $0 \leq c_3 < \min \{ 2(c_2 - c_1), 2c_1 - 1 \}$ for which $\lambda_n = \mathcal{O} \left( n^{c_2} \right)$ and $p_n = \mathcal{O} \left( n^{c_3} \right)$, then as $n \to \infty$,
		$$
		P\left(
		\bnw (\lambda_n) \stackrel{s}{=} \bm{\beta}_0
		\big| \mathcal{F}_n 
		\right)	
		\to 1
		\quad a.s. \,\, P_D. 
		$$   
		\item [(b)] Under weighting scheme (\ref{eq.Bnw.WeightPenalty2}), if $(W_i, W_{0,j}) \stackrel{iid}{\sim} \rm{Exp} (\theta_w)$ for some $\theta_w > 0$, and if $\bm{\eta} = \bm{1}_{p_n-q}$, and if there exists $\frac{1}{2} < c_1 <  c_2 < 1.5 - c_1$ and $0 \leq c_3 < \min \{ \frac{2}{3}(c_2 - c_1), 2c_1 - 1 \}$ for which $\lambda_n = \mathcal{O} \left( n^{c_2} \right)$ and $p_n = \mathcal{O} \left( n^{c_3} \right)$, then as $n \to \infty$,
		$$
		P\left(
		\bnw (\lambda_n) \stackrel{s}{=} \bm{\beta}_0
		\big| \mathcal{F}_n 
		\right)	
		\to 1
		\quad a.s. \,\, P_D. 
		$$
	\end{itemize}
\end{thm}

Theorem \ref{thm_Model_Select} could be interpreted as the ``concentration" of the conditional distribution of signs of $\bnw$ around the neighborhood of the true signs of $\bm{\beta}$ as $n \to \infty$. Comparing the three weighting schemes, we can see that assigning random weights on the penalty term further impedes how fast $p_n$ could increase with $n$ while achieving conditional model selection consistency, especially when the penalty terms do not share a common random weight in weighting scheme (\ref{eq.Bnw.WeightPenalty2}). This adversely affects/violates the strong irrepresentable assumption~(\ref{assume_StrongIrrepresent}), unless under a stringent condition where $\bm{\eta} = \bm{1}$. One sufficient condition for  $\bm{\eta} = \bm{1}$ would be zero correlation between any relevant predictor and any irrelevant predictor, i.e. $C_{n(21)} = \bm{0}$ for all $n$. 

We also point out that the conditional model selection consistency property under a fixed dimensional ($p_n = p$) setting could be easily obtained by taking $c_3 = 0$ in Theorem \ref{thm_Model_Select}. \\

The next two results concern with the properties of conditional consistency and conditional asymptotic normality of the random-weighting samples under a fixed-dimension ($p_n = p$) setting.   

\begin{thm} \label{thm_low_Consistency}
	Suppose $p_n = p$ is fixed. Assume (\ref{assume_4thmoment}), (\ref{assume_boundedX}) and (\ref{assume_X'X}). 
	\begin{itemize}
		\item [ (a) ] \textbf{(Conditional Consistency)} If $\dfrac{\lambda_n}{n} \to 0$, then for all three weighting schemes (\ref{eq.NoWeightPenalty}), (\ref{eq.Bnw.WeightPenalty}) and (\ref{eq.Bnw.WeightPenalty2}),
		$$
		\bnw \CONV{c.p.} \bm{\beta}_0 \quad a.s. \,\, P_D.
		$$
		\item [ (b) ] 	If $\dfrac{\lambda_n}{n} \to \lambda_0 \in (0,\infty)$, then 
		$$
		\left( \bnw - \bm{\beta}_0 \right) \CONV{c.d.} 
		\argmin_{\bm{u}} g(\bm{u}) \quad a.s. \,\, P_D,
		$$
		where
		$$
		g(\bm{u}) = 
		\mu_W \bm{u}' C \bm{u} + 
		\lambda_0 \sum_{j=1}^p W_j | \beta_{0,j} + u_j | 
		$$
		and
		\begin{itemize}
			\item [(i)] $W_j$ = 1 for all $j$ under weighting scheme (\ref{eq.NoWeightPenalty}),
			\item [(ii)] $W_j = W_0$ for all $j$ and $W_0 \sim F_W$ under weighting scheme (\ref{eq.Bnw.WeightPenalty}), 
			\item [(iii)] $W_j \stackrel{iid}{\sim} F_W$ under weighting scheme (\ref{eq.Bnw.WeightPenalty2}).
		\end{itemize}
	\end{itemize} 
\end{thm}
In other words, the conditional distribution of $\bnw$ concentrates
in the neighborhood of $\argmin_{\bm{u}} g(\bm{u})$ as the sample size increases.
In fact, for part (b)(i) of Theorem \ref{thm_low_Consistency}, conditional convergence in probability takes place since $g(\bm{u})$ is not a random function (i.e., does not involve any non-degenerate random variables).

\begin{thm} \label{thm_low_AsympDistn}
	\textbf{(Asymptotic Conditional Distribution)} Suppose $p_n = p$ is fixed. Assume (\ref{assume_4thmoment}), (\ref{assume_boundedX}) and (\ref{assume_X'X}). Let $\bSC$ be a strongly consistent estimator of $\bm{\beta}$ in the linear model (\ref{eq.LinearModel}) such that for $\bm{e}_n = \bm{Y} - X \bSC$,
	\begin{align} \label{assume:X'e}
	\dfrac{1}{\sqrt{n}} X' \bm{e}_n 
	\to \bm{0} 
	\quad a.s. \,\, P_D.
	\end{align}
	If $q = p$ and $\dfrac{\lambda_n}{\sqrt{n}} \to \lambda_0 \in [0,\infty)$, then
	$$
	\sqrt{n} \left( 
		\bnw - \bSC
	\right) 
	\CONV{c.d.} 
	\argmin_{\bm{u}} V (\bm{u})
	\quad a.s. \,\, P_D,
	$$
	where
	$$ 
	V(\bm{u}) = -2 \bm{u}' \Psi + \mu_W \bm{u}' C \bm{u} 
	+ \lambda_0 \sum_{j=1}^p W_j
	\left[ u_j \, \text{sgn}(\beta_{0,j}) \right],
	$$
	for $\Psi \sim N \left( \bm{0} , \sigma^2_W \sigma^2_{\epsilon} C \right)$, and
	\begin{itemize}
		\item [(i)] $W_j$ = 1 for all $j$ under weighting scheme (\ref{eq.NoWeightPenalty}),
		\item [(ii)] $W_j = W_0$ for all $j$, $W_0 \sim F_W$ and $W_0 \perp \Psi$ under weighting scheme (\ref{eq.Bnw.WeightPenalty}), 
		\item [(iii)] $W_j \stackrel{iid}{\sim} F_W$ and $W_j \perp \Psi$ for all $j$ under weighting scheme (\ref{eq.Bnw.WeightPenalty2}).
	\end{itemize} 
	In particular, if $\lambda_0 = 0$, then for all three weighting schemes (\ref{eq.NoWeightPenalty}), (\ref{eq.Bnw.WeightPenalty}) and (\ref{eq.Bnw.WeightPenalty2}),
	$$
	\sqrt{n} \left( 
	\bnw - \bSC
	\right) 
	\CONV{c.d.} 
	N \left( 
		\bm{0} \,\, , \,\, 
		\dfrac{ \sigma^2_W \sigma^2_{\epsilon} }{\mu_W^2} C^{-1} 
	\right)
	\quad a.s. \,\, P_D.
	$$	 
\end{thm}   

The OLS estimator $\bLS$ and the standard LASSO estimator $\bLAS (\lambda_n^*)$ with $\lambda_n^* = o(\sqrt{n})$ are two qualified candidates for $\bSC$ to satisfy the conditions in Theorem \ref{thm_low_AsympDistn}. (Note that $\lambda_n^*$ does not necessarily have to be the same as the $\lambda_n$ that we use for our random-weighting approach.) Firstly, due to Assumption (\ref{assume_X'X}), $\bLS$ is strongly consistent \citep{LSEstrong}, and
$$
X' \eLS 
= \left( X'Y - X'X (X'X)^{-1} X'Y \right)
= \bm{0}.
$$ 
Meanwhile, since $\mathbb{E}(|\epsilon_i|) < \infty$ for all $i$ and $\lambda_n^* = o(\sqrt{n})$, $\bLAS (\lambda_n^*)$ is strongly consistent \citep{Chatterjee&Lahiri}, and the KKT conditions ensure that
$$
\dfrac{1}{\sqrt{n}} \left\Vert X' \eLAS \right\Vert_2
=  \dfrac{1}{\sqrt{n}} \left\Vert X' \left( \bm{y} - X \bLAS \right)  \right\Vert_2
\leq \dfrac{\lambda_n^* \sqrt{p}}{\sqrt{n}} 
\to 0 \quad a.s. \,\, P_D. 
$$  
We also point out that centering on the true regression parameter 
$$
\sqrt{n} \left( \bnw - \bm{\beta}_0 \right). 
$$ 
results in additional terms that depend on the sample path of realized data $\{y_1, y_2, \cdots\}$. Consequently, convergence in conditional distribution almost surely under $P_D$ (just like the result in Theorem \ref{thm_low_AsympDistn}) could not be achieved. We refer readers to Remark \ref{rmk_centering} in the Appendix for more details.

On the other hand, a more sophisticated argument is needed to establish the asymptotic conditional distribution for the case of $0< q < p$. First, note that for $j \in \{j: \beta_{0,j} = 0 \}$, $\sqrt{n} \widehat{\beta}_{n,j}^{SC}$ has an asymptotic normal distribution (denoted $Z_j$) under $P_D$. By the Skorokhod representation theorem, there exists random variables $U_{n,j}$ and $U_j$ such that $U_{n,j} \stackrel{d}{=} \sqrt{n} \widehat{\beta}_{n,j}^{SC}$, $U_j \stackrel{d}{=} Z_j$, and $U_{n,j} \to U_j \,\, a.s. \,\, P_D$. Then, for $(\lambda_n / \sqrt{n}) \to \lambda_0 \in [0,\infty)$,
\begin{align} \label{eq.SkorokhodResult}
\sqrt{n} \left( \bnw - \bSC \right) 
\CONV{c.d.}
\argmin_{\bm{u}} V^* (\bm{u})
\quad a.s. \,\, P_D,
\end{align}
\noindent where
\begin{alignat*}{2}
V^*(\bm{u}) 
&= &&-2 \bm{u}' \Psi + \mu_W \bm{u}' C \bm{u} \\
&	&&+ \lambda_0 \sum_{j=1}^p W_j
\left[
u_j \, \text{sgn}(\beta_{0,j}) 
\mathbbm{1}_{ \{ \beta_{0,j} \neq 0 \}}
+ \left( | U_j + u_j | - | U_j | \right) 
\mathbbm{1}_{ \{ \beta_{0,j} = 0 \}} 
\right],
\end{alignat*}
for $\Psi$ and $\{W_j\}_{1 \leq j \leq p}$ defined in Theorem \ref{thm_low_AsympDistn}. 

The current ``one-step" random-weighting setup (\ref{eq.Bnw.setup}) in Algorithm \ref{alg:ALG_general} does not produce random-weighting samples that have conditional sparse normality property. From Theorems \ref{thm_Model_Select} and \ref{thm_low_AsympDistn}, it is evident that even under a fixed dimensional ($p_n = p$) setting, the random weighting samples achieve conditional model selection consistency when $\lambda_n = \mathcal{O} \left( n^{c} \right)$ for some $\frac{1}{2} < c < 1$, whereas conditional asymptotic normality happens when  $\lambda_n = o \left( \sqrt{n} \right)$. 

Unsurprisingly, this finding about (lack of) conditional sparse normality approximation coincides with many existing Bayesian and frequentist results. For instance, in the Bayesian framework, Theorem 7 of \citet{Castillo2015} proved that the Bayesian LASSO approach \citep{BayesianLasso} could not achieve asymptotic sparse normality for any one given $\lambda_n$ due to the conflicting demands of sparsity-inducement and normality approximation on the regularization parameter $\lambda_n$. In the frequentist setting, \citet{Liu&Yu} pointed out that there does not exist one $\lambda_n$ that allows a standard LASSO estimator (\ref{eq.LassoObj}) to simultaneously achieve model selection and asymptotic normality. Consequently, many variations of ``two-step" LASSO estimators (e.g., \citet{Zou2006}'s ALasso), and their corresponding bootstrap procedures (e.g., \citet{DasGreg2019}'s perturbation bootstrap of ALasso) were introduced to overcome this shortcoming. 

\subsection{Two-step Procedure} \label{sec:MainResults2}

We now propose an extension to our random-weighting procedure in LASSO regression (\ref{eq.Bnw.setup}). Specifically, we retain the random-weighting framework of repeatedly assigning random-weights and optimizing the objective function (\ref{eq.Bnw.setup}), except that now optimization consists of two-steps: In step one, we optimize 
\begin{align} \label{eq:first_step}
\min_{\bm{\beta}}
\left\{
\sumin W_i ( y_i - \bm{x}_i' \bm{\beta} )^2 
+ \lambda_n \sum_{j=1}^{p_n} W_{0,j} |\beta_j|
\right\}
\end{align}
to select variables. Let $\widehat{S}_n^w \subseteq \{1, \cdots, p_n\}$ be the set of variables being selected in (\ref{eq:first_step}), and let $(\widehat{S}_n^w)^c$ be the set of discarded variables. In addition, denote $X_{\widehat{S}_n^w}$ as the $n \times | \widehat{S}_n^w  |$ submatrix of $X$ whose columns correspond to the selected variables in (\ref{eq:first_step}). Then, in step two, we obtain our random-weighting samples by solving 
\begin{align} \label{eq:second_step}
\bnw := 
\begin{bmatrix}
	\widehat{\bm{\beta}}^w_{n, \widehat{S}_n^w} \\
	\\
	\widehat{\bm{\beta}}^w_{n, (\widehat{S}_n^w)^c} 
\end{bmatrix}
:= 
\begin{bmatrix}
	\left(
		X_{\widehat{S}_n^w}' D_n X_{\widehat{S}_n^w} 
	\right)^{-1}
		X_{\widehat{S}_n^w}' D_n Y \\
	\\
	\bm{0}
\end{bmatrix},
\end{align}
where the partition of $\bnw$ corresponds to $\widehat{S}_n^w$ and $\left(\widehat{S}_n^w \right)^c$. 

\begin{algorithm}
	\SetAlgoLined
	\caption{Random-Weighting in LASSO+LS regression}
	\label{alg:ALG_general_2step}
	\Input{ 
		\begin{itemize}[itemsep=0pt]
			\item data: $D = (\bm{y}, X)$ 
			\item regularization parameter: $\lambda_n$
			\item number of draws: $B$
			\item choice of random weight distribution: $F_W$ 
			\item choice of weighting schemes: (\ref{eq.NoWeightPenalty}), (\ref{eq.Bnw.WeightPenalty}) or (\ref{eq.Bnw.WeightPenalty2})  
		\end{itemize}
	}
	\Output{
		\begin{itemize}[itemsep=0pt] 
			\item $B$ sets of selected variables $\{\widehat{S}_n^{w,b} \}_{b=1}^B$ 
			\item $B$ parameter samples $\{\widehat{\bm{\beta}}_n^{w,b} \}_{b=1}^B$ 
		\end{itemize}
	}
	\For {$b = 1$ to $B$} {
		Draw i.i.d. random weights from $F_W$ and substitute them into ~(\ref{eq.Bnw.setup}) \;
		Optimize (\ref{eq:first_step}) to obtain $\widehat{S}_n^{w,b}$ \; 
		Based on the selected set of variables $\widehat{S}_n^{w,b}$, obtain $\widehat{\bm{\beta}}_n^{w,b}$ by solving ~(\ref{eq:second_step}) \; 
	}   
\end{algorithm}

For convenience, we shall refer to this proposed extension as a ``two-step procedure", which is laid out in detail in Algorithm \ref{alg:ALG_general_2step}. This extension can be seen as the random-weighting version of \citet{Liu&Yu}'s LASSO+LS procedure, i.e., a LASSO step (\ref{eq.LassoObj}) for variable selection followed by a least-square estimation for the selected variables. We shall denote this unweighted two-step LASSO+LS estimator as $\widehat{\bm{\beta}}_n^{LAS+LS}$, and let $\widehat{S}_n$ be the set of variables selected (in the first step) by this estimator. Notice that $\widehat{S}_n$ and $\widehat{S}^w_n$ may be different due to the presence of random-weights in the selection step of (\ref{eq:first_step}). The superscript \textit{w} of $\widehat{S}_n^w$ helps to remind readers that the set of selected variables in (\ref{eq:first_step}) could change with different sets of assigned random weights.  

In this subsection, we adopt the same assumptions as we did in Theorem \ref{thm_Model_Select}, including the fact that $p_n \leq n$ and $X$ is full rank for all $n$. Thus $X_{\widehat{S}^w_n}$ is full rank and consequently, 
$$
X_{\widehat{S}^w_n}' D_n X_{\widehat{S}^w_n}
$$
is also full rank and is invertible for all $n$.   

For ease of presentation, we  introduce a bit of additional notation. Let $S_0$ be the true set of relevant variables. To be consistent with our previous notation, we remind readers that $S_0 = \{1, \cdots, q\}$ without loss of generality, and $X_{S_0}$ = $X_{(1)}$. We  also partition $\bnw$ and $\widehat{\bm{\beta}}_n^{LAS+LS}$ into  
$$
\bnw = 
\begin{bmatrix}
\bnwa \\
\\
\widehat{\bm{\beta}}^w_{n(2)}
\end{bmatrix}
\quad \quad \text{ and } \quad \quad
\widehat{\bm{\beta}}_n^{LAS+LS} =
\begin{bmatrix}
\widehat{\bm{\beta}}^{LAS+LS}_{n(1)} \\
\\
\widehat{\bm{\beta}}^{LAS+LS}_{n(2)}
\end{bmatrix}
$$
respectively, which correspond to the partition of $\bm{\beta}_0 = \left[ \bm{\beta}_{0(1)} \,\,\, \bm{\beta}_{0(2)} \right]'$. We observe that if $\widehat{S}^w_n = S_0$, then 
$$
\widehat{\bm{\beta}}^w_{n, \widehat{S}_n^w} = \bnwa
\quad \text{ and } \quad 
\widehat{\bm{\beta}}^w_{n, (\widehat{S}_n^w)^c}  = \widehat{\bm{\beta}}^w_{n (2)} 
= \bm{\beta}_{0(2)} = \bm{0}.
$$
Similarly, if $\widehat{S}_n = S_0$, then 
$$
\widehat{\bm{\beta}}_{n, \widehat{S}_n}^{LAS+LS}  = \widehat{\bm{\beta}}_{n (1)}^{LAS+LS} 
\quad \text{ and } \quad 
\widehat{\bm{\beta}}_{n, (\widehat{S}_n)^c}^{LAS+LS}   = \widehat{\bm{\beta}}_{n (2)}^{LAS+LS} 
= \bm{\beta}_{0(2)} = \bm{0}.
$$ 
We are now ready to establish the conditional sparse normality property of the two-step random-weighting samples (\ref{eq:second_step}) under growing $p_n$ setting with appropriate regularity conditions.

\begin{thm} \label{thm_cond_oracle}
	\textbf{(Conditional Sparse Normality)} Adopt all regularity assumptions as stated in Theorem \ref{thm_Model_Select} (including assumptions about the different rates of $\lambda_n$ and $p_n$ for weighting schemes (\ref{eq.NoWeightPenalty}), (\ref{eq.Bnw.WeightPenalty}) and (\ref{eq.Bnw.WeightPenalty2})). Furthermore, assume $\mu_W = 1$ and $C_{n(11)} \to C_{11}$ for some nonsingular matrix $C_{11}$. Let $\bnw$ be the two-step random-weighting samples defined in (\ref{eq:second_step}), and let $\widehat{\bm{\beta}}_n^{LAS+LS}$ be the unweighted two-step LASSO+LS estimator (i.e. a LASSO variable selection step (\ref{eq.LassoObj}) followed by least-squares estimation for the selected variables). Then,
	$$
	P\left(
		\widehat{S}^w_n = S_0
		\big| \mathcal{F}_n 
	\right)	
	\to 1
	\quad a.s. \,\, P_D,
	$$
	and
	$$
	\sqrt{n} \left( 
	\widehat{\bm{\beta}}_{n(1)}^w - \widehat{\bm{\beta}}^{LAS+LS}_{n(1)}
	\right) 
	\CONV{c.d.} 
	N_q \left( 
	\bm{0} \,\, , \,\, 
	\sigma^2_W \sigma^2_{\epsilon} C_{11}^{-1} 
	\right)
	\quad a.s. \,\, P_D.
	$$
\end{thm}

Theorem \ref{thm_cond_oracle} highlights the improvement brought about by the extended random-weighting framework. With a common regularization parameter $\lambda_n$ (and all regularity conditions that apply), the two-step random-weighting samples  attain conditional model selection consistency and achieve conditional asymptotic normality (by centering at the unweighted two-step LASSO+LS estimator) on the true support $S_0$ under growing $p_n$ setting.   

We conclude this section by establishing that the random-weighting samples from the two-step procedure also achieve the conditional consistency property under growing $p_n$ setting. This could be viewed as an improvement to the result that we have in Theorem \ref{thm_low_Consistency}(a) which applies to fixed dimensional setting only. 

\begin{thm} \label{thm_high_Consistency}
	\textbf{(Conditional Consistency)} Adopt all regularity assumptions as stated in Theorem \ref{thm_Model_Select} (including assumptions about the different rates of $\lambda_n$ and $p_n$ for weighting schemes (\ref{eq.NoWeightPenalty}), (\ref{eq.Bnw.WeightPenalty}) and (\ref{eq.Bnw.WeightPenalty2})). Let $\bnw$ be the two-step random-weighting samples defined in (\ref{eq:second_step}). Then
	$$
	\left\Vert 
		\bnw - \bm{\beta}_0 
	\right\Vert_2
	\CONV{c.p.} 0
	\quad a.s. \,\, P_D. 
	$$  
\end{thm}

Theorem \ref{thm_high_Consistency} indicates a concentration of the conditional distribution of $\bnw$ near $\bm{\beta}_0$ with increasing sample size  given almost any data set. 

\section{Discussion} \label{sec:discuss}

\subsection{Approximate Bayesian Inference}

In fixed dimensional ($p_n$ = $p$) setting where $\bm{\beta}_0$ is not sparse (i.e. $q = p$), Theorems \ref{thm_low_Consistency} and \ref{thm_low_AsympDistn} describe the first order behavior of the conditional distribution of the one-step random-weighting samples $\bnw$. Under typical parametric Bayesian inference for $\bm{\beta}$ in the linear model (\ref{eq.LinearModel}), for any prior measure of $\bm{\beta}$ that is absolutely continuous in a neighborhood of $\bm{\beta}_0$ with a continuous positive density at $\bm{\beta}_0$, the Berstein-von Mises Theorem (e.g., Theorem 10.1 of \citet{vanderVaartbook}) ensures that for every Borel set $A \subset \Theta \subset \mathbb{R}^p$,
$$
P \left[
\sqrt{n} 
\left(
\bm{\beta} -  \bMLE
\right)
\in A \big| \mathcal{F}_n
\right]
\to 
P \left[
Z \in A
\right]
$$   
along almost every sample path, where $Z \sim N ( \bm{0}, \sigma^2_{\epsilon} C^{-1} )$. Hence, based on Theorem \ref{thm_low_AsympDistn} (with centering on $\bMLE$ = $\bLS$), for any $\lambda_n = o(\sqrt{n})$, by drawing random weights from $F_W$ with unitary mean and variance ($\mu_W = \sigma^2_W = 1$), the conditional distribution of the one-step random-weighting samples $\bnw$ converges to the same limit as in the Bernstein-von Mises Theorem, i.e., the conditional distribution of $\bnw$ is the same -- at least up to the first order -- as the posterior distribution of $\bm{\beta}$ under the regime of Bayesian inference. 

Theorem \ref{thm_low_AsympDistn} (with centering on $\bMLE$) highlights an important implication for the choice of $F_W$ in deploying the random-weighting approach to approximate posterior inference. Specifically, non-unitary mean or variance of the random weights would cause the random-weighting samples to converge to a conditional normal distribution with an asymptotic variance that is different from the one guaranteed by the Bernstein-von-Mises Theorem. 

 \citet{Newton&Raftery}'s first-order approximation theory for the random-weighting method relies on some classical regularity assumptions that do not hold in the LASSO setting studied here~(\ref{eq.LassoObj}). 
 The present work therefore extends the range of cases in which random-weighting operates successfully in large samples to achieve approximate Bayesian inference. 
 %Our work proves the affirmative in that setting, assuming the conditions laid out in Theorems \ref{thm_low_Consistency} and \ref{thm_low_AsympDistn}.

Comparison of random weighting and posterior
distribution is less straightforward in cases where $\bm{\beta}_0$ is sparse. \citet{Castillo2015} used a mixture of point masses at zero and continuous distributions as a sparse prior in their full Bayesian procedures for high-dimensional sparse linear regression. For this sparse prior, they showed that the resulting posterior distribution is not approximated by a non-singular normal, but by a random mixture of different dimensional normal distributions.
\textcolor{black}{Whilst we do not have an explicit result on the distributional approximation for $\bnw$ in growing-$p_n$ setting (e.g., Theorem 6 of \citet{Castillo2015}), our Theorem \ref{thm_cond_oracle} ensures that the conditional distribution of $\bnw$ does amass around the true support of $\bm{\beta}$, and on the true support, the random-weighting samples attain asymptotic Gaussian distributional behavior. Theorem 3.4 is therefore comparable to Corollary 2 in \citet{Castillo2015}, although different techniques are deployed; for instance we consider almost sure weak conditional convergence, whereas \citet{Castillo2015} considers sample average total-variation distance convergence, and we have no explicit prior structure.  Yet the basic message of both is that the mass of the posterior distribution, on the one hand, and the random-weighting distribution, on the other, are similarly concentrating on the correct model subset according to the same Gaussian law.} We also acknowledge the fact that these Bayesian models could handle high-dimensional problem where $p_n$ grows nearly exponential with sample size $n$ by using sparse-inducing priors on $\bm{\beta}$. On the other hand, our results require $p_n$ to grow at a polynomial rate of $o(\sqrt{n})$.

%We refer readers to Section 2.4 of \citet{Castillo2015} for more details. On the other hand, our Theorem \ref{thm_cond_oracle} describes the asymptotic Gaussian distributional behavior  of the two-step random-weighting samples $\bnw$ on the true support $S_0$, conditionally upon data. Further,  the conditional probability of random-weighting samples having the same signs as the true model converges to one for almost every data set (Theorem (\ref{thm_Model_Select} ). This result has a similar flavor to the model selection consistency property (conditional on data) of the existing Bayesian procedures  \citep[e.g.,][]{NarisettyHe2014,Castillo2015}.                   

\subsection{Sampling Theory Interpretation}

Though random weighting was motivated from a Bayesian perspective, the two-step random-weighting procedure is a valid bootstrap procedure for \citet{Liu&Yu}'s LASSO+LS estimator $\widehat{\bm{\beta}}_n^{LAS+LS}$ under growing $p_n$ setting. Specifically, using very similar regularity assumptions, \citet{Liu&Yu} showed that their LASSO+LS method results in consistent model selection under $P_D$, and 
$$
\sqrt{n} \left(
\widehat{\bm{\beta}}_{n(1)}^{LAS+LS}
- \bm{\beta}_{0(1)}
\right)
$$
converges to $N \left( \bm{0} \, , \, \sigma^2_\epsilon C_{11}^{-1} \right)$ under $P_D$.  Hence, based on Theorem \ref{thm_cond_oracle}, by fulfilling the appropriate regularity assumptions and drawing random weights from $F_W$ with unitary mean and variance ($\mu_W = \sigma^2_W = 1$), the conditional distribution of the two-step random-weighting samples $\bnw$ converges to the same distributional limit of the LASSO+LS estimator under $P_D$. This enables the two-step random-weighting procedure to produce bootstrap samples that provide valid distributional approximation to the LASSO+LS estimator for inference procedures such as hypothesis testing or constructing confidence regions.

We also point out that by capitalizing on the sub-Gaussian nature of $\bm{\epsilon}$, \citet{Liu&Yu}'s proposed residual bootstrap procedure for their LASSO+LS estimator works under high-dimensional setting where $p_n$ grows nearly exponential with sample size $n$. On the other hand, in this paper, we only require finite fourth moment assumptions for both error term $\bm{\epsilon}$ and random weights $\bm{W}$, and our random-weighting procedure only allows $p_n$ to grow at a polynomial rate of $o(\sqrt{n})$.   

Similarly, under fixed dimensional ($p_n = p$) setting where $\bm{\beta}_0$ is not sparse (i.e. $q = p$), our one-step random-weighting approach in Algorithm \ref{alg:ALG_general} could also be a valid bootstrap procedure for the standard LASSO estimator $\bLAS (\lambda_n)$. Specifically, \citet{Knight&Fu} proved that for $(\lambda_n/ \sqrt{n}) \to \lambda_0 \in [0,\infty)$, 
$$
\sqrt{n} \left(
	\bLAS(\lambda_n) - \bm{\beta}_0
\right)
$$      
converges to the same distributional limit stated in Theorem \ref{thm_low_AsympDistn} under $P_D$. However, for the case where $q < p$, the one-step random-weighting procedure no longer provides valid distributional approximation to $\bLAS (\lambda_n)$, as evident from the Skorokhod argument. This mimics the asymptotic conditional distribution of the LASSO parametric residual bootstrap \citep{Knight&Fu}.      

\section{Numerical Experiments} \label{sec:numerical}

We perform simulation studies and  data analysis using R  \citep{R};  all source code is available at the Github public repository:  \url{https://github.com/wiscstatman/optimizetointegrate/tree/master/Tun}.  

\subsection{Simulation: Part I}

A simulation study of one-step random-weighting procedures (Algorithm \ref{alg:ALG_general}) was previously reported \citep{WBB}, and so here we study performance of the two-step random-weighting procedure (Algorithm \ref{alg:ALG_general_2step}) for all three weighting schemes (\ref{eq.NoWeightPenalty}), (\ref{eq.Bnw.WeightPenalty}) and (\ref{eq.Bnw.WeightPenalty2}) -- denoted RW1, RW2 and RW3 respectively --  in several experimental settings, and compare it with: 
\begin{itemize}
	\item Bayesian LASSO \citep{BayesianLasso}, which can be easily implemented with R package \texttt{monomvn} \citep{monomvn}
	\item parametric residual bootstrap \citep{Knight&Fu}, which is a very common and easily implementable bootstrap procedure in LASSO regression. We denote this method as RB thereafter.  
\end{itemize}

We drew inspiration from \citet{Das2019}, \citet{Liu&Yu} and \citet{WBB} in setting up our simulation schemes. Specifically, we consider 8 simulation settings as tabulated in Table \ref{tabl:simul_setting}. In all  settings, the generative state $\bm{\beta}_0 = (\beta_{0,1}, \cdots, \beta_{0,p})'$ is defined as $\beta_{0,j} = (3/4) + (1/4)j$ for $j = 1, \cdots, q$ and $\beta_{0,j} = 0$ for $j = q+1, \cdots, p$. The predictors $\bm{x}_i$ are drawn from $p$-variate normal distribution with different covariance structures. $\Sigma^{(1)}$ has the following structure
\begin{align} \label{sigmaX_structure}
\Sigma^{(1)}_{i,j} 
= \mathbbm{1}_{\{i = j\}} +
\mathbbm{1}_{\{i \neq j\}} \times 
\left(
	0.3^{|i-j|}
	\mathbbm{1}_{\{i \leq q\}}
	\mathbbm{1}_{\{j \leq q\}}
\right) 
\quad \text{for} \quad
1 \leq i,j \leq 10.
\end{align}   
$\Sigma^{(3)}$ also has the same structure as (\ref{sigmaX_structure}), except that it has larger dimension $p= 50$. Meanwhile, $\Sigma^{(2)}$ has the following structure: for $1 \leq i,j \leq 10$,  
$$
\Sigma^{(2)}_{i,j} 
= \mathbbm{1}_{\{i = j\}} +
\mathbbm{1}_{\{i \neq j\}} \times 
\left[
	0.4
	\mathbbm{1}_{\{i \leq q\}}
	\mathbbm{1}_{\{j \leq q\}}
	+ 0.5
	\left(
	1 - 
	\mathbbm{1}_{\{i \leq q\}}
	\mathbbm{1}_{\{j \leq q\}}
	\right)
\right].
$$
We verify that only simulation settings 5 and 6 violate the strong irrepresentable condition (\ref{assume_StrongIrrepresent}), whereas the other six simulation settings satisfy assumption (\ref{assume_StrongIrrepresent}). By simulating i.i.d. $\epsilon_i$ and $\bm{x}_i$, we generate $y_i = \bm{x}_i \bm{\beta}_0 + \epsilon_i$ for $i = 1, \cdots, n$. 

\begin{table}[h!] 
	\caption{Simulation Settings} \label{tabl:simul_setting}
	\centering
	\begin{tabular}{cccccc}
		\hline
		\\[-.5em]
		Setting & $n$ & $p$ & $q$ & $\epsilon_i$ & $\bm{x}_i \sim N_p(\bm{0},\Sigma)$ \\
		\\[-.5em]
		\hline
		\\[-.5em]
		1 & 100 & 10 & 6 & $N(0,1)$ & $\Sigma = \Sigma^{(1)}$ \\ 
		\\[-.5em]
		2 & 500 & 10 & 6 & $N(0,1)$ & $\Sigma = \Sigma^{(1)}$ \\
		\\[-.5em]
		\hline 
		\\[-.5em]
		3 & 100 & 10 & 6 & $\chi^2_2 - 2$ & $\Sigma = \Sigma^{(1)}$ \\ 
		\\[-.5em]
		4 & 500 & 10 & 6 & $\chi^2_2 - 2$ & $\Sigma = \Sigma^{(1)}$ \\ 
		\\[-.5em]
		\hline
		\\[-.5em]
		5 & 100 & 10 & 6 & $N(0,1)$ & $\Sigma = \Sigma^{(2)}$ \\ 
		\\[-.5em]
		6 & 500 & 10 & 6 & $N(0,1)$ & $\Sigma = \Sigma^{(2)}$ \\ 
		\\[-.5em]
		\hline
		\\[-.5em]
		7 & 100 & 50 & 6 & $N(0,1)$ & $\Sigma = \Sigma^{(3)}$ \\ 
		\\[-.5em]
		8 & 500 & 50 & 6 & $N(0,1)$ & $\Sigma = \Sigma^{(3)}$ \\ 
		\\[-.5em]
		\hline
	\end{tabular}
\end{table}

\begin{figure}[!]
	\centering
	\includegraphics[scale=0.8]{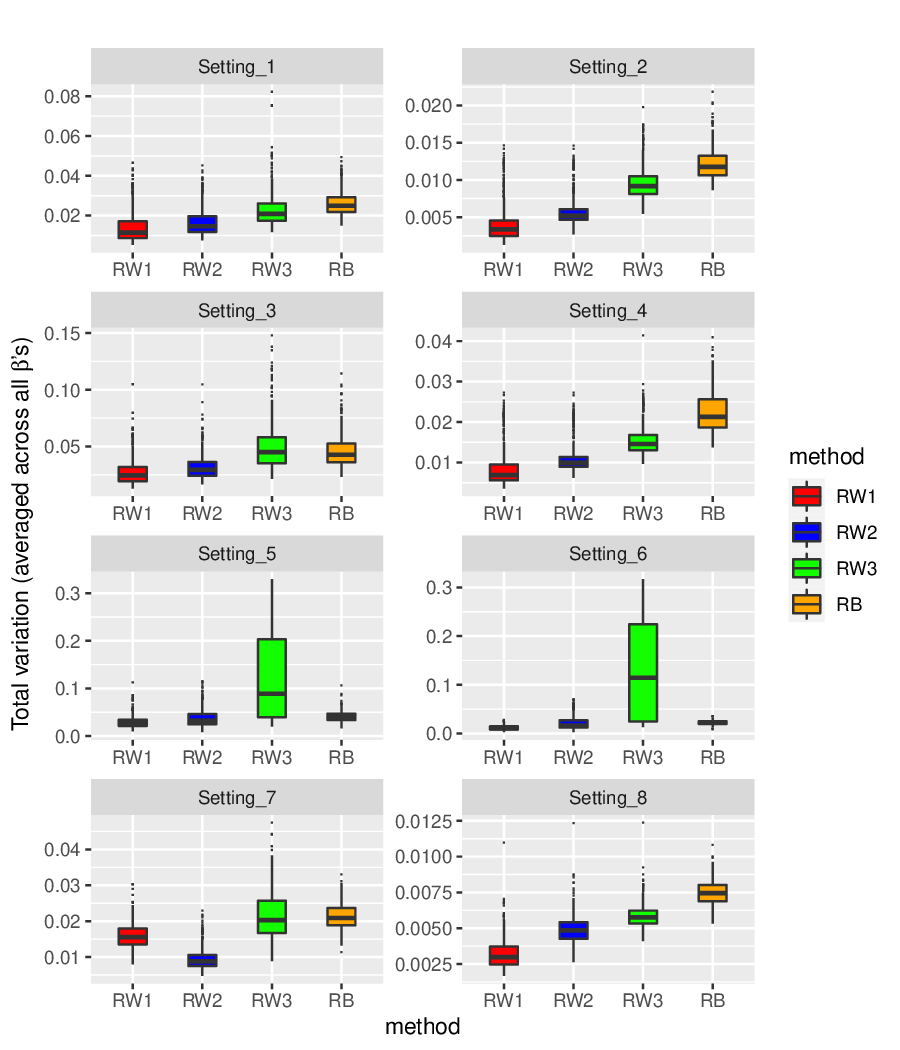}
	\caption{Simulation Part I: Sampling distribution of total variation distance between random-weighting distribution and target posterior (averaged across all $\beta$'s) among $T = 500$ simulated data sets in 8 simulation settings between ecdf of MCMC samples and ecdf of samples from each of the 4 methods: two-step random-weighting approach using weighting schemes (\ref{eq.NoWeightPenalty}) (denoted RW1), (\ref{eq.Bnw.WeightPenalty}) (denoted RW2) and (\ref{eq.Bnw.WeightPenalty2}) (denoted RW3), and LASSO residual bootstrap (denoted RB).}
	\label{fig:simul1_TV}
\end{figure}  

\textbf{Purpose of simulation setup:} The even-numbered simulation settings share the same specifications as their odd-numbered counterparts except with larger sample size $n$ (e.g. Setting 2 versus Setting 1, Setting 4 versus Setting 3, et cetera). Simulation Settings 3 and 4 are used as an example of cases where the error term $\bm{\epsilon}$ is no longer normally distributed, whereas Simulation Settings 5 and 6 are set up to illustrate the situations where the strong irrepresentable condition (\ref{assume_StrongIrrepresent}) is violated. Finally, we increase the dimension $p$ of predictors by five-fold in Settings 7 and 8 to compare performances in higher-dimensional setting.     

For each simulation setting, we generate $T=500$ independent datasets. For each simulated data set, we draw $B=1000$ posterior/bootstrap samples from the 5 aforementioned methods: Bayesian LASSO (BLASSO), two-step random-weighting with schemes (\ref{eq.NoWeightPenalty}), (\ref{eq.Bnw.WeightPenalty}) and (\ref{eq.Bnw.WeightPenalty2}), and residual bootstrap. For the Bayesian LASSO procedure, we specify a 2000 burn-in period. In addition, Bayesian LASSO imposes a noninformative marginal prior on $\sigma^2_\epsilon$, $\pi(\sigma^2_\epsilon) \sim 1/\sigma^2_\epsilon$, and a Jeffrey's prior on $\lambda_n$. To induce sparsity in the MCMC samples of $\bm{\beta}$, the posterior distribution is sampled by a Reversible Jump Markov Chain Monte Carlo (RJMCMC) algorithm \citep{rjmcmc}, with a uniform prior specified on the number of non-zero coefficients to be included in the model. For the three random-weighting schemes, all i.i.d. random weights are drawn from a standard exponential distribution. The regularization parameter $\lambda_n$ is chosen via cross-validation using \citet{Liu&Yu}'s (unweighted) LASSO+LS procedure, and then the same $\lambda_n$ is used to draw the 1000 random-weighting samples according to Algorithm \ref{alg:ALG_general_2step}. We note that the optimization step (\ref{eq:first_step}) can be easily computed using R package \texttt{glmnet} \citep{glmnet}. Meanwhile for residual bootstrap, its regularization parameter $\lambda_n^{\rm RB}$ is chosen via cross-validation using standard LASSO, and values of $\lambda_n^{\rm RB}$ are thereafter fixed for all bootstrap computations on the same dataset. 

\begin{figure}[!]
	\centering
	\includegraphics[scale=0.7]{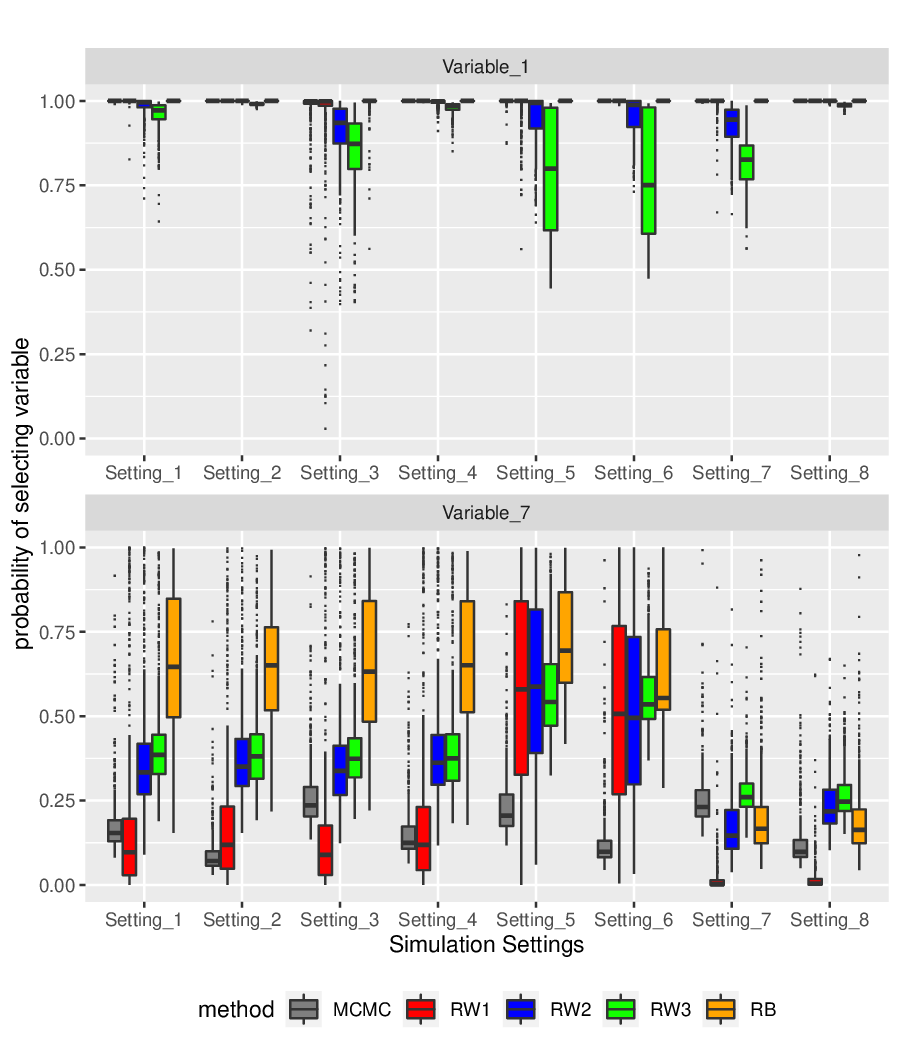}
	\caption{Simulation Part I: Sampling distribution of conditional (on data) probabilities of selecting $\beta_1$ and $\beta_7$ among $T = 500$ simulated data sets in 8 simulation settings by the 5 methods: MCMC via Bayesian LASSO, two-step random-weighting approach using weighting schemes (\ref{eq.NoWeightPenalty}) (denoted RW1), (\ref{eq.Bnw.WeightPenalty}) (denoted RW2) and (\ref{eq.Bnw.WeightPenalty2}) (denoted RW3), and LASSO residual bootstrap (denoted RB).}
	\label{fig:simul1_prob_select}
\end{figure}  

For each of the five aforementioned methods, we obtain $\{ \widehat{\beta}_j^{(b,t)} \}$ that represents the $j^{th}$ component of sampled/bootstrapped $\bm{\beta}$ in the $b^{th}$ iteration for the $t^{th}$ simulated data set, where $j = 1, \cdots, p$, and $b = 1, \cdots, B$, and $t = 1, \cdots, T$. To be precise, we have 
$$
\left\{
\widehat{\beta}_{j (\rm MCMC)}^{(b,t)}, 
\widehat{\beta}_{j (\rm RW1)}^{(b,t)}, 
\widehat{\beta}_{j (\rm RW2)}^{(b,t)}, 
\widehat{\beta}_{j (\rm RW3)}^{(b,t)}, 
\widehat{\beta}_{j (\rm RB)}^{(b,t)}
\right\}
$$
that correspond to the sampled/bootstrapped $\bm{\beta}$'s of the five aforementioned methods, but for brevity we drop the subscripts whenever it does not cause any confusion, since each method is subject to the same performance evaluation. We then assess the performances of each of these five methods -- BLASSO, RW1, RW2, RW3 and RB -- in each of the 8 simulation settings using the following comparison criteria:
\begin{itemize}
	\item Estimation MSE of coefficients. Specifically, for each simulated data set $t = 1, \cdots, T$, we keep track of 
	$$
	{\rm MSE}^{(t)} = \dfrac{1}{B} \sum_{b=1}^B 
	\left\Vert
	\bm{Y}^{(t)} - X^{(t)} \widehat{\bm{\beta}}^{(b,t)}
	\right\Vert_2^2.
	$$
	\item Out-of-sample prediction MSE (abbreviated as MSPE thereafter), where test sets are of the same size as the corresponding training sets. Similarly, for each simulated data set $t = 1, \cdots, T$, we keep track of 
	$$
	{\rm MSPE}^{(t)} = \dfrac{1}{B} \sum_{b=1}^B 
	\left\Vert
	\bm{Y}^{(t)}_{\rm test} - X^{(t)}_{\rm test} \widehat{\bm{\beta}}^{(b,t)}
	\right\Vert_2^2.
	$$ 
	\item Conditional (on data) probability of selecting the $j^{th}$ variable where $j = 1, \cdots, p$. Specifically, for each simulated data set $t = 1, \cdots, T$, we keep track of 
	$$
	\hat{p}_j^{(t)} :=  \dfrac{1}{B}
	\left\vert
	\left\{
	b: \widehat{\beta}_j^{(b,t)} \neq 0
	\right\}
	\right\vert.
	$$ 
	We note that the computation of $\hat{p}_j^{(t)}$ is sensible because all the five methods (including BLASSO with RJMCMC implementation) induce sparsity in the sampled/bootstrapped $\bm{\beta}$'s.
	\item Coverage and average width of the two-sided 90\% credible/confidence interval (CI) for the $j^{th}$ variable  where $j = 1, \cdots, p$. Specifically, denote $\hat{r}_{0.05,j}^{(t)}$ and $\hat{r}_{0.95,j}^{(t)}$ as the $5^{th}$ percentile and $95^{th}$ percentile of the empirical distribution of $\{ \widehat{\beta}_j^{(b,t)} \}_{1 \leq b \leq B}$. Then, the average width (across $T=500$ simulated data sets) of the two-sided 90\% CI for the $j^{th}$ variable is computed as
	$$
	\hat{l}_j :=
	\dfrac{1}{T} \sum_{t=1}^T 
	\left(
	\hat{r}_{0.95,j}^{(t)} - \hat{r}_{0.05,j}^{(t)}
	\right),
	$$
	and its corresponding empirical coverage is calculated as
	$$
	\hat{q}_j :=
	\dfrac{1}{T} %\sum_{t=1}^T 
	\left\vert
	\left\{
	t: \hat{r}_{0.05,j}^{(t)} 
	\leq \beta_{0,j} \leq
	\hat{r}_{0.95,j}^{(t)}
	\right\}
	\right\vert.
	$$
\end{itemize}  

In addition, we obtain the total variation distance between empirical cumulative distribution function (ecdf) of MCMC samples and ecdf of samples produced by one of the other four methods -- the two-step random-weighting (RW1, RW2 and RW3) and residual bootstrap (RB).  The intent is to assess how well the random-weighting methods approximate the MCMC-approximated posterior. Specifically, for the $j^{th}$ variable in the $t^{th}$ simulated data set, let 
$$
\hat{F}_{j(MCMC)}^{(t)} = \text{ ecdf of }
\left\{ \widehat{\beta}_{j(MCMC)}^{(b,t)} \right\}_{1 \leq b \leq B},
$$
and let $\hat{F}_{j(.)}^{(t)}$ be the ecdf of samples produced by one of the other 4 methods: RW1, RW2, RW3 or RB. Note that the ecdf's are easily obtained via the function \texttt{ecdf} in R \texttt{base} package \citep{R}. Then, for each of the 4 methods, we keep track of the total variation (averaged across all $p$ variables) for each simulated data set $t = 1, \cdots, T$:
$$
{\rm TV}^{(t)} = \dfrac{1}{p} \sum_{j=1}^p
\dfrac{1}{2}
\sum_{\omega \in \Omega}
\left\vert
\hat{F}_{j(MCMC)}^{(t)} (\omega)
- \hat{F}_{j(.)}^{(t)} (\omega)
\right\vert,
$$     
where the inner summation is approximated using a trapezoidal rule with an interval width of 0.001.

\begin{table}[ht]
	\caption{Empirical coverage $\hat{q}_j$ and average width $\hat{l}_j$ (in parentheses) of the two-sided 90\% CI for the first 10 variables in Simulation Setting 8, using the five approaches: MCMC via BLASSO, two-step random-weighting approach using weighting schemes (\ref{eq.NoWeightPenalty}) (denoted RW1), (\ref{eq.Bnw.WeightPenalty}) (denoted RW2) and (\ref{eq.Bnw.WeightPenalty2}) (denoted RW3), and LASSO residual bootstrap (denoted RB).}
	\centering
	\begin{tabular}{r|ccccc}
		\hline
		$\beta_{0,j}$ & MCMC & RW1 & RW2 & RW3 & RB \\ 
		\hline
		\multirow{2}{*}{1.00} & 0.918 & 0.878 & 0.882 & 0.906 & 0.344 \\ 
		& (0.161) & (0.152) & (0.152) & (0.16) & (0.153) \\ 
		&&&&&\\[-.6em]
		\multirow{2}{*}{1.25} & 0.908 & 0.88 & 0.876 & 0.904 & 0.588 \\ 
		& (0.169) & (0.158) & (0.159) & (0.168) & (0.16) \\ 
		&&&&&\\[-.6em]
		\multirow{2}{*}{1.50} & 0.894 & 0.864 & 0.868 & 0.886 & 0.578 \\ 
		& (0.168) & (0.158) & (0.158) & (0.165) & (0.16) \\
		&&&&&\\[-.6em] 
		\multirow{2}{*}{1.75} & 0.918 & 0.886 & 0.892 & 0.9 & 0.596 \\ 
		 & (0.168) & (0.159) & (0.159) & (0.165) & (0.16) \\ 
		 &&&&&\\[-.6em]
		\multirow{2}{*}{2.00} & 0.922 & 0.894 & 0.882 & 0.898 & 0.556 \\ 
		& (0.168) & (0.159) & (0.159) & (0.164) & (0.16) \\ 
		&&&&&\\[-.6em]
		\multirow{2}{*}{2.25} & 0.886 & 0.866 & 0.872 & 0.874 & 0.35 \\ 
		& (0.161) & (0.151) & (0.152) & (0.157) & (0.153) \\ 
		&&&&&\\[-.6em]
		\multirow{2}{*}{0.00} & 1 & 1 & 1 & 1 & 0.998 \\ 
		& (0.04) & (0.016) & (0.096) & (0.099) & (0.023) \\ 
		&&&&&\\[-.6em]
		\multirow{2}{*}{0.00} & 1 & 0.998 & 1 & 1 & 1 \\ 
		& (0.041) & (0.018) & (0.097) & (0.1) & (0.024) \\ 
		&&&&&\\[-.6em]
		\multirow{2}{*}{0.00} & 1 & 1 & 1 & 1 & 1 \\ 
		& (0.04) & (0.015) & (0.097) & (0.099) & (0.023) \\ 
		&&&&&\\[-.6em]
		\multirow{2}{*}{0.00} & 0.998 & 1 & 1 & 1 & 1 \\ 
		& (0.04) & (0.015) & (0.097) & (0.1) & (0.023) \\ 
		\hline
	\end{tabular}
\label{tabl:coverage_width}
\end{table}

Firstly, as expected, performance improves with larger sample size $n$, such as smaller MSE's, smaller MSPE's, higher coverage probabilities and narrower CI's.  Secondly, we note that the MSE's and MSPE's are very similar among all the five methods in all 8 simulation settings (figures not shown). However, the two-step random-weighting approach, especially weighting schemes (\ref{eq.NoWeightPenalty}) and (\ref{eq.Bnw.WeightPenalty}) -- denoted RW1 and RW2, outperforms the LASSO residual bootstrap (denoted RB) in all other performance measures.

Figure \ref{fig:simul1_TV} displays the sampling distribution of total variation distance between random-weighting distribution and target posterior (averaged across all $\beta$'s),  $\{TV^{(t)}\}_{1 \leq t \leq T}$, among the $T=500$ simulated data sets in the 8 simulation settings for the 4 methods: RW1, RW2, RW3 and RB. Generally, larger sample size $n$ leads to smaller total variations. Moreover, in all simulation settings, RW1 and RW2 have smaller total variations than that of RB, which illustrates the viability of the two-step random-weighting samples to approximate posterior inference. RW3 has larger total variations especially in Settings 5 and 6, where the strong irrepresentable condition (\ref{assume_StrongIrrepresent}) is violated. This illustrates the need for restrictive regularity assumption for weighting scheme (\ref{eq.Bnw.WeightPenalty2}) that we highlighted in part (c) of Theorem \ref{thm_Model_Select}.     

In Figure \ref{fig:simul1_prob_select}, we show the sampling distributions of $\big\{ \hat{p}_1^{(t)} \big\}_{1 \leq t \leq T}$ and $\big\{ \hat{p}_7^{(t)} \big\}_{1 \leq t \leq T}$ among the $T=500$ simulated data sets in the 8 simulation settings for all the five methods. Recall that the first variable corresponds to $\beta_{0,1} = 1$ and the seventh variable corresponds to $\beta_{0,7} = 0$. Sampling distribution of conditional (on data) probabilities of selecting other relevant predictors is similar to that of the first variable, and sampling distribution of conditional  probabilities of selecting other irrelevant predictors is similar to that of the seventh variable. In all 8 simulation settings, all methods almost always select the first variable, except for RW3 in Simulation Settings 5 and 6, due to the violation of condition (\ref{assume_StrongIrrepresent}). However, similar to MCMC, the two-step random-weighting schemes (especially RW1) have lower conditional  probabilities of selecting the seventh variable (which is an irrelevant predictor) than the LASSO RB. This illustrates that the two-step random-weighting approach is more capable of discarding irrelevant variables as compared to LASSO residual bootstrap. Only in Simulation Settings 5 and 6 do we see similarly high conditional  probabilities of selecting the seventh variable among RW1, RW2, RW3 and RB, due to violation of condition (\ref{assume_StrongIrrepresent}).

Empirical coverage and average width of the two-sided 90\% CI's for relevant predictors (i.e. $\beta_{0,j} \neq 0$)  paint a similar story. For illustration, the empirical coverage $\hat{q}_j$ and average width $\hat{l}_j$ (in parentheses) of the two-sided 90\% CI for the first 10 variables, i.e. for $j= 1, \cdots, 10$, in Simulation Setting 8, are tabulated in Table \ref{tabl:coverage_width}. Generally, average widths of CI's are similar among all five methods in all but two simulation settings, where RW3 has much wider 90\% CI's in Simulation Settings 5 and 6. Interestingly, empirical coverage for MCMC and random-weighting samples is similar and close to 90\% , but the LASSO residual bootstrap samples always have the lowest empirical coverage, especially in Simulation Settings 7 and 8, where their empirical coverage is only around 30\% - 40\%.

\subsection{Simulation: Part II}

On a separate calculation, we use Simulation Setting 2 (see Table \ref{tabl:simul_setting}) to illustrate that there are computational advantages in using $\lambda_n$ chosen via cross-validation on the unweighted LASSO+LS procedure \citep{Liu&Yu}, instead of cross-validation on the standard LASSO method, for obtaining the two-step random-weighting samples. For brevity, we shall refer to the former as the two-step cross validation, and the latter as the one-step cross validation. 

\begin{figure}[!]
	\centering
	\includegraphics[scale=0.6]{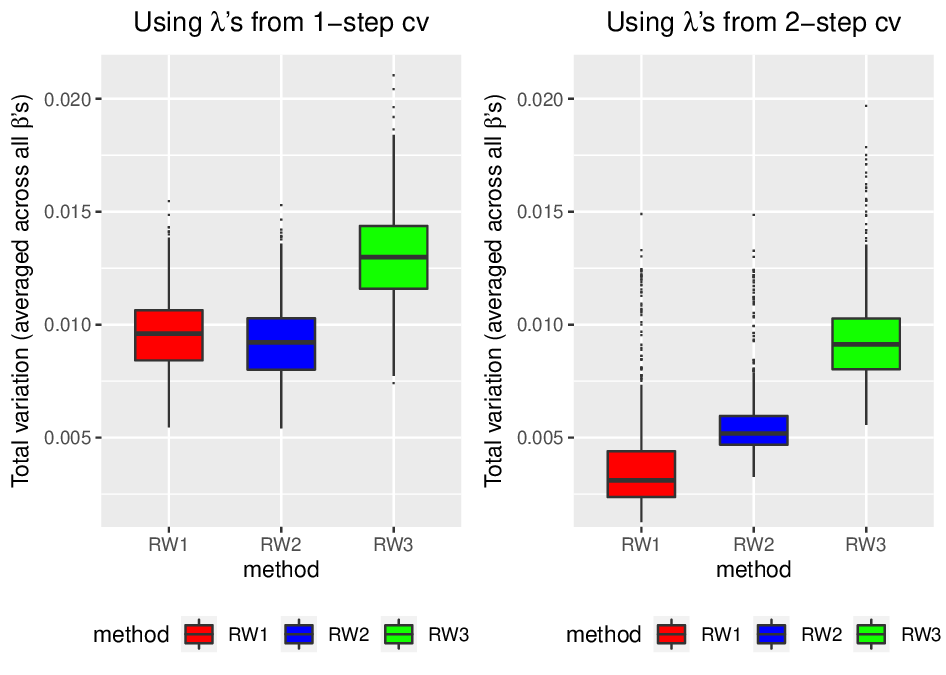}
	\caption{Simulation Part II: Sampling distribution of total variation distance between random-weighting distribution and 
	target posterior (averaged across all $\beta$'s) among $T = 500$ simulated data sets in Simulation Setting 2 between ecdf of MCMC samples and ecdf of the two-step random-weighting samples, computed with $\lambda_n$ obtained via 1-step cross validation or 2-step cross validation, using weighting schemes (\ref{eq.NoWeightPenalty}) (\ref{eq.Bnw.WeightPenalty}) and (\ref{eq.Bnw.WeightPenalty2}) (denoted RW1, RW2 and RW3 respectively).}
	\label{fig:simul2_TV}
\end{figure}  

\begin{figure}[!]
	\centering
	\includegraphics[scale=0.6]{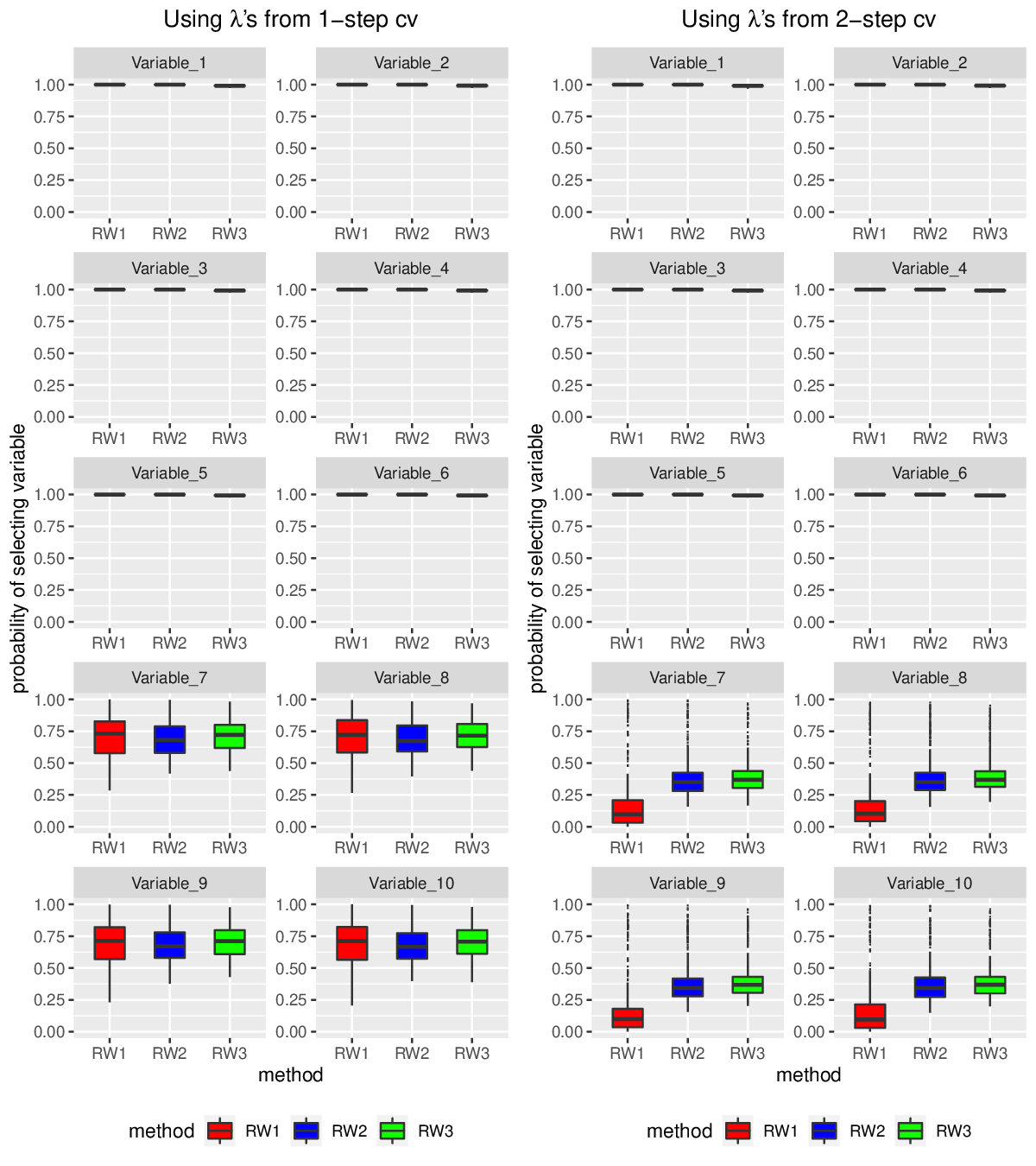}
	\caption{Simulation Part II: Sampling distribution of conditional (on data) probabilities of selecting $\bm{\beta}$'s among $T = 500$ simulated data sets in Simulation Setting 2 by the two-step random-weighting approach, computed with $\lambda_n$ obtained via 1-step cross validation or 2-step cross validation, using weighting schemes (\ref{eq.NoWeightPenalty}) (\ref{eq.Bnw.WeightPenalty}) and (\ref{eq.Bnw.WeightPenalty2}) (denoted RW1, RW2 and RW3 respectively).}
	\label{fig:simul2_prob_select}
\end{figure}  

Specifically, for each of the $T=500$ simulated data sets under Simulation Setting 2, we repeat the two-step random-weighting calculations outlined in Algorithm \ref{alg:ALG_general_2step}, but with $\lambda_n$ chosen via cross-validation on the standard LASSO method. This is in fact the same regularization parameter $\lambda_n^{\rm RB}$ that we used to generate the residual bootstrap samples. 

We find from the simulation results that the two-step cross-validation leads to larger $\lambda_n$ as compared to the one-step cross-validation. This ties back to the conflicting demands of the standard LASSO method on $\lambda_n$: smaller $\lambda_n$ allows more variables into the model to reduce estimation MSE; and larger $\lambda_n$ enables more regularization to discard irrelevant variables. On the other hand, using a two-step LASSO+LS procedure frees up these conflicting constraints on $\lambda_n$. 

For these two sets of random-weighting samples, we repeat the same calculations of performance measures as we did in Part I of our simulation studies. We found out that MSE's, MSPE's and empirical coverage of the two-sided 90\% CI are very similar between these two sets of random-weighting samples. However, from Figure \ref{fig:simul2_TV}, we see that larger regularization $\lambda_n$ based on the two-step cross validation leads to lower total variation distance between random-weighting distribution and target posterior, which indicates better approximation to the posterior samples. Meanwhile, in Figure \ref{fig:simul2_prob_select}, the random-weighting samples computed with the larger $\lambda_n$ have much lower conditional  probabilities of selecting irrelevant variables (variables 7 -- 10), whilst almost always selecting relevant predictors (variables 1 -- 6). This also helps to illustrate the fact that the two-step random-weighting approach is able to utilize more regularization to discard irrelevant predictors while maintaining estimation accuracy.   

\subsection{Benchmark data example}

To further illustrate the two-step random-weighting methodology, we apply it to the often-analyzed Boston Housing data set, which is available in the R package \texttt{MASS} \citep{MASS}. Data from $n = 506$ housing prices in the suburbs of Boston are available, with response the median value of owner-occupied homes in \$1000's, and with 13 variables ($p= 13$) listed in Table \ref{tabl:housing_variable}.

\begin{table}[h!] 
	\caption{Variables in Boston Housing Data Set} \label{tabl:housing_variable}
	\centering
	\begin{tabular}{l|l}
		\hline
		Abbreviation & Variable \\
		\hline
		crim & per capita crime rate by town \\
		zn & proportion of residential land zoned for lots over 25,000 sq.ft. \\
		indus & proportion of non-retail business acres per town \\
		chas & Charles River dummy variable (= 1 if tract bounds river; 0 otherwise) \\
		nox & nitrogen oxides concentration (parts per 10 million) \\
		rm & average number of rooms per dwelling \\
		age & proportion of owner-occupied units built prior to 1940 \\
		dis & weighted mean of distances to five Boston employment centers \\
		rad & index of accessibility to radial highways \\
		tax & full-value property-tax rate per \$10,000 \\
		ptratio & pupil-teacher ratio by town \\
		Black & proportion of Black residents by town \\
		lstat & lower status of the population (percent) \\
		\hline
\end{tabular}
\end{table}

\begin{figure}[!]
	\centering
	\includegraphics[scale=0.72]{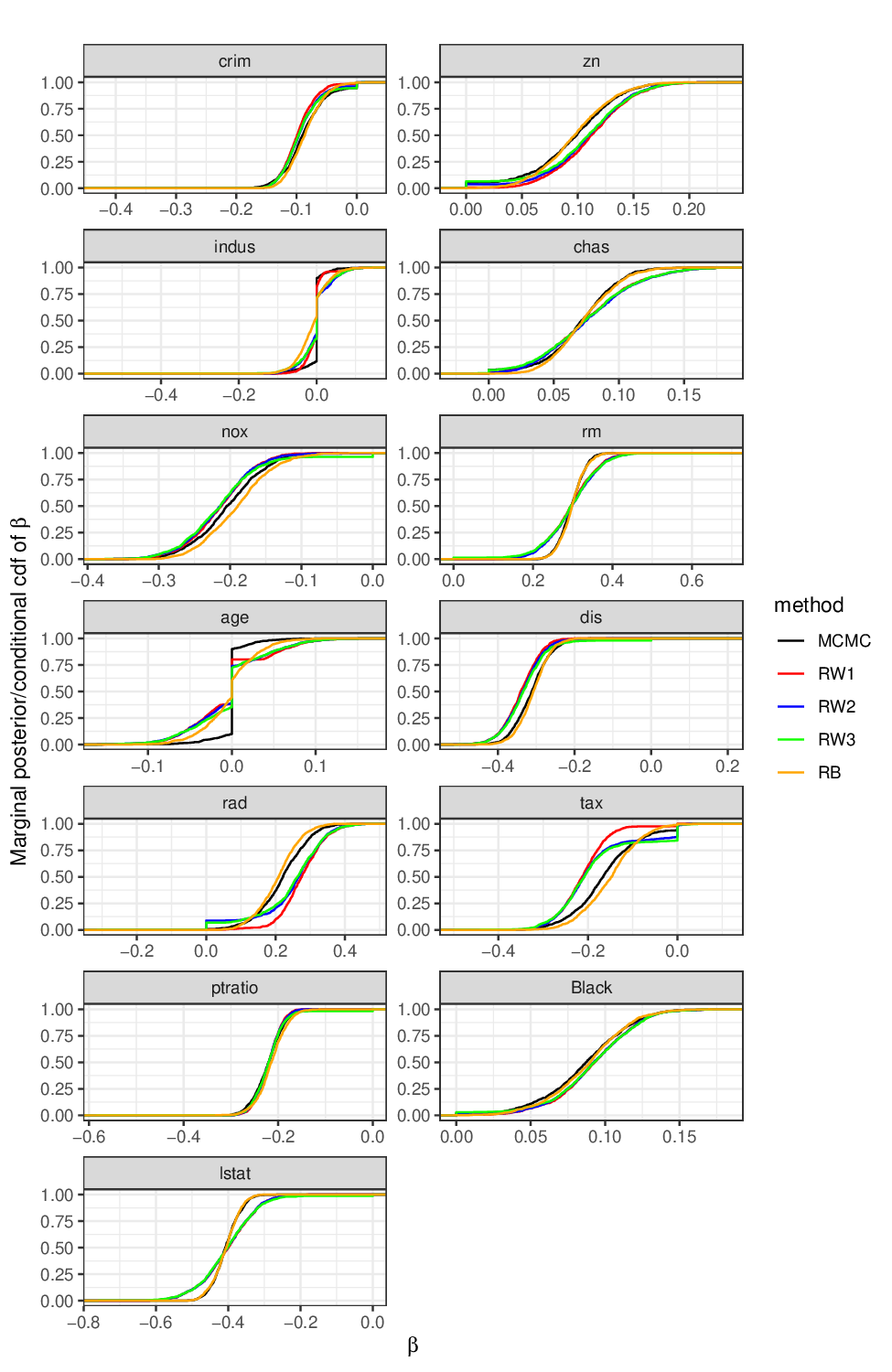}
	\caption{Boston Housing data example: Marginal posterior distribution plots for $\bm{\beta} = (\beta_1, \cdots, \beta_{13})'$ sampled from the 4 methods -- MCMC via Bayesian LASSO, and the two-step random-weighting approach using weighting schemes (\ref{eq.NoWeightPenalty}) (\ref{eq.Bnw.WeightPenalty}) and (\ref{eq.Bnw.WeightPenalty2}) (denoted RW1, RW2 and RW3 respectively).}
	\label{fig:boston_beta_cdf}
\end{figure}

Again, we apply Bayesian LASSO, the random-weighting approach for all three weighting schemes (\ref{eq.NoWeightPenalty}), (\ref{eq.Bnw.WeightPenalty}) and (\ref{eq.Bnw.WeightPenalty2}) according to Algorithm \ref{alg:ALG_general_2step}, as well as the parametric residual bootstrap method \citep{Knight&Fu} with $B=1000$.  We use the same prior specifications as well as RJMCMC implementation for Bayesian LASSO as we did in our simulation studies. For the random-weighting approach, random weights are drawn from a standard exponential distribution, and the regularization parameter is chosen with cross-validation using \citet{Liu&Yu}'s unweighted LASSO+LS procedure (i.e. 2-step cross-validation). Meanwhile, for residual bootstrap, its regularization parameter is chosen via cross-validation using standard LASSO. 

Figure \ref{fig:boston_beta_cdf} shows the marginal posterior distributions of $\bm{\beta}$'s sampled from MCMC as well as the marginal conditional (on data) distributions of $\bm{\beta}$'s obtained from the random-weighting methods and the parametric residual bootstrap. For most of the coefficients, there is very good agreement among the methods. One notable feature is that the parametric residual bootstrap approach induces the least sparsity among all five methods for variables \texttt{indus} and \texttt{age}. In addition, Bayesian LASSO appears to introduce slightly more sparsity than the random-weighting schemes for the variable \texttt{age}. Besides that, random-weighting with different penalty weights (\ref{eq.Bnw.WeightPenalty2}) appears to produce lower outliers for variables \texttt{crim}, \texttt{indus} and \texttt{ptratio}.\appendix

\section{} \label{sec_Appendix}

We present the proofs for all the theorems, proposition and corollaries in this paper. Many subsequent proofs rely on this following result.

\allowdisplaybreaks

\begin{lem} \label{lem_ASconv}
	Let $U_1, U_2, \cdots$ be any i.i.d. random variables with $\mathbb{E} (U_i) = 0$ and $\mathbb{E} [(U_i)^2] = \sigma^2 < \infty$. Then for any bounded sequence of real numbers $\{k_i\}$ and for any $\frac{1}{2} < c < 1$, 
	$$
	\dfrac{1}{n^c}
	\sumin k_i U_i \CONV{a.s.} 0.
	$$ 
\end{lem}

\begin{proof}
	Since $\{k_i\}$ are bounded, $\exists$ $M>0$ such that $|k_i| \leq M$ $\forall$ $i$. Then
	\begin{align*}
	\sum_{n=1}^{\infty} Var
	\left( \dfrac{k_n U_n}{n^c} \right)
	= \sigma^2 \sum_{n=1}^{\infty} \dfrac{k_n^2}{n^{2c}} 
	\leq \sigma^2 M^2 \sum_{n=1}^{\infty} \dfrac{1}{n^{2c}} 
	< \infty.
	\end{align*} 
	B  y Theorem 2.5.3 of \citet{durrett}, with probability one,
	$$
	\sum_{n=1}^{\infty}  \dfrac{k_n U_n}{n^c} < \infty.
	$$
	Finally, apply Kronecker's Lemma to obtain the desired result. \\
\end{proof}

\begin{lem} \label{lem_Cnw11inv}
	Assume assumptions (\ref{assume_boundedX}) and (\ref{assume_X'X_11}). Then, 
	$$
	\left\Vert
	\left( \cnwa \right)^{-1} 
	\right\Vert_2
	= O_p(1).
	$$
\end{lem} 

\begin{proof}
	Due to assumptions (\ref{assume_boundedX}) and (\ref{assume_X'X_11}) and that $q$ is fixed, $C_{n(11)}$ is invertible for all $n$. We also verify the invertibility of $\cnwa$ by recognizing that 
	$$
	\cnwa = \dfrac{1}{n} X_{(1)}' D_n X_{(1)} = 
	\dfrac{1}{n} \left( D_n^{\frac{1}{2}} X_{(1)} \right)'
	\left( D_n^{\frac{1}{2}} X_{(1)} \right) 
	$$
	where $D_n^{1/2} = diag \left(\sqrt{W_1}, \cdots, \sqrt{W_n} \right)$, which is a full-rank square matrix. Thus,
	$$
	{\rm rank} \left( \cnwa \right)
	= {\rm rank} \left( D_n^{\frac{1}{2}} X_{(1)} \right)
	= {\rm rank} \left( X_{(1)} \right)
	= q,
	$$
	i.e. $\cnwa$ is full-rank and is invertible for every $n$. Next, 
	$$
	\cnwa = C_{n(11)} + 
	\dfrac{1}{n} X'_{(1)} (D_n - \mu_W I_n) X_{(1)}
	$$
	where the Strong Law of Large Numbers ensures that
	$$
	\dfrac{1}{n} X'_{(1)} (D_n - \mu_W I_n) X_{(1)}
	\CONV{a.s.} \bm{0} 
	$$
	due to assumption (\ref{assume_boundedX}). Since $C_{n(11)}$ is invertible for all $n$, we have 
	$$
	\left\Vert
	 \left( \cnwa \right)^{-1} 
	\right\Vert_2
	= \left\Vert
		\left( C_{n(11)} + o(1) \right)^{-1}
	\right\Vert_2
	= \mathcal{O}(1) \,\, a.s.
	$$
\end{proof}

In fact, if we assume $C_{n(11)} \to C_{11}$ for some nonsingular matrix $C_{11}$ in Lemma \ref{lem_Cnw11inv}, then by the Strong Law of Large Numbers and Continuous Mapping Theorem,
	\begin{align*}
	\left( \cnwa \right)^{-1}
	\CONV{a.s.} \frac{1}{\mu_W} C_{11}^{-1}. 
	\end{align*}
 
\begin{lem} \label{lem_CtildeW}
Assume assumptions (\ref{assume_boundedX}) and (\ref{assume_X'X_11}). For any $\frac{1}{2} < c_1 < 1$, if $\exists$ $0 \leq c_3 < 2 c_1 - 1$ for which $p_n = \mathcal{O}(n^{c_3})$, then
$$
\left\Vert
n^{1-c_1} \widetilde{C}^w_n 
\right\Vert_2 
= o_p(1). 
$$
\end{lem} 

\begin{proof}
Let 
$$
H = X_{(1)} C_{n(11)}^{-1} C_{n(12)} - X_{(2)}.
$$
Then 
$$
n^{1-c_1} \widetilde{C}^w_n 
= \dfrac{1}{n^{c_1}} H'(\mu_W I_n - D_n) X_{(1)} 
\left( \cnwa \right)^{-1}.
$$
Due to assumptions (\ref{assume_boundedX}) and (\ref{assume_X'X_11}) and that $q$ is fixed, every element of the matrix $H$ is bounded. Let $h_{ij}$ and $x_{ij}$ be the $(i,j)^{th}$ element of $H$ and $X_{(1)}$ respectively. For $0 \leq c_3 < 2 c_1 - 1$, by Lemma \ref{lem_ASconv},  
 $$
 \dfrac{1}{n^{c_1 - \frac{c_3}{2}}}
 \sumin h_{k,i} x_{i,l}
 (W_i - \mu_W)
 \CONV{a.s.} 0
 $$
 for every $k = 1, \cdots, p_n - q$ and $l = 1, \cdots, q$. Thus, 
 \begin{align*}
 &\left\Vert
 \dfrac{1}{n^{c_1}} 
 H' (\mu_W I_n - D_n) 
 X_{(1)}
 \right\Vert_2^2 \\
 &\leq \left\Vert
 \dfrac{1}{n^{c_1}} 
 H' (\mu_W I_n - D_n) 
 X_{(1)}
 \right\Vert_F^2 \\
 &= \sum_{k=1}^{p_n-q} \sum_{l=1}^q
 \left[
 	\dfrac{1}{n^{\frac{c_3}{2}}}
 	\times
 	\dfrac{1}{n^{c_1 - \frac{c_3}{2}}}
 	\sumin h_{k,i} x_{i,l}
 	(\mu_W - W_i)
 \right]^2 \\
 &= \mathcal{O}(p_n) \times o\left(n^{-c_3}\right) 
 = o(1) \quad {a.s.} .
 \end{align*}

 Finally, by Lemma \ref{lem_Cnw11inv},
 $$
 \left\Vert
 	n^{1-c_1} \widetilde{C}^w_n 
 \right\Vert_2 
 \leq 
 \left\Vert
 	\dfrac{1}{n^{c_1}} 
 	H' (\mu_W I_n - D_n) 
 	X_{(1)}
 \right\Vert_2
 \left\Vert
 	\left( \cnwa \right)^{-1} 
 \right\Vert_2
 = o_p(1). 
 $$
\end{proof}	

\begin{lem} \label{lem_X'DnX}
	Suppose that $p_n = p$ is fixed. Assume (\ref{assume_boundedX}) and (\ref{assume_X'X}). Then, as $n \to \infty$, 
	$$
	\dfrac{\mu_W}{n} X' D_n X
	\CONV{a.s.} \mu_W C. 
	$$
\end{lem} 

\begin{proof}
	Due to assumption (\ref{assume_boundedX}), the Strong Law of Large Numbers gives 
	$$
	\dfrac{1}{n} X' (D_n -\mu_W I_n) X 
	= \dfrac{1}{n} \sumin (W_i - \mu_W) \bm{x}_i \bm{x}_i' 
	\CONV{a.s.} \bm{0},
	$$ 
	where $\bm{x}_i$ is the $i^{th}$ row of $X$. Then, due to assumption (\ref{assume_X'X}), 
	$$
	\dfrac{1}{n} X' D_n X 
	= \dfrac{1}{n} X' (D_n - \mu_W I_n) X 
	+ \dfrac{\mu_W}{n} X'X
	\CONV{a.s.} \bm{0} + \mu_W C = \mu_W C.
	$$ 
\end{proof}

An immediate consequence of Lemma \ref{lem_X'DnX} is that when $p$ is fixed,
$$
C_{n(ij)}^w \CONV{a.s.} \mu_W C_{ij} \quad \forall \,\, i,j = 1,2.
$$
\\
We remind readers that in this paper, we consider a common probability space $P = P_D \times P_W$, which correspond to the two sources of randomness $(\bm{\epsilon}, \bm{W})$. Note that the product probability space highlights the fact that the random weights $\bm{W}$ are drawn independently from the data $D$. The rest of the proofs deals with convergence of conditional probabilities/distributions (given data, i.e. given $\mathcal{F}_n$) for expressions containing $\bm{\epsilon}$, where the convergence takes place almost surely under $P_D$ (i.e. for almost every data set). See \citet{mason1992rank}  for relevant background.
%the underlying concept and techniques of proofs regarding convergence of conditional probabilities/distributions a.s.-$P_D$ in a random-weighting or weighted bootstrap setup.    

\begin{lem} \label{lem_EpsDnEps}
	Assume (\ref{assume_4thmoment}). Then
	$$
	\dfrac{\bm{\epsilon}' D_n \bm{\epsilon}}{n}
	\CONV{c.p.} \mu_W \sigma^2_{\epsilon}
	\quad a.s. \,\, P_D.
	$$
\end{lem} 

\begin{proof}
	Clearly, 
	$$
	\dfrac{1}{n} \sumin\epsilon_i^2
	\to \sigma^2_{\epsilon} \quad a.s. \,\, P_D.
	$$
	Due to assumption~(\ref{assume_4thmoment}),
	$$
	\dfrac{1}{n} \sumin \epsilon_i^4 
	= \mathcal{O}(1) 
	\quad a.s. \,\, P_D,
	$$
	which leads to
	$$
	\dfrac{1}{n^2} \sumin  \mathbb{E} (\epsilon_i^4 W_i ^2 \big| \mathcal{F}_n)
	= \dfrac{1}{n^2} \sumin \epsilon_i^4 \mathbb{E} (W_i ^2)
	= \dfrac{\sigma^2_W + \mu_W^2}{n} \left( \dfrac{1}{n} \sumin \epsilon_i^4 \right)
	\to 0 \,\, a.s. \,\, P_D. 
	$$
	Hence, by the Weak Law of Large Numbers (e.g., Theorem 1.14(ii) of \citet{junshao}),
	$$
	\dfrac{1}{n} \bm{\epsilon}' (D_n - \mu_W I_n) \bm{\epsilon}
	= \dfrac{1}{n} \sumin \epsilon_i^2 (W_i - \mu_W)
	\CONV{c.p.} 0 \quad a.s. \,\, P_D,
	$$ 
	and thus,
	$$
	\dfrac{\bm{\epsilon}' D_n \bm{\epsilon}}{n}
	= \dfrac{1}{n} \sumin \epsilon_i^2 (W_i - \mu_W)
	+ \dfrac{\mu_W}{n} \sumin\epsilon_i^2
	\CONV{c.p.} 0 + \mu_W \sigma^2_{\epsilon}
	= \mu_W \sigma^2_{\epsilon} \quad a.s. \,\, P_D.
	$$ 
\end{proof}

%Another proof for Lemma \ref{lem_EpsDnEps} goes like this: For any $\xi > 0$, by Chebyshev's inequality, assumption (\ref{assume_4thmoment}) ensures that
%\begin{align*}
%P \left(
%	\left|
%		\frac{1}{n} \bm{\epsilon}' (D_n - \mu_W I_n) \bm{\epsilon} 
%	\right| > \xi
%\Bigg| \mathcal{F}_n
%\right) 
%&\leq \dfrac{1}{\xi^2} Var
%	\left(
%		\frac{1}{n} \sumin \epsilon_i^2 W_i
%	\Bigg| \mathcal{F}_n
%\right) \\
%&= \dfrac{\sigma^2_W}{\xi^2 n^2} 
%\sumin \epsilon_i^4 \\
%&\to 0 \quad a.s. \,\, P_D.
%\end{align*}
%Thus, we have
%$$
%\frac{1}{n} \bm{\epsilon}' (D_n - \mu_W I_n) \bm{\epsilon}  
%\CONV{c.p.} \bm{0} \quad a.s. \,\, P_D,
%$$
%whereas 
%$$
%\dfrac{\mu_W}{n} \bm{\epsilon}' \bm{\epsilon}
%\to \mu_W \sigma^2_\epsilon 
%\quad a.s. \,\, P_D,
%$$
%and the result follows from Slutsky's Theorem.

\begin{lem} \label{lem_Znw1}
	Assume (\ref{assume_4thmoment}), (\ref{assume_boundedX}) and (\ref{assume_X'X_11}). Then for any $c>0$,
	$$
	\dfrac{1}{n^c} \znwa
	= o_p(1) \quad a.s. \,\, P_D.
	$$
\end{lem}

\begin{proof}
	Let $x_{ij}$ be the $(i,j)^{th}$ element of $X_{(1)}$. Then, we can rewrite
	\begin{align*}
	\left(
		\dfrac{1}{n^c}
		\left\Vert 
			\znwa
		\right\Vert_2
	\right)^2
	&= \dfrac{1}{n^{2c}}
	\sum_{j=1}^q 
	\left(
		\dfrac{1}{\sqrt{n}} \sumin \epsilon_i x_{ji} (W_i - \mu_W) 
		+ \dfrac{\mu_W}{\sqrt{n}} \sumin \epsilon_i x_{ji} 
	\right)^2 \\
	&= \sum_{j=1}^q 
	\left(
		\dfrac{1}{n^{\frac{1}{2}+c}} \sumin \epsilon_i x_{ji} (W_i - \mu_W) 
		+ \dfrac{\mu_W}{n^{\frac{1}{2}+c}} \sumin \epsilon_i x_{ji} 
	\right)^2,
	\end{align*}
	where we note that
	$$
	\mathbb{E}
	\left(
		\sumin \epsilon_i x_{ji} W_i  
	\bigg| \mathcal{F}_n
	\right)
	= \sumin \epsilon_i x_{ji} \mathbb{E} (W_i)
	= \mu_W \sumin \epsilon_i x_{ji}, 
	$$
	and
	$$
	Var
	\left(
		\sumin \epsilon_i x_{ji} W_i 
		\bigg| \mathcal{F}_n
	\right)
	= \sumin \epsilon_i^2 x_{ji}^2 Var(W_i) 
	= \sigma^2_W \sumin \epsilon_i^2 x_{ji}^2 .
	$$
	Now, due to assumption (\ref{assume_boundedX}), 
	$$
	\dfrac{1}{n} \sumin \epsilon_i^2 x_{ji}^2 
	= \mathcal{O}(1) \,\,\, a.s. \,\, P_D 
	\,\,\, \Longrightarrow \,\,\,
	\sumin \epsilon_i^2 x_{ji}^2 
	= \mathcal{O}(n) \,\,\, a.s. \,\, P_D,
	$$
	and coupled with assumption (\ref{assume_4thmoment}),
	$$
	\dfrac{1}{n} \sumin \epsilon_i^4 x_{ji}^4 
	= \mathcal{O}(1) \,\,\, a.s. \,\, P_D 
	\,\,\, \Longrightarrow \,\,\,
	\sumin \epsilon_i^4 x_{ji}^4 
	= \mathcal{O}(n) \,\,\, a.s. \,\, P_D.
	$$
	Thus, by using assumptions (\ref{assume_4thmoment}) and (\ref{assume_boundedX}) and that $F_W$ has finite fourth moment, the Liapounov's sufficient condition is satisfied
	\begin{align*}
	&\left[ 
		\sumin \epsilon_i^2 x_{ji}^2 Var(W_i) 
	\right]^{-2}
	\left[
		\sumin \epsilon_i^4 x_{ji}^4
		\mathbb{E} (W_i - \mu_W)^4
	\right] \\
	&= \mathcal{O} \left( n^{-2} \right) 
	\times
	\mathcal{O} \left( n \right) 
	= \mathcal{O} \left( n^{-1} \right)
	\quad a.s. \,\, P_D, 
	\end{align*}
	in order to deploy the Lindeberg's Central Limit Theorem
	$$
	\dfrac{ \sumin \epsilon_i x_{ji} (W_i - \mu_W) }
	{ \sqrt{ \sigma_W^2 \sumin \epsilon_i^2 x_{ji}^2 } }
	\CONV{c.d.} N(0,1) \quad a.s. \,\, P_D. 
	$$
	Subsequently, for all $j = 1, \cdots, q$, 
	\begin{align*}
	&\dfrac{1}{\sqrt{n}} \sumin \epsilon_i x_{ji} (W_i - \mu_W) \\
	&= \sqrt{\dfrac{\sigma^2_W}{n} \sumin \epsilon_i^2 x_{ji}^2} 
	\times \dfrac{ \sumin \epsilon_i x_{ji} (W_i - \mu_W) }
	{ \sqrt{ \sigma_W^2 \sumin \epsilon_i^2 x_{ji}^2 } } \\
	&=\mathcal{O}_p(1) \quad a.s. \,\, P_D,
	\end{align*}
	and hence,
	$$
	\dfrac{1}{n^{\frac{1}{2}+c}} \sumin \epsilon_i x_{ji} (W_i - \mu_W) 
	= o_p(1) \quad a.s. \,\, P_D.
	$$
	Finally, by assumption (\ref{assume_boundedX}) and Lemma \ref{lem_ASconv},
	$$
	\dfrac{\mu_W}{n^{\frac{1}{2}+c}} \sumin \epsilon_i x_{ji}
	\to 0 \quad a.s. \,\, P_D 
	$$
	for all $j = 1, \cdots, q$. Since $q$ is fixed,
	$$
	\left(
		\dfrac{1}{n^c}
		\left\Vert 
			\znwa
		\right\Vert_2
	\right)^2
	= o_p(1) \quad a.s. \,\, P_D,
	$$ 
	and the result follows.

\end{proof}

If we assume that $C_{n(11)} \to C_{11}$ for some nonsingular matrix $C_{11}$ in Lemma \ref{lem_Znw1}, notations could be simplified in the preceding proof by using Cramer-Wold device. We point out to readers that the $C_{n(11)} \to C_{11}$ assumption is required in Theorem \ref{thm_cond_oracle} but not in Theorem \ref{thm_Model_Select}. The following proof contains some interim results that will be utilized in the proof of Theorem \ref{thm_cond_oracle}. 

Specifically, let $\bm{x}_{i(1)}$ be the $i^{th}$ row of $X_{(1)}$. Then, for every $\bm{z} \in \mathbb{R}^q$,    
\begin{align*}
	&\bm{z}' \left[ \dfrac{1}{\sqrt{n}} X'_{(1)} (D_n - \mu_W I_n) \bm{\epsilon} \right] \\
	&= \dfrac{1}{\sqrt{n}} \sumin \epsilon_i (W_i - \mu_W) \bm{z}' \bm{x}_{i(1)} \\
	&= \sqrt{ \dfrac{\sigma_W^2}{n} \sumin \epsilon_i^2 \left( \bm{z}' \bm{x}_{i(1)} \right)^2 }
	\times \dfrac{ \sumin \epsilon_i (W_i - \mu_W) \bm{z}' \bm{x}_{i(1)} }
	{ \sqrt{ \sigma_W^2 \sumin \epsilon_i^2 \left( \bm{z}' \bm{x}_{i(1)} \right)^2 } },
\end{align*}
where we note that
	$$
	\mathbb{E}
	\left(
		\sumin \epsilon_i W_i (\bm{z}' \bm{x}_{i(1)}) 
		\bigg| \mathcal{F}_n
	\right)
	= \sumin \epsilon_i (\bm{z}' \bm{x}_{i(1)}) \mathbb{E} (W_i)
	= \mu_W \sumin \epsilon_i (\bm{z}' \bm{x}_{i(1)}), 
	$$
and
	$$
	Var
	\left(
		\sumin \epsilon_i W_i (\bm{z}' \bm{x}_{i(1)}) 
		\bigg| \mathcal{F}_n
	\right)
	= \sumin \epsilon_i^2 (\bm{z}' \bm{x}_{i(1)})^2 Var(W_i) 
	= \sigma^2_W \sumin \epsilon_i^2 (\bm{z}' \bm{x}_{i(1)})^2 .
	$$
Now,
\begin{align*}
\dfrac{1}{n} \sumin \epsilon_i^2 \left( \bm{z}' \bm{x}_{i(1)} \right)^2 
&=  \bm{z}' \left( \dfrac{1}{n}\sumin \epsilon_i^2 \bm{x}_{i(1)} \bm{x}_{i(1)}' \right) \bm{z} \\
&= \bm{z}' \left(
	\sigma^2_\epsilon C_{n(11)} +
	\dfrac{1}{n}\sumin \left( \epsilon_i^2 - \sigma^2_\epsilon \right) \bm{x}_{i(1)} \bm{x}_{i(1)}'
\right) \bm{z} \\
&\to \bm{z}' \left( \sigma^2_\epsilon C_{11} \right) \bm{z} 
\quad a.s. \,\, P_D
\end{align*}
due to assumption (\ref{assume_boundedX}) and the Strong Law of Large Numbers. Thus,
$$
\sumin \epsilon_i^2 \left( \bm{z}' \bm{x}_{i(1)} \right)^2 
= \mathcal{O}(n) \quad a.s. \,\, P_D. 
$$
In addition, by assumptions (\ref{assume_4thmoment}) and (\ref{assume_boundedX}),
\begin{align*}
\dfrac{1}{n} \sumin \epsilon_i^4 \left( \bm{z}' \bm{x}_{i(1)} \right)^4 
&\leq (q M_1 \Vert \bm{z} \Vert_2)^4 
\left( \dfrac{1}{n} \sumin \epsilon_i^4  \right)
= \mathcal{O}(1) \quad a.s. \,\, P_D,
\end{align*}
which implies
$$
\sumin \epsilon_i^4 \left( \bm{z}' \bm{x}_{i(1)} \right)^4
= \mathcal{O}(n) \quad a.s. \,\, P_D.
$$
Therefore, by using assumptions (\ref{assume_4thmoment}) and (\ref{assume_boundedX}) and that $F_W$ has finite fourth moment, we could verify the Liapounov's sufficient condition
\begin{align*}
&\left[ 
	\sumin \epsilon_i^2 \left( \bm{z}'\bm{x}_{i(1)} \right)^2 Var(W_i) 
\right]^{-2}
\left[
	\sumin \epsilon_i^4 
	\left( \bm{z}'\bm{x}_{i(1)} \right)^4
\mathbb{E} (W_i - \mu_W)^4
\right] \\
&= \mathcal{O} \left( n^{-2} \right) 
\times
\mathcal{O} \left( n \right) 
= \mathcal{O} \left( n^{-1} \right)
\quad a.s. \,\, P_D, 
\end{align*}
in order to deploy the Lindeberg's Central Limit Theorem
$$
\dfrac{ \sumin \epsilon_i (W_i - \mu_W) \bm{z}' \bm{x}_{i(1)} }
{ \sqrt{ \sigma_W^2 \sumin \epsilon_i^2 \left( \bm{z}' \bm{x}_{i(1)} \right)^2 } }
\CONV{c.d.} N(0,1) \quad a.s. \,\, P_D. 
$$
Then, by Slutsky's Theorem, for every $\bm{z} \in \mathbb{R}^q$, 
$$
\bm{z}' \left[ \dfrac{1}{\sqrt{n}} X'_{(1)} (D_n - \mu_W I_n) \bm{\epsilon} \right] 
\CONV{c.d.} N 
\left(
	0 \,\, , \,\,
	\bm{z}' \left( \sigma^2_W \sigma^2_\epsilon C_{11} \right) \bm{z} 
\right).
$$
and by Cramer-Wold device,
$$
\dfrac{1}{\sqrt{n}} X'_{(1)} (D_n - \mu_W I_n) \bm{\epsilon} 
\CONV{c.d.} N_q 
\left(
\bm{0} \,\, , \,\,
\sigma^2_W \sigma^2_\epsilon C_{11} 
\right),
$$
Since assumption (\ref{assume_boundedX}) and Lemma \ref{lem_ASconv} ensure that for any $c >0$,
	$$
	\dfrac{1}{n^{\frac{1}{2}+c}} X'_{(1)} \bm{\epsilon}
	\to \bm{0} \quad a.s. \,\, P_D, 
	$$
we finally have
	$$
	\dfrac{1}{n^c} \znwa
	= \dfrac{1}{n^c} \left[ \dfrac{1}{\sqrt{n}} X'_{(1)} (D_n - \mu_W I_n) \bm{\epsilon} \right]
	+ \dfrac{\mu_W}{n^{\frac{1}{2}+c}} X'_{(1)} \bm{\epsilon}
	= o_p(1) \quad a.s. \,\, P_D.
	$$
\\

\begin{lem} \label{lem_Znw3}
	Assume (\ref{assume_4thmoment}), (\ref{assume_boundedX}) and (\ref{assume_X'X_11}).
	\begin{itemize}
		\item [ (a) ] If there exists $\frac{1}{2} < c_1 < c_2 < 1.5 - c_1$ and $0 \leq c_3 < 2(c_2 - c_1)$ for which $p_n = \mathcal{O}(n^{c_3})$, then
		$$
		\dfrac{1}{n^{c_2 - \frac{1}{2}}}
		\left\Vert \znwc \right\Vert_2
		= o_p(1) \quad a.s. \,\, P_D.
		$$
		\item [ (b) ] If there exists $\frac{1}{2} < c_1 < c_2 < 1.5 - c_1$ and $0 \leq c_3 < \frac{2}{3}(c_2 - c_1)$ for which $p_n = \mathcal{O}(n^{c_3})$, then
		$$
		\dfrac{p_n-q}{n^{c_2 - \frac{1}{2}}}
		\left\Vert \znwc \right\Vert_2
		= o_p(1) \quad a.s. \,\, P_D.
		$$
	\end{itemize} 
\end{lem}

\begin{proof}
	Let 
	$$
	H = X_{(1)} C_{n(11)}^{-1} C_{n(12)} - X_{(2)}.
	$$
	Then 
	$$
	\znwc = \dfrac{1}{\sqrt{n}} H' D_n \bm{\epsilon}.
	$$
	Due to assumptions (\ref{assume_boundedX}) and (\ref{assume_X'X_11}) and that $q$ is fixed, every element of the matrix $H$ is bounded. Let $h_{ij}$ be the $(i,j)^{th}$ element of $H$. Then, for all $j = 1, \cdots, p_n-q$,
	$$
	\dfrac{1}{n} \sumin h_{ji}^2 \epsilon_i^2 
	= O(1) \,\,\, a.s. \,\, P_D
	\,\, \Longrightarrow \,\,  
	\sumin h_{ji}^2 \epsilon_i^2 
	= O(n) \,\,\, a.s. \,\, P_D,
	$$ 
	and
	$$
	\dfrac{1}{n} \sumin h_{ji}^4 \epsilon_i^4 
	= O(1) \,\,\, a.s. \,\, P_D
	\,\, \Longrightarrow \,\, 
	\sumin h_{ji}^4 \epsilon_i^4  
	= O(n) \,\,\, a.s. \,\, P_D  
	$$
	due to assumption (\ref{assume_4thmoment}). Next, we note that
	$$
	\mathbb{E}
	\left(
	\sumin h_{ji} \epsilon_i W_i
	\bigg| \mathcal{F}_n
	\right)
	= \sumin h_{ji} \epsilon_i \mathbb{E} (W_i)
	= \mu_W \sumin h_{ji} \epsilon_i , 
	$$
	and
	$$
	Var
	\left(
	\sumin h_{ji} \epsilon_i W_i
	\bigg| \mathcal{F}_n
	\right)
	= \sumin h_{ji}^2 \epsilon_i^2 Var(W_i) 
	= \sigma^2_W \sumin h_{ji}^2 \epsilon_i^2.
	$$
	By using assumptions (\ref{assume_4thmoment}) and (\ref{assume_boundedX}) and that $F_W$ has finite fourth moment, we could verify the Liapounov's sufficient condition
	\begin{align*}
	&\left[ 
		\sumin h_{ji}^2 \epsilon_i^2 Var(W_i) 
	\right]^{-2}
	\left[
		\sumin h_{ji}^4 \epsilon_i^4 
		\mathbb{E} (W_i - \mu_W)^4
	\right] \\
	&= \mathcal{O} \left( n^{-2} \right) 
	\times
	\mathcal{O} \left( n \right) 
	= \mathcal{O} \left( n^{-1} \right)
	\quad a.s. \,\, P_D, 
	\end{align*}
	in order to deploy the Lindeberg's Central Limit Theorem
	$$
	\dfrac{\sumin h_{ji} \epsilon_i (W_i - \mu_W) }
	{ \sqrt{ \sigma^2_W \sumin h_{ji}^2 \epsilon_i^2 } }
	\CONV{c.d.} N \left( 0, 1 \right)
	\quad a.s. \,\, P_D.
	$$
	Thus, for all $j = 1, \cdots, p_n-q$,
	\begin{align*}
	&\dfrac{1}{\sqrt{n}} \sumin h_{ji} \epsilon_i (W_i - \mu_W) \\
	&= \sqrt{ \frac{\sigma^2_W}{n} \sumin h_{ji}^2 \epsilon_i^2 } \times 
	\dfrac{\sumin h_{ji} \epsilon_i (W_i - \mu_W) }
	{ \sqrt{ \sigma^2_W \sumin h_{ji}^2 \epsilon_i^2 } }\\
	&= O_p(1) \quad a.s. \,\, P_D, 
	\end{align*}
	which leads to 
	$$
	\dfrac{1}{n^{c_1}} \sumin h_{ji} \epsilon_i (W_i - \mu_W) 
	= o_p(1) \quad a.s. \,\, P_D,
	$$
	whereas Lemma \ref{lem_ASconv} ensures that
	$$
	\dfrac{1}{n^{c_1}} \sumin h_{ji} \epsilon_i  
	\to 0 \quad a.s. \,\, P_D.
	$$
	Therefore, for part (a) of Lemma \ref{lem_Znw3},
	\begin{align*}
	&\left( 
		\dfrac{1}{n^{c_2 - \frac{1}{2}}}
		\left\Vert \znwc \right\Vert _2
	\right)^2 \\
	&=  \dfrac{1}{n^{2 c_2 - 1}}
		\left\Vert \znwc \right\Vert^2 _2 \\
	&= \dfrac{1}{n^{2 c_2 - 1}}
	\sum_{j=1}^{p_n-q}
	\left(
		\dfrac{1}{\sqrt{n}} \sumin h_{ji} \epsilon_i (W_i - \mu_W)
		+ \dfrac{1}{\sqrt{n}} \sumin h_{ji} \epsilon_i 
	\right)^2 \\
	&= \dfrac{n^{2 c_1 - 1}}{n^{2 c_2 - 1}}
	\sum_{j=1}^{p_n-q}
	\left(
		\dfrac{1}{n^{c_1}} \sumin h_{ji} \epsilon_i (W_i - \mu_W)
		+ \dfrac{1}{n^{c_1}} \sumin h_{ji} \epsilon_i 
	\right)^2 \\
	&= \mathcal{O} \left( n^{2(c_1 - c_2)} \right)
	\times o_p \left( n^{c_3} \right) \quad a.s. \,\, P_D \\
	&= o_p(1) \quad a.s. \,\, P_D 
	\end{align*}
	since $c_3 < 2(c_2 - c_1)$. \\
	
	For part (b) of Lemma \ref{lem_Znw3},
	\begin{align*}
	&\left( 
	\dfrac{p_n - q}{n^{c_2 - \frac{1}{2}}}
	\left\Vert \znwc \right\Vert _2
	\right)^2 \\
	&= \mathcal{O} \left( n^{2(c_1 - c_2 + c_3)} \right)
	\times o_p \left( n^{c_3} \right) \quad a.s. \,\, P_D \\
	&= o_p(1) \quad a.s. \,\, P_D 
	\end{align*}
	since $c_3 < \frac{2}{3}(c_2 - c_1)$.
\end{proof}

\begin{lem} \label{lem_X'DnEps}
	Assume (\ref{assume_boundedX}) and that $p_n = p$ is fixed. Then
	$$
	\dfrac{1}{n} X' D_n \bm{\epsilon}
	\CONV{c.p.} \bm{0} 
	\quad a.s. \,\, P_D.
	$$
\end{lem}

\begin{proof}
	Let $\bm{x}_i$ and $x_{ij}$ be the $i^{th}$ row and $(i,j)^{th}$ element of $X$ respectively. Due to assumption (\ref{assume_boundedX}),
	$$
	\dfrac{1}{n} X' \bm{\epsilon} \to \bm{0} 
	\quad a.s. \,\, P_D,
	$$
	and for all $j = 1, \cdots, p$,
	\begin{align*}
	&\dfrac{1}{n^2} \sumin \mathbb{E}
	\left( x^2_{ji} \epsilon_i^2 W_i^2 \Big| \mathcal{F}_n \right) \\
	&= \dfrac{1}{n^2} \sumin x^2_{ji} \epsilon_i^2 \mathbb{E} (W_i^2) \\
	&\leq \dfrac{M_1^2 ( \sigma^2_W + \mu_W^2 ) }{n} \left(
	\dfrac{1}{n} \sumin \epsilon_i^2
	\right) \\
	&\to 0 \quad a.s. \,\, P_D.
	\end{align*}
	Hence, by the Weak Law of Large Numbers (e.g., Theorem 1.14(ii) of \citet{junshao}),
	$$
	\dfrac{1}{n} X' (D_n - \mu_W I_n) \bm{\epsilon}
	= \dfrac{1}{n} \sumin \epsilon_i (W_i  - \mu_W) \bm{x}_i 
	\CONV{c.p.} \bm{0} \quad a.s. \,\, P_D.
	$$
%	For all $j = 1, \cdots, p$ and for any $\xi > 0$, by Chebyshev's inequality,
%	\begin{align*}
%	P \left(
%		\left\vert
%			\dfrac{1}{n} \sumin \epsilon_i x_{ji} (W_i - \mu_W)
%		\right\vert
%		> \xi \Bigg| \mathcal{F}_n
%	\right) 
%	&\leq \dfrac{1}{\xi^2} Var
%	\left(
%		\dfrac{1}{n} \sumin \epsilon_i x_{ji} W_i 
%		\Bigg| \mathcal{F}_n
%	\right) \\
%	&= \dfrac{\sigma^2_W}{\xi^2 n^2} 
%	\sumin \epsilon_i^2 x^2_{ji} \\
%	&\to 0 \quad a.s. \,\, P_D.
%	\end{align*}
	Finally, 
	$$
	\dfrac{X' D_n \bm{\epsilon}}{n} 
	= \dfrac{1}{n} X' (D_n - \mu_W I_n) \bm{\epsilon} 
	+ \dfrac{\mu_W}{n} X' \bm{\epsilon}
	\CONV{c.p.} \bm{0} \quad a.s. \,\, P_D.  
	$$
\end{proof}

\begin{lem} \label{lem_X'DnResid}
	Suppose that $p_n = p$ is fixed. Assume (\ref{assume_4thmoment}), (\ref{assume_boundedX}), (\ref{assume_X'X}), and
	$$
	\dfrac{1}{\sqrt{n}} X' \bm{e}_n \to \bm{0} 
	\quad a.s. \,\, P_D,
	$$
	where $\bm{e}_n$ is the residual of the strongly consistent estimator $\bSC$ of the linear model (\ref{eq.LinearModel}). Then,
	$$
	\dfrac{1}{\sqrt{n}} X' D_n \bm{e}_n \,
	\CONV{c.d.} N_p \left( \bm{0}, \sigma^2_W \sigma^2_\epsilon C \right)
	\quad a.s. \,\, P_D.
	$$ 
\end{lem}

\begin{proof}
	Due to assumption (\ref{assume_X'X}), 
	$$
	\dfrac{\sigma^2_\epsilon}{n} X'X 
	\to \sigma^2_\epsilon C.
	$$
	Since $\bSC$ is a strongly consistent estimator of $\bm{\beta}$ in (\ref{eq.LinearModel}), we have
	$$
	\left(
	\bSC - \bm{\beta}_0 
	\right)
	\to \bm{0} \quad a.s. \,\, P_D.  
	$$
	Let $\bm{x}_i$ be the $i^{th}$ row of $X$, and let $e_i$ be the $i^{th}$ element of $\bm{e}_n$. Due to assumption (\ref{assume_boundedX}) and Lemma \ref{lem_ASconv} and the fact that $\bSC$ is strongly consistent, 
	\begin{alignat*}{2}
	& &&\dfrac{1}{n} \sumin (e_i^2 - \sigma^2_{\epsilon}) \bm{x}_i \bm{x}'_i \\
	&= &&\dfrac{1}{n} \sumin
	\left(
	\left[
	\bm{x}'_i 
	\left( \bm{\beta}_0 - \bSC \right)
	+ \epsilon_i
	\right]^2
	- \sigma^2_{\epsilon}
	\right)
	\bm{x}_i \bm{x}'_i \\
	&= &&\dfrac{1}{n} \sumin (\epsilon_i^2 - \sigma^2_{\epsilon}) \bm{x}_i \bm{x}'_i \\
	& &&+ \dfrac{2}{n} \sumin \epsilon_i 
	\left[ \bm{x}'_i \left( \bm{\beta}_0 - \bSC \right) \right] 
	\bm{x}_i \bm{x}'_i \\
	& &&+ \dfrac{1}{n} \sumin
	\left[ \bm{x}'_i \left( \bm{\beta}_0 - \bSC \right) \right]^2 
	\bm{x}_i \bm{x}'_i \\ 
	&\to &&\,\, \bm{0} \quad a.s. \,\, P_D, 
	\end{alignat*}
	which leads to
	\begin{align} \label{eq:X'diag(e)X}
	\dfrac{1}{n} \sumin e_i^2 \bm{x}_i \bm{x}_i' 
	= \dfrac{1}{n} \sumin (e_i^2 - \sigma^2_{\epsilon}) \bm{x}_i \bm{x}_i' 
	+ \dfrac{\sigma^2_{\epsilon} }{n} X'X
	\to \sigma^2_{\epsilon} C 
	\quad a.s. \,\, P_D. 
	\end{align}
	Now for every $\bm{z} \in \mathbb{R}^p$, consider
	\begin{align*}
	&\bm{z}' \left[
		\dfrac{1}{\sqrt{n}} X' (D_n - \mu_W I_n) \bm{e}_n
	\right] \\ 
	&= \dfrac{1}{\sqrt{n}} \sumin 
		e_i (W_i - \mu_W) (\bm{z}' \bm{x}_i) \\
	&= \sqrt{ \dfrac{\sigma^2_W}{n} \sumin e_i^2 (\bm{z}' \bm{x}_i)^2 }
	\times \dfrac{ \sumin e_i (W_i - \mu_W) (\bm{z}' \bm{x}_i) }
	{ \sqrt{ \sigma^2_W \sumin e_i^2 (\bm{z}' \bm{x}_i)^2 } }.
	\end{align*}
	We verify that 
	$$
	\mathbb{E} \left\{
		\sumin e_i W_i (\bm{z}' \bm{x}_i)
		\bigg| \mathcal{F}_n 
	\right\}
	= \mu_W \sumin e_i  (\bm{z}' \bm{x}_i),
	$$
	and
	$$
	Var \left(
		\sumin e_i W_i (\bm{z}' \bm{x}_i)
		\bigg| \mathcal{F}_n 
	\right)
	= \sigma^2_W \sumin e_i^2  (\bm{z}' \bm{x}_i)^2.
	$$
	From (\ref{eq:X'diag(e)X}), we have
	\begin{align*}
	\dfrac{1}{n} \sumin e_i^2 (\bm{z}' \bm{x}_i)^2 
	=  \bm{z}' \left( 
		\dfrac{1}{n} \sumin e_i^2 \bm{x}_i \bm{x}_i'
	\right) \bm{z} 
	\to \bm{z}' \left(
		\sigma^2_\epsilon C
 	\right) \bm{z} 
 	\quad a.s. \,\, P_D,
	\end{align*}
	and thus
	$$
	\sumin e_i^2 (\bm{z}' \bm{x}_i)^2 
	= \mathcal{O}(n) \quad a.s. \,\, P_D.
	$$
	Due to assumptions (\ref{assume_4thmoment}) and (\ref{assume_boundedX}) and the fact that $\bSC$ is strongly consistent,
	\begin{align*}
	&\dfrac{1}{n} \sumin e_i^4 (\bm{z}' \bm{x}_i)^4 \\
	&\leq (p M_1 \Vert \bm{z} \Vert_2)^4 
	\times \left( \dfrac{1}{n} \sumin e_i^4 \right) \\
	&= (p M_1 \Vert \bm{z} \Vert_2)^4 \times
	\left( 
		\dfrac{1}{n} \sumin 
		\left[ 
			\epsilon_i - \bm{x}_i' \left( \bSC - \bm{\beta}_0 \right)  
		\right]^4 
	\right) \\
	&\leq (p M_1 \Vert \bm{z} \Vert_2)^4 \times
	\left[
		\dfrac{1}{n} \sumin 
		\left(
			|\epsilon_i| + p M_1 \left\Vert \bSC - \bm{\beta}_0 \right\Vert_2  
		\right)^4 
	\right] \\
	&= \mathcal{O}(1) \quad a.s. \,\, P_D,
	\end{align*}
	and thus
	$$
	\sumin e_i^4 (\bm{z}' \bm{x}_i)^4 
	= \mathcal{O}(n) \quad a.s. \,\, P_D.
	$$
	Since the i.i.d. random weights are drawn from $F_W$ which has finite fourth moment, the Liapounov's sufficient condition is satisfied 
	\begin{align*}
	&\left[ 
	\sumin e_i^2 (\bm{z}' \bm{x}_i)^2 Var(W_i) 
	\right]^{-2}
	\left[
	\sumin e_i^4 (\bm{z}' \bm{x}_i)^4 \mathbb{E} (W_i - \mu_W)^4
	\right] \\
	&= \mathcal{O} \left( n^{-2} \right)
	\times 
	\mathcal{O} \left( n \right) \\
	&= \mathcal{O} \left( n^{-1} \right) 
	\quad a.s. \,\, P_D
	\end{align*}
	in order to deploy the Lindeberg's Central Limit Theorem
	$$
	\dfrac{ \sumin e_i (W_i - \mu_W) (\bm{z}' \bm{x}_i) }
	{ \sqrt{ \sigma^2_W \sumin e_i^2 (\bm{z}' \bm{x}_i)^2 } }
	\CONV{c.d.} N(0,1) 
	\quad a.s. \,\, P_D.
	$$
	By Slutsky's Theorem, for every $\bm{z} \in \mathbb{R}^p$, 
	$$
	\bm{z}' \left[
		\dfrac{1}{\sqrt{n}} X' (D_n - \mu_W I_n) \bm{e}_n
	\right]
	\CONV{c.d.} N \left( 0 \,\, , \,\, \bm{z}' \left( \sigma^2_W \sigma^2_\epsilon C \right)  \bm{z} \right) 
	\quad a.s. \,\, P_D,
	$$
	and by Cramer-Wold device,
	$$
	\dfrac{1}{\sqrt{n}} X' (D_n - \mu_W I_n) \bm{e}_n
	\CONV{c.d.} N_p \left( \bm{0} \,\, , \,\, \sigma^2_W \sigma^2_\epsilon C \right) 
	\quad a.s. \,\, P_D.
	$$
	Finally,
	$$
	\dfrac{1}{\sqrt{n}} X' D_n \bm{e}_n
	\CONV{c.d.} N_p \left( \bm{0} \,\, , \,\, \sigma^2_W \sigma^2_\epsilon C \right) 
	\quad a.s. \,\, P_D
	$$
	since by assumption (\ref{assume:X'e}),
	$$
	\dfrac{\mu_W}{\sqrt{n}} X' \bm{e}_n \to \bm{0}
	\quad a.s. \,\, P_D. 
	$$

\end{proof}

We are now ready to prove the main results presented in the main text. The proof of Proposition \ref{lem_ModelSelect} is similar to that of Proposition 1 of \citet{BinYu}.

\begin{proof}[Proof of Proposition \ref{lem_ModelSelect}]
	First, we note that since rank($X$) = $p_n$, where $p_n \leq n$, the solution to (\ref{eq.Bnw.setup}) is unique by \citet{Osborne2000} and \citet{LassoUnique}. 
	We begin with weighting scheme (\ref{eq.Bnw.WeightPenalty2}).  Results for the  other  two  simpler  weighting  schemes  could  then  be easily inferred.   
	 \begin{alignat*}{2}
	 \bnw &= \argmin_{\bm{\beta}}
	 \Bigg\{
	 &&\dfrac{1}{n} (Y - X \bm{\beta})' D_n (Y - X \bm{\beta})  
	 + \dfrac{\lambda_n}{n} 
	 \sum_{j=1}^{p_n} W_{0,j} |\beta_{j}|
	 \Bigg\} \\
	 &= \argmin_{\bm{\beta}}
	 \Bigg\{
	 &&\dfrac{1}{n} 
	 [\bm{\epsilon} - X (\bm{\beta} - \bm{\beta}_0)]' 
	 D_n 
	 [\bm{\epsilon} - X (\bm{\beta} - \bm{\beta}_0)] \\
	 & &&+ \dfrac{\lambda_n}{n} 
	 \sum_{j=1}^{p_n} W_{0,j} | \beta_{0,j} + \beta_{j} - \beta_{0,j} |
	 \Bigg\}.  
	 \end{alignat*}
	 Therefore, 
	 \begin{alignat*}{2}
	 &(\bnw - \bm{\beta}_0) &&\\
	 &= \argmin_{\bm{u}_n}
	 \Bigg\{
	 &&\dfrac{1}{n} (\bm{\epsilon} - X \bm{u}_n)' D_n (\bm{\epsilon} - X \bm{u}_n)
	 + \dfrac{\lambda_n}{n} 
	 \sum_{j=1}^{p_n} W_{0,j} | \beta_{0,j} + u_{n,j} |
	 \Bigg\} \\
	 &= \argmin_{\bm{u}_n}
	 \Bigg\{
	 &&\bm{u}_n' \left( \dfrac{X' D_n X}{n} \right) \bm{u}_n
	 -2 \bm{u}_n' \left( \dfrac{X' D_n \bm{\epsilon}}{n} \right) 
	 + \dfrac{\bm{\epsilon}' D_n \bm{\epsilon}}{n} \\
	 & &&+ \dfrac{\lambda_n}{n} 
	 \sum_{j=1}^{p_n} W_{0,j} | \beta_{0,j} + u_{n,j} |
	 \Bigg\}. 
	 \end{alignat*}
	 The term $(\bm{\epsilon}' D_n \bm{\epsilon})/n$ could be dropped since for every $n$, it does not contain $\bm{u}_n$ and Lemma \ref{lem_EpsDnEps} ensures that it converges in conditional probability to a finite limit. Differentiating the first two terms with respect to $\bm{u}_n$ yields
	 $$
	 \dfrac{1}{n} \left\{ 2 X' D_n X \bm{u}_n - 2 X' D_n \bm{\epsilon} \right\}
	 = \dfrac{1}{n} \left\{ 2 \sqrt{n} 
	 \left[
	 C_n^w 
	 \left( \sqrt{n} \bm{u}_n \right) 
	 - 		\bm{Z}_n^w
	 \right] \right\}. 
	 $$
	 For $j = 1, \cdots, p_n$, considering sub-differentials of the penalty term with respect to $u_{n,j}$ yields
	 \begin{align*}
	 &\begin{cases}
	 \frac{\lambda_n}{n} W_{0,j} \times \text{sgn}
	 \left(\beta_{0,j} + u_{n,j} \right)
	 & \text{for $\beta_{0,j} + u_{n,j}  \neq 0$} \\
	 \frac{\lambda_n}{n}  W_{0,j} \times [-1,1]
	 & \text{for $\beta_{0,j} + u_{n,j}  = 0$} 
	 \end{cases} \\
	 &=\begin{cases}
	 \frac{\lambda_n}{n} W_{0,j} \times \text{sgn}
	 \left(\widehat{\beta}^w_{n,j} \right)
	 & \text{for $\widehat{\beta}^w_{n,j}  \neq 0$} \\
	 \frac{\lambda_n}{n}  W_{0,j} \times [-1,1]
	 & \text{for $\widehat{\beta}^w_{n,j}  = 0$} 
	 \end{cases}
	 \end{align*}
	 Note that $\bnw = \widehat{\bm{u}}_n + \bm{\beta}_0$, which can be partitioned into 
	 \[
	 \bnw = 
	 \begin{bmatrix}
	 \widehat{\bm{\beta}}^w_{n(1*)} \\
	 \widehat{\bm{\beta}}^w_{n(2*)} 
	 \end{bmatrix},
	 \]
	 where $\widehat{\bm{\beta}}^w_{n(1*)}$ consists of non-zero elements of $\bnw$, and $\widehat{\bm{\beta}}^w_{n(2*)} = \bm{0}$. The asterisk here is to distinguish the partition of random-weighting samples $\bnw$ from the true partition of $\bm{\beta}_0$. It follows that     
	 \begin{align*}
	 &2 \sqrt{n} 
	 \left[
	 C_n^w 
	 \left( \sqrt{n} \widehat{\bm{u}}_n \right) 
	 - \bm{Z}_n^w
	 \right] \\
	 &= 2 \sqrt{n} 
	 \left\{
	 \begin{bmatrix}
	 \cnwas & \cnwbs \\
	 \cnwcs & \cnwds 
	 \end{bmatrix}
	 \times \sqrt{n} 
	 \begin{bmatrix}
	 \hunas \\
	 \hunbs
	 \end{bmatrix} 
	 -
	 \begin{bmatrix}
	 \znwas \\
	 \znwbs
	 \end{bmatrix}
	 \right\}.
	 \end{align*}
	 Note that $\hunbs$ does not necessarily equal to $\bm{0}$ unless the partition of the random-weighting samples $\bnw$ coincides with the true partition of $\bm{\beta}_0$. As a consequence of the Karush-Kuhn-Tucker (KKT) conditions, we have 
	 \begin{align} \label{eq.KKT1}
	 \cnwas \left[\sqrt{n} \hunas \right] 
	 + \cnwbs \left[\sqrt{n} \hunbs \right] 
	 - \znwas
	 = - \dfrac{\lambda_n}{ 2 \sqrt{n} }
	 \bm{W}_{0(1)} \circ
	 \text{sgn} \left( \widehat{\bm{\beta}}^w_{n(1*)} \right) 
	 \end{align}
	 and
	 \begin{align} \label{eq.KKT2}
	 \left|
	 \cnwcs \left[ \sqrt{n} \hunas \right]
	 + \cnwds \left[ \sqrt{n} \hunbs \right]
	 - \znwbs
	 \right|
	 \leq \dfrac{\lambda_n}{2 \sqrt{n}} 
	 \bm{W}_{0(2)}
	 \end{align}
	 element-wise. Meanwhile, we also note that
	 \begin{align*}
	 \left\{
	 \left\vert \widehat{\bm{u}}_{n(1)} \right\vert
	 <  \left\vert \bm{\beta}_{0(1)} \right\vert
	 \right\} 
	 &=
	 \left\{
	 \widehat{\bm{u}}_{n(1)} 
	 < \left\vert \bm{\beta}_{0(1)} \right\vert
	 \right\}
	 \bigcap
	 \left\{
	 \widehat{\bm{u}}_{n(1)} 
	 > - \left\vert \bm{\beta}_{0(1)} \right\vert
	 \right\} \\
	 &=
	 \left\{
	 \widehat{\bm{\beta}}^w_{n(1)} 
	 < \bm{\beta}_{0(1)}
	 + \left\vert \bm{\beta}_{0(1)} \right\vert
	 \right\}
	 \bigcap
	 \left\{
	 \widehat{\bm{\beta}}^w_{n(1)} 
	 > \bm{\beta}_{0(1)}
	 - \left\vert \bm{\beta}_{0(1)} \right\vert
	 \right\},
	 \end{align*}
	 where all inequalities hold element-wise. Thus, $\widehat{\bm{\beta}}^w_{n(1)} < 0$ element-wise if $\bm{\beta}_{0(1)} < 0$ element-wise, and vice versa. In other words,
	 \begin{align} \label{eq.activesubset}
	 \left\{
	 \text{sgn} \left( \widehat{\bm{\beta}}^w_{n(1)} \right)
	 = \text{sgn} \left( \bm{\beta}_{0(1)} \right)
	 \right\}
	 \supseteq
	 \left\{
	 \left\vert \widehat{\bm{u}}_{n(1)} \right\vert
	 <  \left\vert \bm{\beta}_{0(1)} \right\vert
	 \text{ element-wise}
	 \right\}.
	 \end{align}
	 Therefore, by (\ref{eq.KKT1}), (\ref{eq.KKT2}), (\ref{eq.activesubset}), and uniqueness of solution for the random-weighting setup (\ref{eq.Bnw.setup}), if there exists $\widehat{\bm{u}}_n$ such that the following equation and inequalities hold:
	 	\begin{align}
	 &\cnwa \left[ \sqrt{n} \huna \right] - \znwa 
	 = - \dfrac{\lambda_n}{2 \sqrt{n}}
	 \bm{W}_{0(1)} \circ
	 \text{sgn} \left( \bm{\beta}_{0(1)} \right) \label{eq.ineq.1} \\
	 &- \dfrac{\lambda_n}{2 \sqrt{n}} \bm{W}_{0(2)} \leq
	 \cnwc \left[ \sqrt{n} \huna \right] - \znwb 
	 \leq  \dfrac{\lambda_n}{2 \sqrt{n}} \bm{W}_{0(2)}
	 \,\,\, \text{element-wise} \label{eq.ineq.2} \\
	 &\left| \huna \right| < \left| \bm{\beta}_{0(1)} \right|
	 \quad \text{element-wise} \label{eq.ineq.3} ,
	 \end{align}
	 then we have $\text{sgn} \left( \widehat{\bm{\beta}}^w_{n(1)} \right)
	 = \text{sgn} \left[ \bm{\beta}_{0(1)} \right]$
	 and
	 $\hunb = \widehat{\bm{\beta}}^w_{n(2)} = \bm{\beta}_{0(2)} = \bm{0}$,
	 ie. 
	 $$
	 \bnw \stackrel{s}{=} \bm{\beta}_0,
	 $$
	 and 
	 \begin{alignat*}{2}
	 & P \Bigg(\bnw &&\stackrel{s}{=} \bm{\beta}_0 \bigg| \mathcal{F}_n \Bigg) \\
	 &\geq
	 P \bigg(
	 &&\left\{
	 \left| \cnwc \left[ \sqrt{n} \huna \right] - \znwb \right|
	 \leq  \dfrac{\lambda_n}{2 \sqrt{n}} \bm{W}_{0(2)}
	 \,\,\, \text{{\rm element-wise}}
	 \right\} \\
	 & &&\bigcap \left\{
	 \cnwa \left[ \sqrt{n} \huna \right] - \znwa 
	 = - \dfrac{\lambda_n}{2 \sqrt{n}}
	 \bm{W}_{0(1)} \circ
	 \text{sgn} \left[ \bm{\beta}_{0(1)} \right]
	 \right\} \\
	 & &&\bigcap \left\{
	 \left| \huna \right| < \left| \bm{\beta}_{0(1)} \right|
	 \,\,\, \text{{\rm element-wise}}
	 \right\}
	 \bigg| \mathcal{F}_n
	 \bigg).  %\quad a.s. \,\, P_D.
	 \end{alignat*}
	 %\textcolor{red}{Note that the lower bound of the conditional probability $P \left( \bnw \stackrel{s}{=} \bm{\beta_0} \Big| \mathcal{F}_n \right)$ holds $a.s.$-$P_D$ because the preceding Set arguments are valid for almost every sample path $\{y_1, y_2, \cdots \}$.} 
	 Now we proceed to simplify these equation and inequalities (\ref{eq.ineq.1}), (\ref{eq.ineq.2})  and (\ref{eq.ineq.3}). 
	 Equation (\ref{eq.ineq.1}) can be re-written as 
	 \begin{align} \label{eq.rewrite.ineq1.}
	 \sqrt{n} \huna 
	 = \left( \cnwa \right)^{-1}
	 \left[
	 	\znwa - \dfrac{\lambda_n}{2 \sqrt{n}}
	 	\bm{W}_{0(1)} \circ
	 	\text{sgn} \left[ \bm{\beta}_{0(1)} \right]
	 \right].
	 \end{align}
	 Substituting inequality (\ref{eq.ineq.3}) into equation (\ref{eq.rewrite.ineq1.}) above leads to $A_n^w$. Replace the expression
	 $$
	 \bm{W}_{0(1)} \circ \text{sgn} \left[\bm{\beta}_{0(1)} \right] 
	 $$
	 in equation (\ref{eq.rewrite.ineq1.}) with $W_0 \text{sgn} \left[ \bm{\beta}_{0(1)} \right]$ and $\text{sgn} \left[ \bm{\beta}_{0(1)} \right]$ for weighting schemes (\ref{eq.Bnw.WeightPenalty}) and (\ref{eq.NoWeightPenalty}) respectively to obtain $A_n^w$ .
	  
	 Next, substituting equation (\ref{eq.rewrite.ineq1.}) into inequality (\ref{eq.ineq.2}) and simple arithmetic yield
	 \begin{equation*} 
	 \begin{split} 
	 \widetilde{B}_n^w \equiv \bigg\{
	 &\left\vert
	 \widetilde{C}^w_n \znwa 
	 + \znwc
	 - \dfrac{\lambda_n}{2 \sqrt{n}} 
	 \cnwc \left( \cnwa \right)^{-1} 
	 \bm{W}_{0(1)} \circ
	 \text{sgn} \left[ \bm{\beta}_{0(1)} \right] 
	 \right\vert \\
	 &- \dfrac{\lambda_n}{2 \sqrt{n}} 
	 \left|
	 C_{n(21)} C^{-1}_{n(11)}
	 \bm{W}_{0(1)} \circ
	 \text{sgn} \left[ \bm{\beta}_{0(1)} \right] 
	 \right|  \\
	 &\leq
	 \dfrac{\lambda_n}{2 \sqrt{n}} 
	 \left(
	 \bm{W}_{0(2)} 
	 - \left\vert
	 C_{n(21)} C^{-1}_{n(11)} 
	 \bm{W}_{0(1)} \circ
	 \text{sgn} \left[ \bm{\beta}_{0(1)} \right] 
	 \right\vert
	 \right)
	 \text{ element-wise}
	 \bigg\}
	 \end{split}
	 \end{equation*} 
	 for weighting scheme (\ref{eq.Bnw.WeightPenalty2}). Now, observe that $B_n^w \subseteq \widetilde{B}_n^w$, since (LHS of $B_n^w) \geq$ (LHS of $\widetilde{B}_n^w$) element-wise. Thus,
	 $$
	 P \left( \bnw \stackrel{s}{=} \bm{\beta}_0 \bigg| \mathcal{F}_n \right) 
	 \geq
	 P \left( 
	 A_n^w \cap \widetilde{B}_n^w 
	 \big| \mathcal{F}_n
	 \right) 
	 \geq
	 P \left( 
	 A_n^w \cap B_n^w 
	 \big| \mathcal{F}_n
	 \right).
	 %\,\, a.s. \,\, P_D.
	 $$
	 For weighting scheme (\ref{eq.Bnw.WeightPenalty}),
	 \begin{equation} \label{eq.BnwTilde}
	 \begin{split} 
	 \widetilde{B}_n^w \equiv \bigg\{
	 &\left\vert
	 \widetilde{C}^w_n \znwa 
	 + \znwc
	 - \dfrac{\lambda_n W_0}{2 \sqrt{n}} 
	 \cnwc \left( \cnwa \right)^{-1} 
	 \text{sgn} \left[ \bm{\beta}_{0(1)} \right] 
	 \right\vert \\
	 &- \dfrac{\lambda_n W_0}{2 \sqrt{n}} 
	 \left|
	 C_{n(21)} C^{-1}_{n(11)}
	 \text{sgn} \left[ \bm{\beta}_{0(1)} \right] 
	 \right|  \\
	 &\leq
	 \dfrac{\lambda_n W_0}{2 \sqrt{n}} 
	 \left(
	 \bm{1}_{p_n-q}
	 - \left\vert
	 C_{n(21)} C^{-1}_{n(11)} 
	 \text{sgn} \left[ \bm{\beta}_{0(1)} \right] 
	 \right\vert
	 \right)
	 \text{ element-wise}
	 \bigg\}. 
	 \end{split}
	 \end{equation} 
	 Now, observe that $B_n^w \subseteq \widetilde{B}_n^w$ , since (LHS of $B_n^w) \geq$ (LHS of $\widetilde{B}_n^w$) element-wise, whereas (RHS of $B_n^w) \leq$ (RHS of $\widetilde{B}_n^w$) element-wise due to the Irrepresentable condition (\ref{assume_StrongIrrepresent}). Therefore,
	 $$
	 P \left( \bnw \stackrel{s}{=} \bm{\beta}_0 \bigg| \mathcal{F}_n \right) 
	 \geq
	 P \left( 
	 A_n^w \cap \widetilde{B}_n^w 
	 \big| \mathcal{F}_n
	 \right) 
	 \geq
	 P \left( 
	 A_n^w \cap B_n^w 
	 \big| \mathcal{F}_n
	 \right).
	 %\,\, a.s. \,\, P_D.
	 $$
	 For weighting scheme (\ref{eq.NoWeightPenalty}), substitute $W_0 = 1$ in (\ref{eq.BnwTilde}) and the result follows. 
\end{proof}

\begin{proof}[Proof of Theorem \ref{thm_Model_Select}]
	From Proposition \ref{lem_ModelSelect}, 
	\begin{align*}
	P\left(
	\bnw (\lambda_n) \stackrel{s}{=} \bm{\beta}_0
	\big\vert \mathcal{F}_n
	\right)	
	&\geq P \left( 
	A_n^w \bigcap B_n^w 
	\big\vert \mathcal{F}_n
	\right) \\
	&= 1- P \left[
	\left( 
	A_n^w \bigcap B_n^w 
	\right)^c
	\bigg\vert \mathcal{F}_n
	\right] \\
	&= 1- P \left[
	\left( A_n^w \right)^c 
	\bigcup 
	\left( B_n^w  \right)^c
	\big\vert \mathcal{F}_n
	\right] \\
	&\geq 1 - \Big\{
	P \left[ 
	\left( A_n^w \right)^c
	\big\vert \mathcal{F}_n 
	\right]
	+ P \left[ 
	\left( B_n^w \right)^c
	\big\vert \mathcal{F}_n
	\right]
	\Big\}.
	\end{align*}
	We now investigate the conditional probabilities $P \left[ \left( A_n^w \right)^c \big\vert \mathcal{F}_n \right]$ and $P \left[ \left( B_n^w \right)^c \big\vert \mathcal{F}_n \right]$ separately. All three weighting schemes (\ref{eq.NoWeightPenalty}), (\ref{eq.Bnw.WeightPenalty}) and (\ref{eq.Bnw.WeightPenalty2}) share very similar $P \left[ \left( A_n^w \right)^c \big\vert \mathcal{F}_n \right]$. We start off with the most general version (\ref{eq.Bnw.WeightPenalty2}) of the weighting schemes. Results for the other two simpler weighting schemes could then be easily inferred. For ease of notation, let
	\begin{align*}
	\bm{z}_n = [z_{n,1}, \cdots, z_{n,q}]' := \left( \cnwa \right)^{-1} 
	\left(
	\znwa - 
	\dfrac{\lambda_n}{2 \sqrt{n}}
	\bm{W}_{0(1)} \circ
	\text{sgn} \left[ \bm{\beta}_{0(1)} \right]
	\right).
	\end{align*}
	Note that
	$$
	\dfrac{\lambda_n}{2n}
	\bm{W}_{0(1)} \circ
	\text{sgn} \left[ \bm{\beta}_{0(1)} \right]
	\stackrel{p}{\longrightarrow} \bm{0}.
	$$
	Hence, by Lemmas \ref{lem_Cnw11inv} and \ref{lem_Znw1},
	\begin{align*}
	P \left[
	\left( A_n^w \right)^c
	| \mathcal{F}_n
	\right] 
	&= P \left(
		\bigcup_{j=1}^q
		\Big\{
			\left\vert z_{n,j} \right\vert
			> \sqrt{n} \left\vert \beta_{0,j} \right\vert
		\Big\}
		\bigg\vert \mathcal{F}_n 
	\right) \\
	&\leq \sum_{j=1}^q P 
	\left(
		\frac{1}{\sqrt{n}}
		\left\vert z_{n,j} \right\vert
		> \left\vert \beta_{0,j} \right\vert
		\bigg\vert \mathcal{F}_n 
	\right) \\
	&\to 0 \quad a.s. \,\, P_D,
	\end{align*}
	because for all $j = 1, \cdots, q$, we have $\left\vert \beta_{0,j} \right\vert > 0$ but
	$$
	\frac{1}{\sqrt{n}}
	\left\vert z_{n,j} \right\vert = o_p(1)
	\quad a.s. \,\, P_D.
	$$ 
	For weighting schemes (\ref{eq.Bnw.WeightPenalty}) and (\ref{eq.NoWeightPenalty}), replace the expression
	$$
	\bm{W}_{0(1)} \circ \text{sgn} \left[\bm{\beta}_{0(1)} \right] 
	$$
	with $W_0 \text{sgn} \left[ \bm{\beta}_{0(1)} \right]$ and $\text{sgn} \left[ \bm{\beta}_{0(1)} \right]$ respectively to obtain the same result 
	$$
	P \left[ \left( A_n^w \right)^c | \mathcal{F}_n \right] 
	\to 0 \quad a.s. \,\, P_D.
	$$ 
	
	We now turn our attention to $P \left[ \left( B_n^w \right)^c \big\vert \mathcal{F}_n \right]$, where weighting scheme (\ref{eq.Bnw.WeightPenalty2}) is markedly different -- and derived separately -- from weighting schemes (\ref{eq.NoWeightPenalty}) and (\ref{eq.Bnw.WeightPenalty}). We first consider weighting scheme (\ref{eq.Bnw.WeightPenalty}), and then infer the result for weighting scheme (\ref{eq.NoWeightPenalty}) as a special case. For ease of notation, define
	\begin{align*}
	\bm{\zeta}_n 
	&= 
	\left[
		\zeta_{n,1}, \cdots, \zeta_{n,p_n-q}
	\right]' 
	:= \znwc, \\    
	\bm{\nu}_n
	&=\left[
		\nu_{n,1}, \cdots, \nu_{n,p_n-q}
	\right]' 
	:= \widetilde{C}^w_n 
	\left(
		\znwa - 
		\dfrac{\lambda_n W_0}{2 \sqrt{n}} 
		\text{sgn} \left[ \bm{\beta}_{0(1)} \right]
	\right).
	\end{align*}
	Then, for any $\xi > 0$,
		\begin{alignat*}{2}
	& &&P \left[ 
		\left( B_n^w \right)^c \big| \mathcal{F}_n 
	\right] \\
	&= P
	&&\left(
		\bigcup_{j=1}^{p_n - q}
		\left\{
			\left\vert
				\zeta_{n,j} + \nu_{n,j}
			\right\vert
			> \dfrac{\lambda_n}{2 \sqrt{n}} \eta_j
		\right\}
		\Bigg| \mathcal{F}_n
	\right) \\
	&\leq P
	&&\left(
		\bigcup_{j=1}^{p_n - q}
		\left\{
			\left\vert \zeta_{n,j} \right\vert
			+ \left\vert \nu_{n,j} \right\vert
			> \dfrac{\lambda_n}{2 \sqrt{n}} \eta_j
		\right\}
		\Bigg| \mathcal{F}_n
	\right) \\
	&\leq P
	&&\left(
		\bigcup_{j=1}^{p_n - q}
		\left[
			\left\{
				\left\vert \zeta_{n,j} \right\vert
				+ \left\vert \nu_{n,j} \right\vert
			> \dfrac{\lambda_n}{2 \sqrt{n}} \eta_j
			\right\}
			\bigcap
			\Big\{
				\left\vert
					\nu_{n,j}
				\right\vert
				\leq \xi
			\Big\}
		\right]
	\Bigg| \mathcal{F}_n
	\right) \\
	& &&+ P \left(
		\bigcup_{j=1}^{p_n - q}
		\left[
			\left\{
				\left\vert \zeta_{n,j} \right\vert
				+ \left\vert \nu_{n,j} \right\vert
				> \dfrac{\lambda_n}{2 \sqrt{n}} \eta_j
			\right\}
			\bigcap
			\Big\{
				\left\vert
					\nu_{n,j}
				\right\vert
				> \xi
			\Big\}
		\right]
		\Bigg| \mathcal{F}_n
	\right) \\
	&\leq P &&\left(
		\bigcup_{j=1}^{p_n - q}
		\left\{
			\left\vert \zeta_{n,j} \right\vert
			> \dfrac{\lambda_n}{2 \sqrt{n}} \eta_j - \xi
		\right\}
		\Bigg| \mathcal{F}_n
	\right)
	+ P \left(
		\bigcup_{j=1}^{p_n - q}
		\left\{
			\left\vert \nu_{n,j} \right\vert
			> \xi
		\right\}
		\Bigg| \mathcal{F}_n
	\right) \\
	&\leq && P \left(
		\bigcup_{j=1}^{p_n - q}
		\left\{
			\left\vert \zeta_{n,j} \right\vert
			> \dfrac{\lambda_n W_0}{2 \sqrt{n}} \eta_j - \xi
		\right\}
		\Bigg| \mathcal{F}_n
	\right) 
	+ P \left(
		\left\Vert \bm{\nu}_n \right\Vert_2
		> \xi
		\Big| \mathcal{F}_n
	\right).
	\end{alignat*}
	Since 
	$$
	\dfrac{\lambda_n W_0}{n^{1.5-c_1}} 
	\text{sgn} \left[ \bm{\beta}_{0(1)} \right]
	= o_p(1),
	$$
	we have, by Lemmas \ref{lem_CtildeW} and \ref{lem_Znw1},
	$$
	\left\Vert \bm{\nu}_n \right\Vert_2 
	\leq \left\Vert n^{1 - c_1} \widetilde{C}^w_n \right\Vert_2 
	\left\Vert 
	\dfrac{1}{n^{1 - c_1}} \znwa 
	- \dfrac{\lambda_n W_0}{2 n^{1.5 - c_1}} \text{sgn} \left[ \bm{\beta}_{0(1)} \right] 
	\right\Vert_2 
	= o_p(1) \quad a.s. \,\, P_D,
	$$
	and thus,
	$$
	P \left(
	\left\Vert \bm{\nu}_n \right\Vert_2
	> \xi 
	\Big| \mathcal{F}_n
	\right)
	= o(1) \quad a.s. \,\, P_D.
	$$
	Now, let
	$$
	\eta_* = \min_{1 \leq j \leq p_n-q} \eta_j,
	$$
	and note that $0 < \eta_* \leq 1$ from assumption (\ref{assume_StrongIrrepresent}). Then, 
	\begin{align*}
	&P \left(
		\bigcup_{j=1}^{p_n - q}
		\left\{
			\left\vert \zeta_{n,j} \right\vert
			> \dfrac{\lambda_n W_0}{2 \sqrt{n}} \eta_j - \xi
		\right\}
		\Bigg| \mathcal{F}_n
	\right) \\
	&\leq P \left(
		\bigcup_{j=1}^{p_n - q}
		\left\{
			\left\vert \zeta_{n,j} \right\vert
			> \dfrac{\lambda_n W_0}{2 \sqrt{n}} \eta_* - \xi
		\right\}
		\Bigg| \mathcal{F}_n
	\right) \\
	&= P \left(
		\max_{1 \leq j \leq p_n-q}
		\big\vert \zeta_{n,j} \big\vert
		> \dfrac{\lambda_n W_0}{2 \sqrt{n}} \eta_* - \xi
		\Bigg| \mathcal{F}_n
	\right) \\
	&\leq P \left(
		\big\Vert \bm{\zeta}_n \big\Vert_2
		> \dfrac{\lambda_n W_0}{2 \sqrt{n}} \eta_* - \xi
		\Bigg| \mathcal{F}_n
	\right) \\
	&= P \left(
		\dfrac{1}{n^{c_2 - \frac{1}{2}}}
		\Big(
			\big\Vert \bm{\zeta}_n \big\Vert_2
			+ \xi
		\Big)
		> \dfrac{\lambda_n W_0}{2 n^{c_2}} \eta_* 
	\Bigg| \mathcal{F}_n
	\right) \\
	&= o(1) \quad a.s. \,\, P_D,
	\end{align*}
	because
	$$
	\dfrac{\lambda_n W_0}{2 n^{c_2}} \eta_*
	= \mathcal{O}_p(1)
	$$
	whereas part (a) of Lemma \ref{lem_Znw3} ensures that 
	$$
	\dfrac{1}{n^{c_2 - \frac{1}{2}}}
	\Big(
	\big\Vert \bm{\zeta}_n \big\Vert_2
	+ \xi
	\Big)
	= o_p(1) \quad a.s. \,\, P_D.
	$$
	Thus, for weighting scheme (\ref{eq.Bnw.WeightPenalty}), we have just shown that 
	$$
	P \left[ \left( B_n^w \right)^c \big\vert \mathcal{F}_n \right]
	= o(1) \quad a.s. \,\, P_D.
	$$
	
	For weighting scheme (\ref{eq.NoWeightPenalty}), take $W_0 = 1$ and repeat the preceding steps to obtain the same result. \\
	
	Now, for weighting scheme (\ref{eq.Bnw.WeightPenalty2}), define 
	\begin{align*}
	\bm{\nu}_n
	&=\left[
		\nu_{n,1}, \cdots, \nu_{n,p_n-q}
	\right]' 
	:= \widetilde{C}^w_n 
	\left(
		\znwa - 
		\dfrac{\lambda_n}{2 \sqrt{n}} 
		\bm{W}_{0(1)} \circ
		\text{sgn} \left[ \bm{\beta}_{0(1)} \right]
	\right), \\
	\bm{\gamma}_n
	&= \left[
		\gamma_{n,1}, \cdots, \gamma_{n,p_n-q}
	\right]' 
	:= C_{n(21)} C_{n(11)}^{-1} 
	\bm{W}_{0(1)} \circ
	\text{sgn} \left[ \bm{\beta}_{0(1)} \right].
	\end{align*}
	and for any $\xi > 0$,
	\begin{align*}
	& P \left[ 
	\left( B_n^w \right)^c \big| \mathcal{F}_n 
	\right] \\
	&= P \left(
		\bigcup_{j=1}^{p_n - q}
		\left\{
			\left\vert
				\zeta_{n,j} + \nu_{n,j}
			\right\vert
			> \dfrac{\lambda_n}{2 \sqrt{n}} 
			\Big(
				W_{0(2),j} - | \gamma_{n,j} |
			\Big)
		\right\}
		\Bigg| \mathcal{F}_n
	\right) \\
	&\leq P \left(
	\bigcup_{j=1}^{p_n - q}
	\left\{
	\left\vert \zeta_{n,j} \right\vert
	> \dfrac{\lambda_n }{2 \sqrt{n}} 
	\Big(
	W_{0(2),j} - | \gamma_{n,j} |
	\Big)
	- \xi
	\right\}
	\Bigg| \mathcal{F}_n
	\right) 
	+ P \left(
	\left\Vert \bm{\nu}_n \right\Vert_2
	> \xi 
	\Big| \mathcal{F}_n
	\right).
	\end{align*}
	Again,
	$$
	\dfrac{\lambda_n}{n^{1.5-c_1}} 
	\bm{W}_{0(1)} \circ
	\text{sgn} \left[ \bm{\beta}_{0(1)} \right]
	= o_p(1),
	$$
	%\textcolor{blue}{[Tun, is this elementwise?]}
	so, by Lemmas \ref{lem_CtildeW} and \ref{lem_Znw1},
	$$
	P \left(
	\left\Vert \bm{\nu}_n \right\Vert_2
	> \xi 
	\Big| \mathcal{F}_n
	\right)
	= o(1) \quad a.s. \,\, P_D.
	$$
	Notice how the penalty weights $\bm{W}_{0(1)}$ and $\bm{W}_{0(2)}$ upend the strong irrepresentable condition (\ref{assume_StrongIrrepresent}). Specifically, 
	$$
	P \left(
		W_{0(2),j} - | \gamma_{n,j} |
		< 0
	\right)
	> 0,
	$$
	which then renders the probability bound to be unhelpful. Instead, notice that from the strong irrepresentable condition (\ref{assume_StrongIrrepresent}),
	$$
	\gamma_{n,j} \leq 
	(1 - \eta_*)
	\times
	\max_{1 \leq j \leq q} W_{0(1), j}
	$$ 
	for all $j= 1, \cdots, q$. We focus on the more restrictive case where 
	$$
	\eta_* = 1 \Longleftrightarrow \bm{\eta} = \bm{1}_{p_n-q},
	$$  
	which leads to a more meaningful probability bound. Then, $\gamma_{n,j} = 0$ for all $j= 1, \cdots, q$, and 
	\begin{align*}
	&P \left(
		\bigcup_{j=1}^{p_n - q}
		\left\{
			\left\vert \zeta_{n,j} \right\vert
			> \dfrac{\lambda_n }{2 \sqrt{n}} 
			W_{0(2),j}  - \xi
		\right\}
		\Bigg| \mathcal{F}_n
	\right) \\
	&\leq P \left(
		\bigcup_{j=1}^{p_n - q}
		\left\{
			\left\vert \zeta_{n,j} \right\vert 
			> \dfrac{\lambda_n }{2 \sqrt{n}} 
			\left( \min_{1 \leq j \leq p_n-q} W_{0(2),j} \right) 
			- \xi 
		\right\}
		\Bigg| \mathcal{F}_n
	\right) \\
	&\leq P \left(
		\Big\Vert \bm{\zeta}_n \Big\Vert_2
		> \dfrac{\lambda_n }{2 \sqrt{n}} 
		\left( \min_{1 \leq j \leq p_n-q} W_{0(2),j} \right) 
		- \xi 
		\Bigg| \mathcal{F}_n
	\right) \\
	&= P \left(
		\dfrac{1}{n^{c_2 - \frac{1}{2}}}
		\Big(
			\big\Vert \bm{\zeta}_n \big\Vert_2 
			+ \xi
		\Big)
		> \dfrac{\lambda_n }{2 n^{c_2}} 
		\left( \min_{1 \leq j \leq p_n-q} W_{0(2),j} \right) 
		\Bigg| \mathcal{F}_n
	\right)
	\end{align*}
	For the case of exponential random weights
	$$
	F_W(w) = 1 - e^{- \theta_w w}
	$$
	for some $\theta_w > 0$, we immediately have
	$$
	\left(
		\min_{1 \leq j \leq p_n - q} W_{0(2)j}
	\right)
	\sim
	{\rm Exp}
	\left(
		(p_n - q) \theta_w
	\right).
	$$ 
	Then, by part (b) of Lemma \ref{lem_Znw3}, 
	\begin{align*}
	&P \left(
	\dfrac{1}{n^{c_2 - \frac{1}{2}}}
	\Big(
	\big\Vert \bm{\zeta}_n \big\Vert_2 
	+ \xi
	\Big)
	> \dfrac{\lambda_n }{2 n^{c_2}} 
	\left( \min_{1 \leq j \leq p_n-q} W_{0(2),j} \right) 
	\Bigg| \mathcal{F}_n
	\right) \\
	&= P \left(
		W < \theta_w \dfrac{2 n^{c_2}}{\lambda_n}
		\dfrac{p_n-q}{n^{c_2 - \frac{1}{2}}}
		\Big( \big\Vert \bm{\zeta}_n \big\Vert_2 + \xi \Big)
	\Big| \mathcal{F}_n
	\right) 
	\,\, \text{where} \,\, W \sim {\rm Exp} (1)\\
	&= o(1) \quad a.s. \,\, P_D,
	\end{align*}  
	and we have just shown that 
	$$
	P \left[ \left( B_n^w \right)^c \big\vert \mathcal{F}_n \right]
	= o(1) \quad a.s. \,\, P_D
	$$
	for weighting scheme (\ref{eq.Bnw.WeightPenalty2}).\\
	
	Finally, 
	\begin{align*}
	& P\left(
	\bnw (\lambda_n) \stackrel{s}{=} \bm{\beta}_0
	\big\vert \mathcal{F}_n
	\right) \\
	&\geq 1 - \Big\{
	P \left[ 
	\left( A_n^w \right)^c
	\big\vert \mathcal{F}_n 
	\right]
	+ P \left[ 
	\left( B_n^w \right)^c
	\big\vert \mathcal{F}_n
	\right]
	\Big\} \\
	&= 1 - o(1) \quad a.s. \,\, P_D 	
	\end{align*}
	for all three weighting schemes (\ref{eq.NoWeightPenalty}), (\ref{eq.Bnw.WeightPenalty}) and (\ref{eq.Bnw.WeightPenalty2}). 
\end{proof}

\begin{proof}[Proof of Theorem \ref{thm_low_Consistency}]
	From the proof of Proposition \ref{lem_ModelSelect}, 
	\begin{alignat*}{2}
	&(\bnw - \bm{\beta}_0) &&\\
	&= \argmin_{\bm{u}}
	\Bigg\{
	&&\bm{u}' \left( \dfrac{X' D_n X}{n} \right) \bm{u}
	-2 \bm{u}' \left( \dfrac{X' D_n \bm{\epsilon}}{n} \right) 
	+ \dfrac{\bm{\epsilon}' D_n \bm{\epsilon}}{n} \\
	& &&+ \dfrac{\lambda_n}{n} 
	\sum_{j=1}^{p} W_{0,j} | \beta_{0,j} + u_{n,j} |
	\Bigg\} \\
	&:= \argmin_{\bm{u}} &&g_n(\bm{u}). 
	\end{alignat*}
	By Lemmas \ref{lem_X'DnX}, \ref{lem_EpsDnEps} and \ref{lem_X'DnEps}, for $\frac{\lambda_n}{n} \to \lambda_0 \in [0,\infty)$, Slutsky Theorem gives
	$$
	g_n(\bm{u}) \CONV{c.d.} g(\bm{u}) + \mu_W \sigma^2_{\epsilon}
	\quad a.s. \,\, P_D.
	$$
	Note that for weighting schemes (\ref{eq.Bnw.WeightPenalty}) and (\ref{eq.Bnw.WeightPenalty2}), $g(\bm{u})$ is a random function as it contains random weights. Since $g_n(\bm{u})$ is convex and $g(\bm{u})$ has a unique minimum, it follows from \citet{Geyer1996} that 
	$$
	\argmin_{ \bm{u} } g_n( \bm{u} )
	\CONV{c.d.}
	\argmin_{ \bm{u} } \left\{
		g( \bm{u} ) + \mu_W \sigma^2_{\epsilon}
	\right\} 
	= \argmin_{ \bm{u} } g( \bm{u} ) 
	\quad a.s. \,\, P_D.
	$$ 
	For weighting schemes (\ref{eq.NoWeightPenalty}), $g(\bm{u})$ is not a random function. Instead, we note that since $g_n(\bm{u})$ is convex, it follows from pointwise convergence of conditional probability that  
	$$
	\bnw - \bm{\beta}_0 = \mathcal{O}_p(1). 
	$$
	For any compact set $K$, by applying the Convexity Lemma \citep{Pollard1991},
	$$
	\underset{\bm{u} \in K}{\text{sup}} \; 
	\left\vert g_n (\bm{u}) - g (\bm{u}) - \mu_W \sigma^2_{\epsilon} \right\vert
	\CONV{c.p.} 0 \quad a.s. \,\, P_D. 	
	$$
	Therefore,
	$$
	\left( \bnw - \bm{\beta}_0 \right)   
	= 	\argmin_{\bm{u}} g_n(\bm{u})
	\CONV{c.p.} \argmin_{\bm{u}} g(\bm{u})  \quad a.s. \,\, P_D.
	$$ 
	Finally, for all three weighting schemes, if $\lambda_0 = 0$, $\argmin_{\bm{u}} g(\bm{u}) = \bm{0}$, i.e. 
	$$
	\bnw \CONV{c.p.} \bm{\beta}_0  \quad a.s. \,\, P_D.
	$$ 
\end{proof}

\begin{proof}[Proof of Theorem \ref{thm_low_AsympDistn}]
	Let $\bm{e}_n$ be the residual that corresponds to the strongly consistent estimator $\bSC$ of the linear regression model (\ref{eq.LinearModel}), and define 
	$$
	Q_n (\bm{z}) := 
	\left\| 
	D_n^{\frac{1}{2}} (\bm{y} - X \bm{z})
	\right\|_2^2 
	+ \lambda_n \sum_{j=1}^p W_{0,j} |z_j|,
	$$
	which leads to 
	\begin{align*}
	&Q_n \left( \bSC + \dfrac{1}{\sqrt{n}} \bm{u} \right) \\
	&= \left\| 
	D_n^{\frac{1}{2}} 
	\left[
	Y - X \left( \bSC + \dfrac{1}{\sqrt{n}} \bm{u} \right)				
	\right]	
	\right\|_2^2
	+ \lambda_n 
	\sum_{j=1}^p W_{0,j} 
	\left|
		\widehat{\beta}^{\text{SC}}_{n,j} 
		+ \dfrac{1}{\sqrt{n} } u_j
	\right| \\
	&= \left\| 
	D_n^{\frac{1}{2}} 
	\left( 
	\bm{e}_n - \dfrac{1}{\sqrt{n}} X \bm{u}  
	\right)				
	\right\|_2^2
	+ \lambda_n 
	\sum_{j=1}^p W_{0,j} 
	\left|
		\widehat{\beta}^{\text{SC}}_{n,j} 
		+ \dfrac{1}{\sqrt{n} } u_j
	\right| , 
	\end{align*}
	and 
	\begin{align*}
	Q_n \left( \bSC  \right)
	&= \left\| 
	D_n^{\frac{1}{2}} 
	\left( 
	Y - X \bSC 
	\right)				
	\right\|_2^2
	+ \lambda_n 
	\sum_{j=1}^p W_{0,j} 
	\left|
		\widehat{\beta}^{\text{SC}}_{n,j} 
	\right|\\
	&= \left\| 
	D_n^{\frac{1}{2}} \bm{e}_n		
	\right\|_2^2
	+ \lambda_n
	\sum_{j=1}^p W_{0,j}  
	\left|
		\widehat{\beta}^{\text{SC}}_{n,j} 
	\right|.
	\end{align*}
	Now, define
	$$
	V_n( \bm{u} ) := 
	Q_n \left( \bSC  + \dfrac{1}{\sqrt{n}} \bm{u} \right) 
	- Q_n \left( \bSC  \right), 
	$$
	and note that
	$$
	\argmin_{ \bm{u} } V_n( \bm{u} )
	= \argmin_{ \bm{u} }  Q_n \left( \bSC  + \dfrac{1}{\sqrt{n}} \bm{u} \right)
	= \sqrt{n} \left( \bnw - \bSC \right).
	$$
	Notice that $V_n( \bm{u} )$ can be simplified into 
	\begin{align*}
	&\bm{u}' \left( \dfrac{X' D_n X}{n} \right) \bm{u} 
	-2 \bm{u}' \left( \dfrac{X' D_n \bm{e}_n}{\sqrt{n}} \right)\\
	&+ \dfrac{\lambda_n}{ \sqrt{n} } 
	\sum_{j=1}^p W_{0,j}
	\left(
		\left|
			\sqrt{n} \widehat{\beta}^{\text{SC}}_{n,j} + u_j
		\right|
		- \left|
			\sqrt{n} \widehat{\beta}^{\text{SC}}_{n,j}
		\right|
	\right),
	\end{align*}
	where its penalty term can be expanded into 
	\begin{alignat*}{2}
	&\quad \, \dfrac{\lambda_n}{ \sqrt{n} } 
	\sum_{j=1}^p W_{0,j}
	&&\left(
		\left|
			\sqrt{n} \widehat{\beta}^{\text{SC}}_{n,j} + u_j
		\right|
		- \left|
			\sqrt{n} \widehat{\beta}^{\text{SC}}_{n,j}
		\right|
	\right) \\
	&= \dfrac{\lambda_n}{\sqrt{n}} \sum_{j=1}^p W_{0,j}
	&&\Big\{
		\left| \sqrt{n}
			\left[
				\beta_{0,j}
				+ \left(
					\widehat{\beta}_{n,j}^{\text{SC}} - \beta_{0,j} 
				\right)
			\right]  
			+ \mu_j
		\right| \\
	& &&- \left| \sqrt{n}
		\left[
			\beta_{0,j}
			+ \left(
				\widehat{\beta}_{n,j}^{\text{SC}} - \beta_{0,j} 
			\right)
		\right]  
	\right|
	\Big\} \\
	&:= \dfrac{\lambda_n}{\sqrt{n}}
	\sum_{j=1}^p W_{0,j} &&p_n(u_j).
	\end{alignat*}
	For $\beta_{0,j} \neq 0$,
	$$
	\left( \widehat{\beta}_{n,j}^{\text{SC}} - \beta_{0,j} \right) 
	\to  0 \quad a.s. \,\, P_D,
	$$
	and hence $\sqrt{n} \beta_{0,j}$ dominates $u_j$ for large $n$. Thus, it is easy to verify that $p_n(u_j)$ converges to $u_j \text{sgn} \left( \beta_{0,j} \right)$ for all $j \in \{j: \beta_{0,j} \neq 0 \}$. Thus, by Lemmas \ref{lem_X'DnX} and \ref{lem_X'DnResid}, if $q=p$, Slutsky Theorem ensures that
	$$
	V_n(\bm{u}) \CONV{c.d.} V(\bm{u})
	:= \mu_W \bm{u}' C \bm{u} - 2 \bm{u}' \Psi
	+ \lambda_0 \sum_{j=1}^p W_j
	\left[
		u_j \, \text{sgn}(\beta_{0,j}) 
	\right]
	\quad a.s. \,\, P_D, 
	$$
	where $\Psi$ has a $N \left( \bm{0}, \sigma^2_W \sigma^2_{\epsilon} C \right)$ distribution, and
	\begin{itemize}
		\item [(i)] $W_j$ = 1 for all $j$ under weighting scheme (\ref{eq.NoWeightPenalty}),
		\item [(ii)] $W_j = W_0$ for all $j$, $W_0 \sim F_W$ and $W_0 \perp \Psi$ under weighting scheme (\ref{eq.Bnw.WeightPenalty}), 
		\item [(iii)] $W_j \stackrel{iid}{\sim} F_W$ and $W_j \perp \Psi$ for all $j$ under weighting scheme (\ref{eq.Bnw.WeightPenalty2}).
	\end{itemize}
	Since $V_n(\bm{u})$ is convex and $V(\bm{u})$ has a unique minimum, it follows from \citet{Geyer1996} that
	$$
	\sqrt{n} \left( \bnw - \bSC \right) 
	= \argmin_{ \bm{u} } V_n( \bm{u} ) 
	\CONV{c.d.}
	\argmin_{ \bm{u} } V( \bm{u} ) \quad a.s. \,\, P_D
	$$
	when $q = p$. In particular, if $\lambda_0 = 0$, 
	$$
	\argmin_{ \bm{u} } V( \bm{u} )
	= \dfrac{1}{\mu_W} C^{-1} \Psi \sim N 
	\left( \bm{0}, \dfrac{ \sigma^2_W \sigma^2_{\epsilon} }{\mu_W^2} C^{-1} \right) . 
	$$
	However, if $0 < q < p$, then for $j \in \{j: \beta_{0,j} = 0 \}$, $p_n(u_j)$ is back to   
	$$
	\left| 
	\sqrt{n} \widehat{\beta}_{n,j}^{\text{SC}} 
	+ \mu_j
	\right|
	- \left|
	\sqrt{n} \widehat{\beta}_{n,j}^{\text{SC}}
	\right|,
	$$ 
	which depends on the sample path of realized data. This necessitates the Skorokhod argument, thus leading to the penalty term in (\ref{eq.SkorokhodResult}). 
\end{proof}

We need the following lemma to prove Theorem \ref{thm_cond_oracle}:

\begin{lem} \label{lem_asPD_LASLS}
	Consider \citet{Liu&Yu}'s unweighted two-step LASSO+LS estimator $\widehat{\bm{\beta}}_n^{LAS+LS}$, with its corresponding set of selected variables denoted as $\widehat{S}_n$. Adopt assumptions (\ref{assume_boundedX}), (\ref{assume_X'X_11}) and (\ref{assume_StrongIrrepresent}). If there exists $\frac{1}{2} < c_1 <  c_2 < 1$ and $0 \leq c_3 < 2(c_2 - c_1)$ for which $\lambda_n = \mathcal{O} \left( n^{c_2} \right)$ and $p_n = \mathcal{O} \left( n^{c_3} \right)$, then as $n \to \infty$,
	$$
	P\left(
		\widehat{S}_n = S_0	
		\Big| \mathcal{F}_n 
	\right)	
	\to 1
	\quad a.s. \,\, P_D. 
	$$ 
\end{lem}

\begin{proof}
	The first step (i.e. the variable selection step) of obtaining $\widehat{\bm{\beta}}_n^{LAS+LS}$ is effectively the standard LASSO procedure. Thus, by assumption (\ref{assume_StrongIrrepresent}), from the proof of Proposition 1 of \citet{BinYu}, we obtain
	$$
	\left\{
	    \widehat{S}_n = S_0	
	\right\}
	\supseteq
	\left\{
	    A_n \cap B_n
	\right\}
	$$
	and thus
	$$
	P \left(
		\widehat{S}_n = S_0	
		\Big| \mathcal{F}_n 
	\right)	
	\geq 
	P \left(
		A_n \cap B_n
		\big| \mathcal{F}_n 
	\right), 	%\quad a.s. \,\, P_D,
	$$ 
	where
	\begin{align*}
	A_n &\equiv  
	\left\{
		\left\vert 
			C_{n(11)}^{-1}
			\dfrac{X_{(1)}' \bm{\epsilon}}{\sqrt{n}} 
		\right\vert
		\leq  \sqrt{n}
		\left(
			\left\vert  \bm{\beta}_{0(1)} \right\vert
			- \dfrac{\lambda_n}{2n}
			\left\vert
				C_{n(11)}^{-1}
				\text{sgn} \left( \bm{\beta}_{0(1)} \right)
			\right\vert
		\right)
		\,\, \text{{\rm element-wise}}
	\right\} \\
	B_n &\equiv 
	\left\{
		\left\vert 
			\dfrac{1}{\sqrt{n}} 
			\left[
				C_{n(21)} C_{n(11)}^{-1} X_{(1)}' - X_{(2)}'
			\right]
			\bm{\epsilon}
		\right\vert
		\leq  \dfrac{\lambda_n}{2 \sqrt{n}} \bm{\eta} 
		\,\, \text{{\rm element-wise}}
	\right\}.
	\end{align*}
	%\textcolor{red}{Note that the lower bound of the conditional probability $P \left( \widehat{S}_n = S_0 \Big| \mathcal{F}_n \right)$ holds $a.s.$-$P_D$ because the Set arguments deployed in the proof of Proposition 1 of \citet{BinYu} are valid for almost every sample path $\{y_1, y_2, \cdots \}$.} 
	Next, we want to show that
	$$
	P \left(
		A_n^c \big| \mathcal{F}_n
	\right)
	\to 0 \,\,\, a.s. \,\, P_D
	\quad \text{ and } \quad
	P \left(
		B_n^c \big| \mathcal{F}_n
	\right)
	\to 0 \,\,\, a.s. \,\, P_D
	$$
	such that
	$$
	P \left(
	\widehat{S}_n = S_0
	\Big| \mathcal{F}_n 
	\right)	
	\geq 1 - 
	\left[
		P \left(
		A_n^c \big| \mathcal{F}_n
		\right) 
		+ P \left(
		B_n^c \big| \mathcal{F}_n
		\right)
	\right]
	\to 1 \quad a.s. \,\, P_D.
	$$
	First, by assumptions (\ref{assume_boundedX}) and (\ref{assume_X'X_11}), $C_{n(11)}^{-1} = \mathcal{O}(1)$ for all $n$, whereas
	$$
	\dfrac{\lambda_n}{2n}
	C_{n(11)}^{-1}
	\text{sgn} \left( \bm{\beta}_{0(1)} \right)
	\to \bm{0}.
	$$
	By Lemma \ref{lem_ASconv}, for any $\frac{1}{2} < c' < 1$,
	$$
	\dfrac{1}{n^{c'}} X_{(1)}' \bm{\epsilon} 
	\to \bm{0} \quad a.s. \,\, P_D
	\quad \Longrightarrow \quad
	\dfrac{1}{n^{c' - \frac{1}{2} }}
	\left(
		C_{n(11)}^{-1}
		\dfrac{X_{(1)}' \bm{\epsilon}}{\sqrt{n}} 
	\right)
	\to \bm{0} \quad a.s. \,\, P_D.
	$$
	For ease of notation, let 
	$$
	\bm{z} = \left[z_1, \cdots, z_q \right]' := 
	C_{n(11)}^{-1} \dfrac{X_{(1)}' \bm{\epsilon}}{\sqrt{n}}. 
	$$
	Then, for any $\frac{1}{2} < c' < 1$,
	\begin{align*}
	P \left(
		A_n^c \big| \mathcal{F}_n
	\right) 
	&\leq \sum_{j=1}^q 
	P \left(
		|z_j| > \sqrt{n} 
		\left[
			|\beta_{0,j}| + o(1)
		\right] 
		\Big| \mathcal{F}_n
	\right) \\
	&= \sum_{j=1}^q 
	P \left(
		\dfrac{|z_j| }{n^{c' - \frac{1}{2} }}
		> n^{1-c'}
		\Big[
			|\beta_{0,j}| + o(1)
		\Big] 
		\Big| \mathcal{F}_n
	\right) \\
	&\to 0 \quad a.s. \,\, P_D.
	\end{align*}
	Next, using the same notations that we introduced in the proofs of Lemma \ref{lem_Znw3} and Theorem \ref{thm_Model_Select}, let 
	$$
	H = X_{(1)} C_{n(11)}^{-1} C_{n(12)} - X_{(2)},
	$$ 
	and let
	$$
	\eta_* = \min_{1 \leq j \leq p_n - q} \bm{\eta},
	$$
	where assumption (\ref{assume_StrongIrrepresent}) ensures that $0 < \eta_* \leq 1$. Again, due to assumptions (\ref{assume_boundedX}) and (\ref{assume_X'X_11}) and that $q$ is fixed, every element in the matrix $H$ is bounded. Let $h_{ij}$ be the $(i,j)^{th}$ element of $H$. Again, by Lemma \ref{lem_ASconv}, for all $j = 1, \cdots, p_n - q$, 
	$$
	\dfrac{1}{n^{c_1}}
	\sumin h_{ji} \epsilon_i
	\to 0 \quad a.s. \,\, P_D
	$$
	for $\frac{1}{2} < c_1 < 1$. Consequently, we have
	\begin{align*}
	P \left(
	B_n^c \big| \mathcal{F}_n
	\right) 
	&= P \left(
		\bigcup_{j=1}^{p_n-q} 
		\left\{
			\left\vert
				\dfrac{1}{\sqrt{n}} \sumin h_{ji} \epsilon_i
			\right\vert
			> \dfrac{\lambda_n}{2 \sqrt{n}} \eta_j
		\right\}
		\Bigg| \mathcal{F}_n
	\right) \\
	&\leq P 
	\left(
		\max_{1 \leq j \leq p_n-q}
		\left\vert
			\dfrac{1}{\sqrt{n}} \sumin h_{ji} \epsilon_i
		\right\vert
		> \dfrac{\lambda_n}{2 \sqrt{n}} \eta_*
		\Bigg| \mathcal{F}_n
	\right) \\
	&\leq P 
	\left(
		\left\Vert
			\dfrac{1}{\sqrt{n}} H' \bm{\epsilon}
		\right\Vert_2
		> \frac{\lambda_n}{2 \sqrt{n}} \eta_*
		\Bigg| \mathcal{F}_n
	\right) \\
	&= P 
	\left(
		\dfrac{1}{n^{c_2 - \frac{1}{2} }}
		\left\Vert
			\dfrac{1}{\sqrt{n}} H' \bm{\epsilon}
		\right\Vert_2
		> \frac{\lambda_n}{2 n^{c_2}} \eta_*
		\Bigg| \mathcal{F}_n
	\right), 
	\end{align*} 
	where
	\begin{align*}
	\left( 
		\dfrac{1}{n^{c_2 - \frac{1}{2} }}
		\left\Vert
		\dfrac{1}{\sqrt{n}} H' \bm{\epsilon}
		\right\Vert_2
	\right)^2 
	&= \dfrac{1}{ n^{2 c_2 - 1} } 
	\sum_{j=1}^{p_n-q}
	\left(
		\dfrac{1}{\sqrt{n}}
		\sumin h_{ji} \epsilon_i
	\right)^2 \\
	&= \dfrac{ n^{2 c_1 - 1} }{ n^{2 c_2 - 1} }
	\sum_{j=1}^{p_n-q}
	\left(
		\dfrac{1}{n^{c_1}}
		\sumin h_{ji} \epsilon_i
	\right)^2 \\
	&= \mathcal{O} \left(
		\dfrac{1}{ n^{2(c_2 - c_1)} }
	\right) 
	\times
	o \left( n^{c_3} \right) \,\,\, a.s. \,\, P_D \\
	&= o(1) \,\,\, a.s. \,\, P_D
	\end{align*}
	because $c_3 < 2(c_2 - c_1)$ and $\frac{1}{2} < c_1 < c_2 < 1$, whereas 
	$$
	\frac{\lambda_n}{2 n^{c_2}} \eta_*
	= \mathcal{O}(1).
	$$
	Hence $P \left( B_n^c \big| \mathcal{F}_n \right) \to 0$ almost surely under $P_D$ and the result follows.
\end{proof}

Note that the constraints on $c_1$, $c_2$ and $c_3$ in Lemma \ref{lem_asPD_LASLS} cover the more restrictive constraints found in Theorem \ref{thm_Model_Select}. Therefore, the result in Lemma \ref{lem_asPD_LASLS} still holds under the assumptions of Theorem \ref{thm_Model_Select}.

A slightly different layout of the proof for Lemma \ref{lem_asPD_LASLS} would be as follows: using the results in Proposition 1 of \citet{BinYu}, on the probability space $P_D$,
$$
P_D \left(
	\widehat{S}_n = S_0	
\right)	
\geq 
P_D \left(
A_n \cap B_n
\right). 
$$ 
Using the same techniques in the preceding proof, we show that
$$
\lim_{n \to \infty} A_n^c = \emptyset \quad a.s. \,\, P_D
\,\, \Longrightarrow \,\, 
P_D \left(  \lim_{n \to \infty} A_n^c \right) = 0
\,\, \Longrightarrow \,\, 
P_D \left(  A_n^c \,\,\, i.o. \right) = 0,
$$
and
$$
\lim_{n \to \infty} B_n^c = \emptyset \quad a.s. \,\, P_D
\,\, \Longrightarrow \,\, 
P_D \left(  \lim_{n \to \infty} B_n^c \right) = 0
\,\, \Longrightarrow \,\, 
P_D \left(  B_n^c \,\,\, i.o. \right) = 0,
$$
where $i.o.$ stands for ``infinitely often". Then,
\begin{align*}
&P_D \left(
	\left( A_n \cap B_n  \right)^c
	\,\, i.o.
\right)
\leq P_D \left( A_n^c \,\,\, i.o. \right) +
P_D \left( B_n^c \,\,\, i.o. \right) =0 \\
\Longrightarrow \,\,
&P_D \left(
\{ A_n \cap B_n  \}
\,\,\, i.o.
\right) = 1 \\
\Longrightarrow \,\,
&P_D \left(
	\left\{ \widehat{S}_n = S_0 \right\}
	\,\,\, i.o. 
\right)	
\geq 
P_D \left(
	\{ A_n \cap B_n \} 
	\,\,\, i.o. 
\right) = 1 \\
\Longrightarrow \,\,
&P_D \left(
	\lim_{n \to \infty} 
	\widehat{S}_n = S_0
\right) = 1, 
\end{align*}
and thus, on the probability space $P = P_D \times P_W$,
$$
P\left(
\widehat{S}_n = S_0	
\Big| \mathcal{F}_n 
\right)	
\to 1
\quad a.s. \,\, P_D.
$$
We have
$$
\lim_{n \to \infty} A_n^c = \emptyset \quad a.s. \,\, P_D
$$
because for any $\frac{1}{2} < c' < 1$,
$$
\dfrac{1}{n^{c' - \frac{1}{2} }}
\left(
C_{n(11)}^{-1}
\dfrac{X_{(1)}' \bm{\epsilon}}{\sqrt{n}} 
\right)
\to \bm{0} \quad a.s. \,\, P_D
$$
whereas
$$
n^{1 - c'}
\left(
\left\vert  \bm{\beta}_{0(1)} \right\vert
- \dfrac{\lambda_n}{2n}
\left\vert
C_{n(11)}^{-1}
\text{sgn} \left( \bm{\beta}_{0(1)} \right)
\right\vert
\right)
= \mathcal{O} \left( n^{1 - c'} \right).
$$ 
Meanwhile, we establish
$$
\lim_{n \to \infty} B_n^c = \emptyset \quad a.s. \,\, P_D
$$
because 
$$
B_n^c \subseteq 
\left\{
\dfrac{1}{n^{c_2 - \frac{1}{2} }}
\left\Vert
\dfrac{1}{\sqrt{n}} H' \bm{\epsilon}
\right\Vert_2
> \frac{\lambda_n}{2 n^{c_2}} \eta_*	
\right\},
$$
where 
$$
\dfrac{1}{n^{c_2 - \frac{1}{2} }}
\left\Vert
\dfrac{1}{\sqrt{n}} H' \bm{\epsilon}
\right\Vert_2
= o(1) \,\,\, a.s. \,\, P_D 
\quad \text{ but }
\quad 
\frac{\lambda_n}{2 n^{c_2}} \eta_*
= \mathcal{O}(1).
$$
\\
The following version of Sherman–Morrison–Woodbury matrix-inversion identity (e.g., Equation (26) of \citet{MatrixIdentity}) will come in handy later: For any square matrices $A$ and $B$ of conformal sizes where $A$ is invertible, we have
	\begin{align} \label{eq:MatrixIdentity}
	(A+B)^{-1} = A^{-1} - A^{-1} B A^{-1} 
	\left(
		I + B A^{-1}
		\right)^{-1}.
	\end{align}

\begin{proof}[Proof of Theorem \ref{thm_cond_oracle}]
	Since the first-step is in fact equivalent to the one-step procedure, Theorem \ref{thm_Model_Select} immediately gives us 
	$$
	P \left(
		\widehat{S}_n^w = S_0	
	\big| \mathcal{F}_n
	\right) 
	\geq 
	P\left(
		\bnw \stackrel{s}{=} \bm{\beta}_0
	\big| \mathcal{F}_n 
	\right)	
	\to 1
	\quad a.s. \,\, P_D,
	$$ 
	while Lemma \ref{lem_asPD_LASLS} immediately gives us 
	$$
	P \left(
	\widehat{S}_n = S_0	
	\big| \mathcal{F}_n
	\right) 
	\to 1
	\quad a.s. \,\, P_D.
	$$
	Conditional on $\left\{ \widehat{S}_n^w = S_0 \right\}$ and $\left\{ \widehat{S}_n = S_0 \right\}$, since $Y = X_{(1)} \bm{\beta}_{0(1)} + \bm{\epsilon}$,
	\begin{align*}
	&\bnwa - \widehat{\bm{\beta}}_{n(1)}^{LAS+LS} \\
	&= \left(
		X_{(1)}' D_n X_{(1)} 
	\right)^{-1} X_{(1)}' D_n Y
	- \left(
	X_{(1)}' X_{(1)} 
	\right)^{-1} X_{(1)}' Y \\
	&=  \left(
		X_{(1)}' D_n X_{(1)} 
	\right)^{-1} X_{(1)}' D_n \bm{\epsilon}
	- \left(
		X_{(1)}' X_{(1)} 
	\right)^{-1} X_{(1)}' \bm{\epsilon} \\
	&= \left(
		\cnwa
	\right)^{-1}  
	\dfrac{X_{(1)}' (D_n - I_n) \bm{\epsilon}}{n}
	- \left[
		C_{n(11)}^{-1} - \left( \cnwa \right)^{-1}  
	\right] 
	\dfrac{X_{(1)}' \bm{\epsilon}}{n},
	\end{align*}  
	which leads to 
	\begin{align*}
	&\sqrt{n}
	\left(
		\bnwa - \widehat{\bm{\beta}}_{n(1)}^{LAS+LS}
	\right) \\
	&= \left(
	\cnwa
	\right)^{-1}  
	\dfrac{X_{(1)}' (D_n - I_n) \bm{\epsilon}}{\sqrt{n}}
	- \left[
	C_{n(11)}^{-1} - \left( \cnwa \right)^{-1}  
	\right] 
	\dfrac{X_{(1)}' \bm{\epsilon}}{\sqrt{n}}.
	\end{align*}
	Based on the (alternative) proof of Lemma \ref{lem_Cnw11inv}, we have seen that
	$$
	\left( \cnwa \right)^{-1}
	\CONV{a.s.} C_{11}^{-1},
	$$	
	and from the (alternative) proof of Lemma \ref{lem_Znw1}, we could deploy Slutsky's Theorem to obtain
	$$
	\left(
	\cnwa
	\right)^{-1}  
	\dfrac{X_{(1)}' (D_n - I_n) \bm{\epsilon}}{\sqrt{n}}
	\CONV{c.d.}
	N_q \left(
		\bm{0} \,\, , \,\, 
		\sigma^2_W \sigma^2_\epsilon C_{11}^{-1}
	\right)
	\quad a.s. \,\, P_D.
	$$ 
	Meanwhile, we deploy the matrix inversion identity (\ref{eq:MatrixIdentity}) by taking $A = C_{n(11)}$ and
	$$
	B = \dfrac{1}{n} X'_{(1)} (D_n - I_n) X_{(1)}
	$$
	to obtain
	\begin{align*}
	\left( \cnwa \right)^{-1} 
	&= \left[
		C_{n(11)} + \dfrac{1}{n} X'_{(1)} (D_n -  I_n) X_{(1)}
	\right]^{-1} \\
	&= A^{-1} - A^{-1} B A^{-1} 
	\left(
		I_q + B A^{-1}
	\right)^{-1}.
	\end{align*} 
	Then,
	\begin{small}
	\begin{align*}
	&\left[
	C_{n(11)}^{-1} - \left( \cnwa \right)^{-1}
	\right]  \dfrac{X_{(1)}' \bm{\epsilon}}{\sqrt{n}} \\
	&= C_{n(11)}^{-1} 
	\left[ \dfrac{X'_{(1)} (D_n - I_n) X_{(1)}}{n} \right]
	C_{n(11)}^{-1} 
	\left[
		I_q + 
		\left( \dfrac{X'_{(1)} (D_n - I_n) X_{(1)}}{n}  \right)
		C_{n(11)}^{-1} 
	\right]^{-1} \dfrac{X_{(1)}' \bm{\epsilon}}{\sqrt{n}} \\
	&= C_{n(11)}^{-1} 
	\left[ \dfrac{X'_{(1)} (D_n - I_n) X_{(1)}}{n^{1-c}} \right]
	C_{n(11)}^{-1} 
	\left[
		I_q + 
		\left( \dfrac{X'_{(1)} (D_n - I_n) X_{(1)}}{n}  \right)
		C_{n(11)}^{-1} 
	\right]^{-1} \dfrac{X_{(1)}' \bm{\epsilon}}{n^{\frac{1}{2}+c}},
	\end{align*}
	\end{small}
	where Lemma \ref{lem_ASconv} and assumption (\ref{assume_boundedX}) ensure that for any $0 < c < \frac{1}{2}$,
	$$
	\dfrac{1}{n^{1-c}} 
	X'_{(1)} (D_n - I_n) X_{(1)}
	\CONV{a.s.} \bm{0} 
	$$
	and
	$$
	\dfrac{X_{(1)}' \bm{\epsilon}}{n^{\frac{1}{2}+c}}
	\to \bm{0} \quad a.s. \,\, P_D.
	$$
	Since $C_{n(11)}$ is invertible for all $n$, we have
	$$
	C_{n(11)}^{-1} \to C_{11}^{-1},
	$$
	and 
	\begin{align*}
	\left[
		I_q + 
		\left( \dfrac{X'_{(1)} (D_n - I_n) X_{(1)}}{n}  \right)
		C_{n(11)}^{-1} 
	\right]^{-1} 
	&= C_{n(11)} \left(\cnwa \right)^{-1} \\
	&\CONV{a.s.} C_{11} C_{11}^{-1} \\
	&= I_q.
	\end{align*}
	Hence,
	$$
	\left[
	C_{n(11)}^{-1} - \left( \cnwa \right)^{-1}
	\right]  \dfrac{X_{(1)}' \bm{\epsilon}}{\sqrt{n}}
	\CONV{c.p.} \bm{0} \quad a.s. \,\, P_D.
	$$
	Consequently, conditional on $\left\{ \widehat{S}_n^w = S_0 \right\}$ and $\left\{ \widehat{S}_n = S_0 \right\}$, Slutsky's Theorem ensures that 
	$$
	\sqrt{n}
	\left(
	\bnwa - \widehat{\bm{\beta}}_{n(1)}^{LAS+LS}
	\right) 
	\CONV{c.d.}
	N_q \left( \bm{0} \,\, , \,\, \sigma^2_W \sigma^2_\epsilon C_{11}^{-1} \right)
	\quad a.s. \,\, P_D.
	$$
	Finally, for any $t \in \mathbb{R}$,
	\begin{alignat*}{2}
	& &&P \left(
		\sqrt{n} 
		\left(
			\bnwa - \widehat{\bm{\beta}}_{n(1)}^{LAS+LS}
		\right)
		\leq t
		\Big| \mathcal{F}_n
	\right) \\
	&=  &&P \left(
		\sqrt{n} 
		\left(
			\bnwa - \widehat{\bm{\beta}}_{n(1)}^{LAS+LS}
		\right)
		\leq t \, , \, \left\{ \widehat{S}_n^w = S_0, \widehat{S}_n = S_0 \right\}
		\Big| \mathcal{F}_n
		\right) \\
	& &&+
		P \left(
			\sqrt{n} 
			\left(
				\bnwa - \widehat{\bm{\beta}}_{n(1)}^{LAS+LS}
			\right)
			\leq t \, , \, \left\{ \widehat{S}_n^w = S_0, \widehat{S}_n = S_0 \right\}^c
			\Big| \mathcal{F}_n
		\right) \\
	&\leq &&P \left(
		\sqrt{n} 
		\left(
			\bnwa - \widehat{\bm{\beta}}_{n(1)}^{LAS+LS}
		\right)
		\leq t \, , \, \left\{ \widehat{S}_n^w = S_0, \widehat{S}_n = S_0 \right\}
		\Big| \mathcal{F}_n
	\right) \\
	& &&+
	P \left(
		 \left\{\widehat{S}_n^w \neq S_0 \right\}
		 \bigcup
		  \left\{\widehat{S}_n \neq S_0 \right\}
		\Big| \mathcal{F}_n
	\right) \\
	&\leq &&P \left(
		\sqrt{n} 
		\left(
			\bnwa - \widehat{\bm{\beta}}_{n(1)}^{LAS+LS}
		\right)
		\leq t \, , \, \left\{ \widehat{S}_n^w = S_0, \widehat{S}_n = S_0 \right\}
		\Big| \mathcal{F}_n
	\right) \\
	& &&+
	P \left(
		\widehat{S}_n^w \neq S_0 
		\Big| \mathcal{F}_n
	\right) 
	+ P \left(
		\widehat{S}_n \neq S_0 
	\Big| \mathcal{F}_n
	\right) 
	\end{alignat*}
where 
	$$
	P \left( 
		\widehat{S}_n^w \neq S_0 
		\Big| \mathcal{F}_n 
	\right) 
	\to 0  \,\,\, a.s. \,\, P_D
	\quad \text{ and } \quad
	P \left(
		 \widehat{S}_n \neq S_0 
		 \Big| \mathcal{F}_n 
	\right) 
	\to 0 \,\,\, a.s. \,\, P_D,
	$$
and 
	$$
	P \left(
		\sqrt{n} 
		\left(
			\bnwa - \widehat{\bm{\beta}}_{n(1)}^{LAS+LS}
		\right)
		\leq t \, , \, \left\{ \widehat{S}_n^w = S_0, \widehat{S}_n = S_0 \right\}
		\Big| \mathcal{F}_n
	\right) 
	\to P(Z \leq t) 
	$$
almost surely under $P_D$ for $Z \sim N_q \left( \bm{0} \,\, , \,\, \sigma^2_W \sigma^2_\epsilon C_{11}^{-1} \right)$. 
\end{proof}

\begin{proof}[Proof of Theorem \ref{thm_high_Consistency}]
 	Since $Y = X_{(1)} \bm{\beta}_{0(1)} + \bm{\epsilon}$, by conditioning on $\left\{ \widehat{S}_n^w = S_0 \right\}$, we have $\widehat{\bm{\beta}}^w_{n(2)} = \bm{\beta}_{0(2)} = \bm{0}$, and
	\begin{align*}
	\bnwa - \bm{\beta}_{0(1)} 
	&= \left(
	X_{(1)}' D_n X_{(1)} 
	\right)^{-1} X_{(1)}' D_n Y
	- \bm{\beta}_{0(1)} \\
	&= \left(
	X_{(1)}' D_n X_{(1)} 
	\right)^{-1} X_{(1)}' 
	D_n \bm{\epsilon} \\
	&= \left(
		\cnwa
	\right)^{-1}
	\dfrac{X_{(1)}' D_n \bm{\epsilon}}{n} \\
	&\CONV{c.p.} \bm{0} \quad a.s. \,\, P_D
	\end{align*}
by Lemmas \ref{lem_X'DnX} and \ref{lem_Znw1}. Finally, for any $\xi > 0$,
	\begin{align*}
	&P \left(
		\left\Vert
			\bnw - \bm{\beta}_0
		\right\Vert_2
		> \xi
	\Big| \mathcal{F}_n
	\right) \\
	&= P \left(
	\left\Vert
	\bnw - \bm{\beta}_0
	\right\Vert_2
	> \xi \, , \, 
	\widehat{S}_n^w = S_0
	\Big| \mathcal{F}_n
	\right)
	+ P \left(
		\left\Vert
			\bnw - \bm{\beta}_0
		\right\Vert_2
		> \xi \, , \, 
		\widehat{S}_n^w \neq S_0
		\Big| \mathcal{F}_n
	\right) \\
	&\leq P \left(
	\left\Vert
	\bnw - \bm{\beta}_0
	\right\Vert_2
	> \xi \, , \, 
	\widehat{S}_n^w = S_0
	\Big| \mathcal{F}_n
	\right)
	+ P \left(
		\widehat{S}_n^w \neq S_0
	\big| \mathcal{F}_n
	\right) \\
	&\to 0 \quad a.s. \,\, P_D.  
	\end{align*}
\end{proof}

\begin{remark} \label{rmk_centering}
	Consider Theorem \ref{thm_low_AsympDistn} with centering on $\bm{\beta}_0$ 
	$$
	\sqrt{n} \left(
	\bnw - \bm{\beta}_0
	\right).
	$$     
	Using the same technique in the proof of Theorem \ref{thm_low_AsympDistn}, we work with
	$$
	V_n(\bm{u}) := 
	Q_n \left(
		\bm{\beta}_0 
		+ \dfrac{1}{\sqrt{n}} \bm{u}
	\right)
	- Q_n \left(
		\bm{\beta}_0 
	\right)
	$$
	which can be simplified into
	\begin{align*}
	\bm{u}' \left( \dfrac{X' D_n X}{n} \right) \bm{u} 
	-2 \bm{u}' \left( \dfrac{X' D_n \bm{\epsilon}}{\sqrt{n}} \right)
	+ \dfrac{\lambda_n}{ \sqrt{n} } 
	\sum_{j=1}^p W_{0,j}
	\left(
	\left|
	\sqrt{n} \beta_{0,j} + u_j
	\right|
	- \left|
	\sqrt{n} \beta_{0,j}
	\right|
	\right).
	\end{align*}
	Again, assumption \ref{assume_X'X} ensures convergence of the first term, whereas argument for the penalty term in the proof of Theorem \ref{thm_low_AsympDistn} still applies to the third term. However, the second term has
	$$
	\dfrac{X' D_n \bm{\epsilon}}{\sqrt{n}} 
	= \dfrac{1}{\sqrt{n}} X' 
	\left(
		D_n - \mu_W I_n
	\right)
	\bm{\epsilon}
	+ \dfrac{1}{\sqrt{n}} X' \bm{\epsilon},
	$$ 
	where 
	$$
	\dfrac{1}{\sqrt{n}} X' 
	\left(
	D_n - \mu_W I_n
	\right) \bm{\epsilon} = \mathcal{O}_p(1)
	\quad a.s. \,\, P_D,
	$$
	but $(X' \bm{\epsilon})/(\sqrt{n})$ is asymptotically normal under $P_D$ \citep{Knight&Fu}. Thus, conditional on $\mathcal{F}_n$, $(X' D_n \bm{\epsilon})/(\sqrt{n})$ depends on the sample path of realized data $\{y_1, y_2, \cdots \}$, thus causing $\sqrt{n} \left( \bnw - \bm{\beta}_0 \right)$ to be unable to achieve convergence in conditional distribution almost surely under $P_D$. 
\end{remark}

\section*{Acknowledgements}
TLN and MAN were  supported in part by the University of Wisconsin Institute for the Foundations of Data Science through a grant from the US National Science Foundation (NSF1740707). The authors thank the associate editor and an anonymous referee for their valuable feedback and suggestions that lead to a substantially improved manuscript. Insights from Nick Polson and Steve Wright have also served as helpful guideposts in this effort. 

%\textcolor{blue}{[Tun, this is a very impressive manuscript...three cheers!]}
\bibliographystyle{imsart-nameyear}	
\bibliography{WLB_Ref}         

\end{document}